\documentclass[a4paper,11pt]{article}
\usepackage{makeidx}
\usepackage[english]{babel}
\usepackage[latin1]{inputenc}
\usepackage{amsfonts}
\usepackage{bm}
\usepackage{amssymb}
\usepackage{latexsym}
\usepackage{amsmath}
\usepackage{amscd}
\usepackage{amsthm}
\usepackage{hyperref}
\usepackage{rotating}
\usepackage{graphicx}
\usepackage{pifont}
\usepackage{curves}
\usepackage{fancyhdr}
\usepackage{epsfig}
\usepackage{pstricks}
\usepackage{pst-tree}
\usepackage{subfigure}
\usepackage{epic}
\usepackage{alltt}
\usepackage[all]{xy}
\usepackage{appendix}
\usepackage{ulem}
\usepackage{natbib}

\title{A general stochastic model for sporophytic self-incompatibility}

\author{Sylvain Billiard\thanks{Laboratoire de G\'{e}n\'{e}tique et Evolution des Populations V\'{e}g\'{e}tales, UFR de Biologie, USTL, Cit\'{e} Scientifique, 59655 Villeneuve d'Ascq Cedex, France. sylvain.billiard@univ-lille1.fr} ~ and ~  Viet Chi Tran\thanks{Laboratoire P. Painlev\'{e}, UFR de Maths, USTL, Cit\'{e} Scientifique, 59655 Villeneuve d'Ascq Cedex, France ; Centre de Math\'{e}matiques Appliqu\'{e}es, Ecole Polytechnique, Route de Saclay, 91128 Palaiseau Cedex, France. chi.tran@math.univ-lille1.fr} ~ \thanks{Both authors have equally contributed to this paper.}}

\date{\today}

\numberwithin{equation}{section}

\newcommand{\Co}{\mathcal{C}}

\newcommand{\lbrac}{\left[\!\left[}
\newcommand{\rbrac}{\right]\!\right]}

\def\N{\mathbb{N}}
\def\wN{\widetilde{N}}

\def\P{\mathbb{P}}
\def\R{\mathbb{R}}
\def\E{\mathbb{E}}

\def\ind{{\mathchoice {\rm 1\mskip-4mu l} {\rm 1\mskip-4mu l}
{\rm 1\mskip-4.5mu l} {\rm 1\mskip-5mu l}}}

\def\eg{\textit{e.g.} }
\def\ie{\textit{i.e.} }

\def\Card{\mbox{Card}}

\newcommand{\be} {\begin{equation}}
\newcommand{\ee} {\end{equation}}
\newcommand{\bea} {\begin{eqnarray}}
\newcommand{\eea} {\end{eqnarray}}
\newcommand{\Bea} {\begin{eqnarray*}}
\newcommand{\Eea} {\end{eqnarray*}}

\theoremstyle{plain}
\newtheorem{theorem}{Theorem}[section]
\newtheorem{proposition}[theorem]{Proposition}

\newtheorem{corollary}[theorem]{Corollary}
\newtheorem{definition}[theorem]{Definition}

\newtheorem{remark}[theorem]{Remark}
\newtheorem{example}[theorem]{Example}

\begin{document}

\maketitle

\begin{abstract}
Disentangling the processes leading populations to extinction is a major topic in ecology and conservation biology. The difficulty to find a mate in many species is one of these processes. Here, we investigate the impact of self-incompatibility in flowering plants, where several inter-compatible classes of individuals exist but individuals of the same class cannot mate. We model pollen limitation through different relationships between mate availability and fertilization success. After deriving a general stochastic model, we focus on the simple case of distylous plant species where only two classes of individuals exist. We first study the dynamics of such a species in a large population limit and then, we look for an approximation of the extinction probability in small populations. This leads us to consider inhomogeneous random walks on the positive quadrant. We compare the dynamics of distylous species to self-fertile species with and without inbreeding depression, to obtain the conditions under which self-incompatible species could be less sensitive to extinction while they can suffer more pollen limitation.
\end{abstract}

\noindent \textbf{Keywords:} birth and death process, ODE approximations, inhomogeneous random walk on the positive quadrant, inbreeding depression, extinction probability, mating system, distyly.\\
\noindent \textbf{AMS Codes:} 92D40, 92D25, 60J80.\\

\section{Introduction}

We are interested in modeling the specific sexual mating system of a plant population, and especially in highlighting several phenomena which can affect its dynamics:
stochasticity in the demography, pollen limitation, boundary effects.  First, the fate of small populations depends on stochastic processes such as demographic stochasticity \citep{lande98}, which refers to a variance due to randomness in the death and reproduction events. Second, when populations are small or at low density, there can happen Allee effects \citep{allee49}, \ie when a positive relationship between the size of the population and the per capita growth rate appears. In sexual species, the birth rate may depend on the availability in mating partner, what is called the ``mate finding Allee effect'' \citep{gascoigneberecgregorycourchamp}. For example, in species with separate sexes, the mean reproductive rate of a population may be low because the opposite sex is rare \citep{engenlandesaether, saetheretal04,bessaetal04}. These biases can be expected to be higher in small populations because of stochasticity. Finally, at the extreme, the mate
finding Allee effect can result in boundary conditions where compatible
mates disappear, leading inevitably the population to extinction.\\
Many hermaphroditic species of Angiosperms (flowering plants) have mating systems which allow fertilization only among specific classes of individuals: between long-styled and short-styled
plants in distylous species \citep{barrettshore}, or among different mating types in self-incompatible species (more than 50 percent of Ansgiosperm families, \citet{igicetal08}). Self-incompatibility (SI) avoids selfing and mating
between close individuals: SI prevents
reproduction between mates sharing identical S-locus alleles (the allele which determines whether individuals are compatible, by coding for some proteins carried as identifiers by the pollen and ovules), and especially self-fertilization. It was hypothesized that this mating system evolved to avoid inbreeding depression,
\ie the decrease in fitness when mating occurs between kin individuals \citep{porcherlande_jeb, porcherlande_evolution}.
In this system, individuals carrying rare S-locus alleles have access to more possible mates than individuals carrying common S-alleles, which generates negative frequency-dependent selection on the S-locus \citep{wright39}.\\
In Angiosperms, mate finding is moreover mainly passive since it depends on pollen vectors such as insects or wind. This can result in ``pollen limitation'', when a given plant does not receive enough pollen to fertilize all
the ovules it produces. Pollen limitation has been found in many species and
can have many evolutionary and ecological consequences, from the evolution of
mating system to the increase of extinction probabilities in small populations
\citep{ashmanetco}. If a class of individuals goes extinct, other classes have less mating opportunities leading to a lower mean reproduction rate, or even can
not mate anymore, leading to extinction.
Pollen limitation and mating systems are interacting phenomena which can have a strong impact on the fate of populations (see \citet{leducqetal10} and references therein). Shortly, pollen limitation can be increased in strict allogamous species
and the probability of disappearance of a class of individuals can be increased by pollen
limitation. It has been shown indeed that pollen limitation is higher in outcrossing species than in selfing species \citep{larsonbarrett2000}.\\
However, few theoretical investigations have
been done so far to measure the impact of
pollen limitation, mating systems and their interactions on the
extinction and establishment of populations. It is yet a central
question since many Angiosperm families are strictly allogamous because
of SI.  Three demographic models have been investigated
specifically for a given species
\citep{kirchnerrobertcolas,wageniuslonsdorfneuhauser,hoebeethrallyoung}, which thus lack generality. \citet{levinkelleysarkar} considered as us a general model for which they performed individual
based simulations. Moreover, in \citet{kirchnerrobertcolas,hoebeethrallyoung,levinkelleysarkar}, the supposed SI systems were either
gametophytic (only the content of the gametes is expressed) or a
caricatural sporophytic system with only codominance between S-alleles.
In species with sporophytic SI system (SSI), the mating phenotypes of
pollen and pistils are determined by the diploid parental genotypes at
the S-locus and hence dominance interactions are possible among
S-alleles \citep{bateman52}, which can be very complex \citep{castricvekemans04}. The results from the previous
studies can therefore not be generalized to species with SSI, and
especially not to distylous species, which is a particular case of SSI
with only two alleles and consequently only two classes of individuals.
Finally, the different processes affecting the extinction probabilities of
SI populations (pollen limitation, demographic stochasticity and
boundary effect) can not be disentangled in these previous
investigations, which is our purpose here.\\

Our goals are, first to develop a general model to describe the dynamics of a plant population with SSI, with and without pollen limitation and second to use this model to investigate the relative impact of pollen limitation and demographic stochasticity on the fate of populations in the particular case of a distylous species.
We begin with a general stochastic individual based model, in continuous time, for SSI (Section \ref{sectiondescrmicro}). We explicitly model the genetic determinism of the SI phenotype and compute the reproduction rate of each possible genotype. We assume different relationships between the compatible mate availability and the reproductive success, which reflect different models of pollen limitation. Second, we focus on the simple case of distylous species. We analyze the dynamics of distylous species assuming a large population and using approximations by ordinary differential equations (ODEs). We exhibit different behaviors depending on the relationships between the birth and death parameters. We also compare these behaviors with self-fertile species, with or without inbreeding depression and pollen limitation (Section \ref{section:largepop}). Third, we consider extinction in small distylous populations (Section \ref{section:smallpop}). This leads us to study inhomogeneous random walks on the positive quadrant. We use coupling arguments to show that the behaviors of the random walks, namely whether they are subcritical, critical or supercritical, depend on the same criteria as for large populations. This also provides estimates for the extinction probabilities. Finally, individual based simulations are performed (Section \ref{section:simulations}).

\section{Microscopic modelling}\label{sectiondescrmicro}

We first describe the individual dynamics. Then, we precise the different models of reproduction rates. Following \citet{champagnatferrieremeleard}, we propose Stochastic Differential Equations
(SDEs) that mathematically describe the random evolution in time of the population and its large population approximation
using ODEs.

\subsection{Description of the dynamics at the individual level}\label{inddynadescr}

\paragraph{Self-incompatibility genotype} We assume that SSI is controlled by a gene at a single locus, called the S-locus, where $n$ alleles segregate. These alleles are numbered from 1 to $n$. We assume that individuals
are diploid and hermaphroditic. Each individual genotype is of the form $G=\{S^1,S^2\}$, where $S^1$ and $S^2$ are two alleles in $\lbrac 1, n\rbrac=\{1,\dots,n\}$. Since the order of the alleles is not important,
for $u,v\in \lbrac 1, n\rbrac$, $\{u,v\}$ and $\{v,u\}$ are the same genotype.\\
We denote by $E=\{g=\{u,v\},\, u,v\in \lbrac 1, n\rbrac\}$ the space of the different possible genotypes. The set $E$ is finite, with $\Card(E)=n(n+1)/2$.

\paragraph{Individual based model} The population at $t$ is given by a point measure on $E$
\begin{equation}
Z_t(dg)=\sum_{i=1}^{N_t} \delta_{\{S^1(i),S^2(i)\}}(dg)=\sum_{g'=\{u',v'\}\in E} N_t^{u'v'}\delta_{\{u',v'\}}(dg)\label{defzt}
\end{equation}
where $N_t$ is the number of individuals alive at time $t$, $N_t^{u'v'}$ is the number of individuals with genotype $\{u',v'\}\in E$. Each individual $i$ is represented by a Dirac mass weighting its genotype $\{S^1(i),S^2(i)\}$.
Point measures that are considered have finite mass. We denote by $\mathcal{M}_F(E)$ the set of finite measures on $E$. Since $E$ is a finite space, the weak and vague topologies, and the topology associated with the total variation norm on $\mathcal{M}_F(E)$ all coincide.\\

\par Integrating the measure (\ref{defzt}) with respect to chosen real functions $f$ on $E$ provides summaries of the population. We denote $\langle Z_t,f\rangle=\int_{E}f(g)Z_t(dg)=\sum_{i=1}^{N_t}f(\{S^1(i),S^2(i)\})$. \\
If we choose for instance $f=1$ then, $\langle Z_t,1\rangle=N_t$. If we choose $f=\ind_{\{u',v'\}}$ for $\{u',v'\}\in E$, then $\langle Z_t,\ind_{\{u',v'\}}\rangle=N^{u'v'}_t$.

\paragraph{Self-incompatibility phenotypes} SSI phenotype is the production of recognition proteins at the surface of pollen and stigmas (stigmas contain ovules and are the plant's structure on which pollen is deposited).
These proteins can be of $n$ different types (or specificities) depending on the plant's genotype and on the dominance relationships between the $n$ alleles.
For $u\in \lbrac 1, n\rbrac$, let us denote by $e_u\in \R^n$ the vector with all components equal to zero except the component $u$ which is equal to 1.
For an individual of genotype $\{u,v\}$, the stigmas and pollen it produces are said to be of type $\{u,v\}$, and have respectively the phenotypes $\Phi_S (e_u+e_v)$ and $\Phi_P (e_u+e_v)$ in
$\{0,e_u,e_v,e_u+e_v\}$.
The maps $\Phi_P$ and $\Phi_S$ encode the expression of the genotype into the phenotype. For instance, if $ \Phi_P (e_u+e_v)=e_u$, then $u$ is strictly dominant over $v$ in pollen (proteins produced at the surface of pollen are of the single type $u$).
If $\Phi_P (e_u+e_v)=e_u+e_v$,
then $u$ and $v$ are codominant in pollen (proteins produced at the surface of pollen are of both types $u$ and $v$). If $\Phi_P(e_u+e_v)=0$ then the pollen can fertilize any possible ovule in the population.
A cross is incompatible if pollen and stigmas share proteins of the same type. In other words, a stigma $\{u,v\}$ and a pollen $\{u',v'\}$
can cross if and only if
\begin{equation}\Phi_S(e_u+e_v)\cdot \Phi_P(e_{u'}+e_{v'}) =0,\label{condition_compatibilite}
 \end{equation}where $x\cdot y$ is the scalar product of $x,y\in \R^n$. Notice that an individual can self-fertilize if $\Phi_S(e_u+e_v)\cdot \Phi_P(e_{u}+e_{v}) =0$. Let:
\begin{equation}\Theta^{uv}_S=\{ \{u',v'\}\in E,\quad 
\Phi_S(e_u+e_v)\cdot \Phi_P(e_{u'}+e_{v'}) =0\}\end{equation}denote the set of genotypes compatible with stigmas $\{u,v\}$. The size of the population producing pollen compatible with stigmas $\{u,v\}$ is hence:
\begin{equation}
\overline{N}_t^{uv}=\sum_{\{u',v'\}\in \Theta^{uv}_S} N_t^{u'v'}=\int_{E}\ind_{\Phi_S(e_u+e_v)\cdot \Phi_P (e_{u'}+e_{v'}) =0}Z_t(d\{u',v'\}).\label{defNbar(uv)}
\end{equation}
Notice that the set $\Theta^{uv}_P$ of genotypes compatible with pollen $\{u,v\}$ is not necessarily $\Theta^{uv}_S$ since we may have $\Phi_S(e_u+e_v)\cdot \Phi_P(e_{u'}+e_{v'})\not= \Phi_P(e_u+e_v)\cdot \Phi_S(e_{u'}+e_{v'})$.

\paragraph{Mating probabilities} Each plant receives pollen from the rest of the population. We assume that the quantity of pollen of type $\{u,v\}\in E$ received by a plant is equal to the proportion of plants of genotype $\{u,v\}$ in the population. The probability that the ovule of an individual with genotype $\{u',v'\}$ is fertilized by a pollen from an individual with genotype $\{u,v\}$ is denoted by $p^{u'v'}(u,v)$. This probability takes into account the quantity of pollen received by the pistil and the compatibility conditions (\ref{condition_compatibilite}). This is detailed in Section \ref{sectionreproductionrate}.

\paragraph{Segregation} When the genotypes of the plants having produced the ovule and pollen are respectively $\{u,v\}$ and $\{u',v'\}$, then the new individual is of genotype $\{u,u'\}$, $\{u,v'\}$, $\{v,u'\}$ or $\{v,v'\}$ with probability 1/4.

\paragraph{Reproduction rates} Plants $\{u,v\}$ produce ovules in continuous time with an individual rate $\bar{r}>0$. Once produced, the ovules may be fertilized or not depending on the quantity and compatibility of the pollen arriving on the stigmas, thus providing seeds. We define as $R(\overline{N}_t^{uv},N_t)$ the individual reproduction rate of an individual of genotype $\{u,v\}$, which is the individual rate of production of seeds. This rate is bounded by $\bar{r}$ and can depend
on the size of the compatible population or be a constant: it is the product of $\bar{r}$ with the fertilization probability
$$R(\overline{N}_t^{uv},N_t)=\bar{r}\sum_{\{u',v'\}\in E}p^{uv}_t(u',v').$$
The functional forms of $R(.)$ are further discussed in Section \ref{sectionreproductionrate}.

\paragraph{Death} Each plant dies with the constant rate $d>0$. This death rate is kept simple and the complexity of the model is put on the reproduction system.

\subsection{Functional forms of mating probabilities and reproduction rates}\label{sectionreproductionrate}

We describe the five types of reproduction considered in this paper. Three of them describe SI systems: Wright's model \citep{wright39}, the fecundity selection model \citep{vekemansetal98} and the dependence model, which we introduce. In these models, ovules are produced at a constant rate, and given reproduction, the partner is chosen uniformly among compatible individuals. The differences are based on the choice of functional forms for mating probabilities.
The resulting reproduction rates for these three models are represented in Fig. \ref{fig:billiardtranSSI}. Two other models are also introduced for comparison: the self-compatible cases with and without inbreeding depression.\\
Boundary conditions, \textit{i.e.} discontinuities of the reproduction rate that occur when a genotype in the population has no more mate, are observed in the Wright's and dependence models. For the fecundity selection model, even if there is SI, the reproduction rate is proportional to the size of the compatible population and vanishes continuously at the boundary. In self-compatible cases, there is no boundary effect because an individual can mate with any other individuals, even itself.\\
In Wright's model, there is no pollen limitation. In the fecundity selection and the dependence models, pollen limitation is introduced through the mating probabilities.

\subsubsection{Model 1: Wright's model}\label{section:wright}

First, following \citet{wright39}, each individual with genotype $\{u,v\}\in E$ produces ovules of type
$\{u,v\}$, one at a time, at the constant $\bar{r}$. An ovule of type $\{u,v\}$ is fertilized by a pollen produced by an individual of genotype $\{u',v'\}$ with a probability that depends on the frequencies in the compatible population:
\begin{equation}
\overline{p}_t^{uv}(u',v')=\frac{N_t^{u'v'}}{\overline{N}_t^{uv}} \ind_{\{u',v'\}\in \Theta^{uv}_S}.\label{defpbar}
\end{equation}
Notice that because of \eqref{defNbar(uv)}, $\overline{p}_t^{uv}(u',v')=0$ when $\overline{N}_t^{uv}=0$. The fertilization probability is 1 if there is at least one compatible individual in the population, $0$ otherwise. The individual reproduction rate is
$R(\overline{N}^{uv}_t,N_t)=\bar{r}\ind_{\overline{N}_t^{uv}>0}.$ This rate does not depend on the quantity of compatible pollen a plant receives.

\subsubsection{Model 2: Dependence model}\label{section:dependence}

Pollen limitation and the self-incompatibility may be modeled by considering mating probabilities of the form:
 \begin{align}
p_t^{uv}(u',v')= \overline{p}_t^{uv}(u',v') \frac{r\big(\overline{N}^{uv}_t\big)}{\bar{r}}\label{def:ptilde}
 \end{align}where $\overline{p}_t^{uv}(u',v')$ is defined in \eqref{defpbar} and where, for positive constants $\alpha$ and $\beta$:
\begin{equation}r(N)=\bar{r}\frac{e^{\alpha N}}{\beta+e^{\alpha N}}.\label{functionrN}\end{equation}In \eqref{def:ptilde}, the probability of finding a mate depends on the number of compatible individuals, and given reproduction, the partner is chosen uniformly at random among the latter. Thus, the fertilization probability is $\ind_{\overline{N}^{uv}_t>0}r(\overline{N}^{uv}_t)/\bar{r}$ and the individual reproduction rate is $R(\overline{N}^{uv}_t,N_t)=r(\overline{N}^{uv}_t)\ind_{\overline{N}^{uv}_t>0}$.
 This model will be termed ``dependence model'' in the sequel. It is a generalization of Wright's model with a fertilization probability that depends on $\bar{N}^{uv}_t$ when reproduction is allowed. This is pollen limitation: among the $\overline{r}$ ovules produced in a time unit, only $R(\overline{N}_t^{uv},N_t)$ end up in producing a new individual. For a large compatible population, the rate tends to its supremum value:
$$\lim_{N\rightarrow +\infty}\bar{r}\frac{e^{\alpha N}}{\beta+e^{\alpha N}} \ind_{N>0}=\bar{r}.$$Wright's model can be viewed as the limiting case of (\ref{functionrN}) when $\alpha\rightarrow +\infty$, that is when there is no pollen limitation.\\

Notice also that this model exhibits boundary effects since
$$\lim_{N\rightarrow 0}\bar{r}\frac{e^{\alpha N}}{\beta+e^{\alpha N}} \ind_{N>0}= \frac{\bar{r}}{1+\beta}>0.$$Because we assume that a single plant produces a lot of pollen, even in cases where we consider pollen limitation,
there is a discontinuity at the set $\{\overline{N}^{uv}=0\}$ when $r(0)>0$. As soon as there exists a small number of compatible plants, all individuals start producing offspring at a positive rate.

\subsubsection{Model 3: Fecundity selection model}\label{section:fecundity}

In the fecundity selection model \citep{vekemansetal98}, the mating probability between an ovule produced by an individual of genotype $\{u,v\}$ and a pollen produced by $\{u',v'\}$ is
\begin{equation}
p^{uv}_t(u',v')=\overline{p}^{uv}_t(u',v')\frac{\overline{N}^{uv}_t}{N_t}=\frac{N^{u' v'}_t}{N_t}\ind_{\{u',v'\}\in \Theta^{uv}_S}\label{fecundityselection}.
\end{equation}If the ovule chooses an incompatible pollen, then it is not fertilized and lost. Under the fecundity selection model, the fertilization probability of an ovule $\{u,v\}$ is $\overline{N}^{uv}_t/N_t$ and the reproduction rate of an individual $\{u,v\}$ is $R(\overline{N}^{uv}_t,N_t)=\bar{r}\overline{N}_t^{uv}/N_t$.
The individual reproduction rate is directly proportional to the frequency of compatible individuals, which thus promotes pollen limitation.

\begin{figure}
\vspace{0.5cm}
\begin{center}
    \unitlength=1.3cm
    \begin{picture}(7,5)
    \put(1,1){\vector(1,0){6}}
    \put(1,1){\vector(0,1){4}}
    \put(0.6,5.3){$R(\overline{N}^{uv}_t,N_t)$}
    \put(7.1,0.9){$\overline{N}^{uv}_t$}
    \put(0.8,1.2){$0$}
    \put(1.2,0.8){$0$}
    \put(0.8,4.65){$\bar{r}$}
    \put(1.4,4.8){Wright's model}
    \put(1.35,3){Dependence model}
    \put(1.3,4.7){\line(1,0){5.3}}
    \put(1.3,1.3){\line(0,1){3.4}}
    \put(1,1.3){\line(1,0){0.3}}
    \put(3,2){Fecundity selection model}
    \dashline{0.2}(1.3,1.3)(6.6,4.4)
    \dashline{0.05}(1,4.4)(6.6,4.4)
    \dashline{0.05}(6.6,1)(6.6,4.4)
    \put(-0.13,4.4){$\bar{r}(1-\frac{1}{N_t})$}
    \put(6.5,0.65){$N_t-1$}
    \thicklines
    \dashline{0.1}(1.3,1.3)(1.3,2.8)
    \dashline{0.1}(1.3,2.8)(2,3.4)(3,4)(4,4.3)(5,4.45)(6,4.6)(6.6,4.65)
    \put(0.55,2.75){$\frac{\bar{r}}{1+\beta}$}
    \end{picture}
\end{center}
\vspace{-1cm}
\caption{\textit{Relationships between the individual reproductive rates of genotype $\{u,v\}$ and the number of compatible individuals $\overline N^{u,v} $ for the three models of mating we assumed: Thin line: The Wright's model; Dashed line: The fecundity selection model; Thick dashed line: The dependence model.}}\label{fig:billiardtranSSI}
\end{figure}
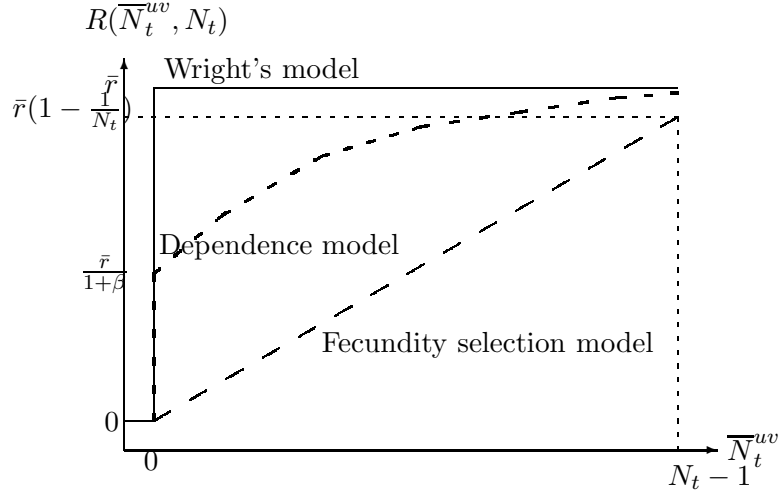


\subsubsection{Rates of reproduction and gamete production}

In the models of Sections \ref{section:wright}, \ref{section:dependence} and \ref{section:fecundity} 
the rate at which pollen produced by individuals of genotype $\{u,v\}$ fertilizes ovules in the population $Z_t$ is:
\begin{align}
\sum_{\{u', v'\}\in E}\overline{r}N^{u'v'}_t p^{u'v'}(u,v)\label{generalratepollen(uv)}
\end{align}where we recall that $p^{u'v'}(u,v)$ is the mating probability of an ovule produced by an individual $\{u',v'\}$ by a pollen produced by an individual $\{u,v\}$. Under Wright's model $p^{u'v'}(u,v)=\overline{p}^{u'v'}(u,v)$.\\
Hence, offspring with genotype $\{u,v\}$ are produced with the rate:
\begin{align}
r^{uv}(Z_t) = \overline{r} \ind_{u\not=v}\Big[ & \frac{1}{4}\Big(  \sum_{u'\not= u}\sum_{v'\not=v} N_t^{uu'}p^{uu'}(v,v')+N_t^{vv'}p^{vv'}(u,u')\Big)\nonumber\\
+ & \frac{1}{2}\Big(\sum_{v'\not=v}N_t^{uu}p^{uu}(v,v')+\sum_{u'\not=u}N_t^{uu'}p^{uu'}(v,v)\Big)\nonumber\\
+ & \frac{1}{2}\Big(\sum_{u'\not=u}N_t^{vv}p^{vv}(u,u')+\sum_{v'\not=v}N_t^{vv'}p^{vv'}(u,u)\Big)\nonumber\\
+ & \quad \Big(N_t^{uu}p^{uu}(v,v)+N_t^{vv}p^{vv}(u,u)\Big)
\Big]\nonumber\\
+\overline{r}\ind_{u=v} \Big[  & \frac{1}{4}  \sum_{u',v'\in \lbrac 1, n\rbrac \setminus \{u\}}N_t^{uu'}p^{uu'}(u,v')\nonumber\\
+ &   \frac{1}{2}\Big(  \sum_{u'\not= u}N_t^{uu'}p^{uu'}(u,u)
 +  \sum_{v'\not=u}N_t^{uu}p^{uu}(u,v')\Big)\nonumber\\
 + & N_t^{uu}p^{uu}(u,u)\Big]. \label{tauxreprod(uv)}
\end{align}The formula \eqref{tauxreprod(uv)} does not simplify in general because of the possible asymmetry of dominance relationships between alleles in pollen and stigmas. We distinguish whether the allele $u$ comes from the pollen or the pistil, and similarly for $v$. We separated the terms according to the parents' homozygosity or heterozygosity.\\
The rate (\ref{tauxreprod(uv)}) is related to the number of seeds of each genotype $\{u,v\}$ that is produced at time $t$. There is a \textit{mass dependence} with respect to the parent who brings the ovules and a \textit{frequency dependence} with respect to the one who brings the pollen. The case of infinite populations with the Wright's and the fecundity selection model is considered in \citet{billiardcastricvekemans}.


We end this section, with two alternative models without SI, for future comparisons with the three models introduced in Sections \ref{section:wright}-\ref{section:fecundity}. Model 4 describes a self-compatible case without any penalization of self-fertilization, whereas Model 5 introduces inbreeding depression, the phenomenon which describes the decrease in fitness when mating occurs between kin individuals \citep{porcherlande_jeb, porcherlande_evolution}.

\subsubsection{Model 4: Self-compatibility without pollen limitation nor inbreeding depression}\label{SCsectionreproductionrate}

To carry out comparisons, it is natural to investigate the case where there is no SI, no pollen limitation and no inbreeding depression. Then, when the reproduction $\bar{r}$ is constant, the rate of production of offspring with genotype $\{u,v\}$ becomes:
\begin{align}
\widetilde{r}^{uv}(Z_t)= & \frac{2\bar{r}}{N_t}\Big(\frac{1}{2}\sum_{u'\not= u} N_t^{uu'}+N^{uu}_t\Big)
\Big(\frac{1}{2}\sum_{v'\not= v} N_t^{vv'}+N^{vv}_t\Big)\quad \mbox{ if }u\not= v\nonumber\\
\widetilde{r}^{uu}(Z_t)= & \frac{\bar{r}}{N_t}\Big(\frac{1}{2}\sum_{u'\not= u} N_t^{uu'}+N^{uu}_t\Big)^2\quad \mbox{ if }u= v.
\end{align}In this rate, we recognize the number of couples that can be constituted with one parent of allele $u$ and one parent of allele $v$,
$\sum_{u'=1}^n N_t^{uu'}\sum_{v'=1}^n N_t^{vv'}$, divided by the total size of the population because of the frequency dependence, and with a correction for homozygous individuals. Notice that the genotypic frequency appears naturally since here, $p^{uv}_t(u',v')=N^{u',v'}_t/N_t$ is the proportion of individual $\{u',v'\}$ in the population.

\subsubsection{Model 5: Self-compatibility with pollen limitation and inbreeding depression}

A detailed study of inbreeding depression would be interesting and deserves a paper for itself. Our purpose here is to provide a schematic point of comparison for SSIs. 
For the sake of simplicity, we assume here that inbreeding depression is suffered by offspring produced by self-fertilization only.\\
Among the ovules produced by the plant, a fraction $s\in [0,1]$ are self-fertilized, \ie are fertilized by the pollen of the plant that has produced them. Among these, a fraction $\delta\in [0,1]$ is not viable.
The rate of production of new individuals through other matings is of the form (\ref{functionrN}) except that the maximal rate is now $(1-s)\bar{r}$ which corresponds to the ovules that have not been self-fertilized. We hence consider the following rate in the case of pollen limitation:
\begin{equation}
R(N_t)=(1-\delta)s\overline{r}+(1-s)\overline{r}\frac{e^{\alpha (N_t-1)}}{\beta+e^{\alpha (N_t-1)}}\ind_{N_t>1}. \label{forme_r(Ncomp)_inbreedingdepr}
\end{equation}In the sequel, this model will be termed by Self-Compatibility with Inbreeding Depression model (SCID). In the presence of self-fertilization, the effect of pollen limitation is lower: the reproduction rate of individuals is less dependent on the quantity of pollen received from the other individuals.


With the rate \eqref{forme_r(Ncomp)_inbreedingdepr}, the rate of production of offspring with genotype $\{u,v\}$, at time $t$ is:
\begin{align}
\widetilde{r}^{uv}(Z_t)= & 2 (1-s)\frac{\bar{r}}{N_t}\frac{e^{\alpha (N_t-1)}}{\beta+e^{\alpha (N_t-1)}}\Big(\frac{1}{2}\sum_{u'\not= u} N_t^{uu'}+N^{uu}_t\Big)
\Big(\frac{1}{2}\sum_{v'\not= v} N_t^{vv'}+N^{vv}_t\Big)\nonumber\\
+ & \frac{(1-\delta)s\bar{r}}{2} N^{uv}_t \quad \mbox{ if }u\not= v\label{tauxreprod(inbreeding)}\\
\widetilde{r}^{uu}(Z_t)= & (1-s)\frac{\bar{r}}{N_t}\frac{e^{\alpha (N_t-1)}}{\beta+e^{\alpha (N_t-1)}}\Big(\frac{1}{2}\sum_{u'\not= u} N_t^{uu'}+N^{uu}_t\Big)^2 +  (1-\delta)s\bar{r} N^{uu}_t ~ \mbox{ if }u= v.\nonumber
\end{align}

\subsection{SDE description and their ODE approximations}\label{sectionODE}

Following \citet{fourniermeleard}, we can express the evolution of the population $(Z_t)_{t\in \R_+}$ by mean of a SDE driven by a Poisson point measure. This equation is given in Appendix \ref{section:SDEdetaillees} and corresponds to the mathematical formulation of the individual-based algorithm of Section \ref{algosimu}. From this, we derive the evolution of $N_t^{uv}$:
\begin{align}
N_t^{uv} = & N_0^{uv} +\int_0^t \big(r^{uv}(Z_s) -d \times N^{uv}_s\big) ds+M^{uv}_t\label{appleffectif}
\end{align}
where $M^{uv}$ is a square integrable martingale starting at 0 that can be written explicitly in term of the Poisson point measure and with bracket:
\begin{align}
 \langle M^{uv}\rangle_t=  \int_0^t \big(r^{uv}(Z_s) +d \times N^{uv}_s\big) ds.\label{crochet}
\end{align}

Heuristically, we can interpret the martingale as a noise term corresponding to the stochasticity and whose variance is given by \eqref{crochet}. The integral term in \eqref{appleffectif} gives the growth rate of the population, as for usual ODEs of population dynamics, which are more usual in the biological literature (see \eqref{ODExi} in the sequel). We refer to \cite{ikedawatanabe,joffemetivier} for introductions to such SDEs.\\
These SDEs are linked in large populations with classical ODEs. In similar cases with such nonlinear dynamics, we know that these ODEs arise as deterministic approximations of the SDEs under a certain large population limit \citep[see][]{champagnatferrieremeleard,ethierkurtz,fourniermeleard,trangdesdev}.
Apart from providing the mathematical frame under which the SDEs can be approximated by ODEs, limit theorems may be interesting for statistical estimation, in particular for linking individual-based statistics with parameter estimates for the ODE \citep[\eg][]{blumtran,arazozaclemencontran}.\\
For the limit theorem, we consider a sequence of populations with initial sizes of order $K$ (Point $(I-i)$ of Definition \ref{defprocrenorm}). To counterbalance this effect, we renormalize the size of the individuals in $1/K$ (Point \textit{II}) so that the population mass remains of the order of a constant. The renormalization of the reproduction rate $(I-ii)$ is a kind of localization factor: if the size of the population is large, only the neighborhood of a given individual will affect its reproduction rate.

 \begin{definition}\label{defprocrenorm}I) We consider the following sequence of processes $(Z^K)_{K\in \N^*}$, where $\N^*=\N\setminus \{0\}$.
\par (i) Let $(Z^K_0)_{K\in \N^*}$ be initial conditions such that there exists a deterministic finite measure $\xi_0\in \mathcal{M}_F(E)$ such that:
\begin{equation}
\lim_{K\rightarrow +\infty}\frac{Z^K_0(dg)}{K}=\xi_0(dg)~\mbox{ and }~ \exists \varepsilon>0,\, \sup_{K\in \N^*}\E\Big(\big(\frac{N_0^K}{K}\big)^{2+\varepsilon}\Big)<+\infty,\label{deflimitez0n}
\end{equation}where $N^K_t=\langle Z^K_t,1\rangle$ is the size of the population described by $Z^K_t$.\\
\par (ii) The rate of production of a given genotype $g$, $r^g(.)$ is replaced by $r^{g,K}(.)=r^g(./K)$. 
The death rate $d$ is unchanged.\\

\noindent II) We also introduce the sequence of renormalized processes $(Z^{(K)})_{K\in \N^*}$:
\begin{equation}
\forall K\in \N^*,\, \forall t\in \R_+,\; Z^{(K)}_t(dg)=\frac{1}{K}Z^K_t(dg).\label{defz(n)}
\end{equation}We define by $N^{(K)}_t=N^K_t/K$ the mass of $Z^{(K)}_t$, that is the renormalized size of the population.\qed
\end{definition}
We use exponents $(K)$ with brackets for the renormalized quantities and exponents $K$ without brackets for the non-renormalized ones. The moment condition in (\ref{deflimitez0n}) is technical and ensures that \citep[for a proof, see][]{fourniermeleard}:
\begin{equation}
\forall T\in \R_+,\, \sup_{K\in \N^*}\E\big(\sup_{t\in [0,T]} (N^{(K)}_t)^2\big)<+\infty.\label{momentordre2}
\end{equation}

The ODEs are obtained when $K\rightarrow +\infty$. The corresponding limit theorem is stated in the next proposition and proved in appendix.

\begin{proposition}\label{proprenormgdepop2}The sequence of renormalized processes $(Z^{(K)})_{K\in \N^*}$ converges uniformly, when $K\rightarrow +\infty$ to the process in $\Co(\R_+,\mathcal{M}_F(E))$ such that:
\begin{equation}
\xi_t(dg)=\sum_{g=\{u,v\}\in E}n_t^{uv}\delta_{\{u,v\}}(dg),\label{defxi}
\end{equation}where for every $\{u,v\}\in E$, $n^{uv}_t$ is the \textit{continuous number} of plants of genotype $\{u,v\}$ at time $t$ and satisfies the following ODE:
\begin{align}
\frac{dn^{uv}_t}{dt}= r^{uv}(\xi_t)- d \, n^{uv}_t\label{ODExi}
\end{align}where $
r^{uv}(\xi_t)$ is obtained from (\ref{tauxreprod(uv)}) by replacing all the $N^{uv}$'s by $n^{uv}$'s.
\end{proposition}

\subsection{Particular case of distylous flowers}\label{section:distyle}

The most simple case of SSI is distyly where $n=2$ and $\Phi_P=\Phi_S\equiv \Phi$. We detail this case, which will be studied carefully in all the rest of the paper. We will see that the two alleles in this case can not be codominant, which corresponds to what happens in natural populations. This leads us to introduce a random walk on the positive quadrant, which is central in the sequel.\\

The possible genotypes are $\{1,1\}$, $\{1,2\}$, $\{2,2\}$. The population is described with $N^{11}$, $N^{12}$ and $N^{22}$. Two alleles $1$ and $2$ are codominant if:
\begin{equation}
\Phi(2 e_1)=e_1,\quad \Phi(2 e_2)=e_2, \quad \Phi(e_1+e_2)=e_1+e_2.\label{codominancedistyle}
\end{equation}
The allele $2$ is dominant and $1$ is recessive if:
\begin{equation}
\Phi(2 e_1)=e_1,\quad \Phi(2 e_2)=e_2, \quad \Phi(e_1+e_2)=e_2.\label{dominancedistyle}
\end{equation}
\begin{proposition}
In the case of codominance (\ref{codominancedistyle}), there is almost sure extinction of the population.
\end{proposition}
\begin{proof}Since $\Phi(e_1+e_2)\cdot \Phi(e_1)=\Phi(e_1+e_2)\cdot \Phi(e_2)=1$ and $\Phi(e_1+e_2)\cdot\Phi(e_1+e_2)=2$, heterozygous individuals $\{1,2\}$ are unable to mate with any individual. The only possible match is between individuals $\{1,1\}$ and $\{2,2\}$, but this produces individuals $\{1,2\}$ which have no compatible mate. Hence, there is at most two generations of individuals. Since the death rate is a positive constant $d>0$, extinction happens in finite time almost surely.
\end{proof}

In the sequel, we will therefore assume that alleles $1$ and $2$ are not codominant. Let us consider the case where $2$ is dominant and $1$ recessive (\ref{dominancedistyle}), the symmetric case being treated in the same manner. In this case, $\{1,1\}$ can mate with $\{1,2\}$ and $\{2,2\}$ whereas the latter can only mate with $\{1,1\}$. The sizes of the respective compatible populations (\ref{defNbar(uv)}) are:
\begin{align}
\overline{N}^{11}_t=N^{12}_t+N^{22}_t,\quad \overline{N}^{12}_t=N^{11}_t,\quad \overline{N}^{22}_t=N^{11}_t.
\end{align}
First, since none of these matings produces offsprings of genotype $\{2,2\}$, this genotype is only present in the first generation and hence disappears in finite time almost surely. For the sake of simplicity, we assume that the initial condition is only made of individuals of genotypes $\{1,1\}$ and $\{1,2\}$. In this case, the genotype $\{2,2\}$ never appears and $N^{12}_t=N_t-N^{11}_t$. Individuals of genotype $\{1,1\}$ can only reproduce with $\{1,2\}$ and reciprocally. As soon as one of the genotypes $\{1,1\}$ or $\{1,2\}$ disappears, the whole population is doomed to extinction since the remaining genotype can not reproduce any more. We are led to study random walks on the positive quadrant, absorbed on the axis $\{x=0\}$ and $\{y=0\}$.\\
In the absence of $\{2,2\}$, $\overline{p}^{11}(1,2)=\overline{p}^{12}(1,1)=1$. The random walk (\ref{appleffectif}) becomes:
\begin{align}
N_t^{11}= & N_0^{11}+M^{11}_t\nonumber\\
+ & \int_0^t \left(\frac{1}{2}\big(R(N^{12}_s,N_s) N^{11}_s +R(N^{11}_s,N_s) N^{12}_s \big)-d\times N^{11}_s\right)ds\nonumber\\
N_t^{12}= & N_0^{12}+ M^{12}_t\label{equationsdistyles}\\
+ & \int_0^t \left(\frac{1}{2}\big(R(N^{12}_s,N_s) N^{11}_s +R(N^{11}_s,N_s) N^{12}_s \big)-d\times N^{12}_s\right)ds.\nonumber
\end{align}

A particular significance is given to the escape time from the first quadrant:
\begin{equation}
\tau=\inf\{t\in \R_+,\, N^{11}_t=0 \mbox{ or }N^{12}_t=0\}.\label{def:tau}
\end{equation}On the set $\{\tau<+\infty\}$, the population goes extinct in finite time almost surely, whereas on the set $\{\tau=+\infty\}$ the population survives forever.

\section{Large populations of distylous flowers}\label{section:largepop}

We consider the Models 1-3 of Section \ref{sectionreproductionrate} for the distylous populations of Section \ref{section:distyle} and address the questions of determining when SSI are advantageous with respect to self-fertilization and inbreeding depression. The study is carried in the case where the population is large and the ODEs studied in Section \ref{sectionODE} are used. The case of small population is tackled in Section \ref{section:smallpop}.\\
In Sections \ref{stationarysol:wright} to \ref{stationarysol:fecundity}, we study the stationary solutions and their stabilities in the case of Models 1 to 3 (Sections \ref{section:wright} to \ref{section:fecundity}). Threshold phenomena and asymmetries showing Allee effects are highlighted. In Section \ref{section:comparaison}, comparisons between a SSI and a system with self-fertility and inbreeding depression are studied.\\

For Wright's (Model 1) and dependence (Model 2)  models, the ODEs are:
\begin{align*}
\left(\begin{array}{c}
n^{11}_t\\
n^{12}_t
\end{array}\right)=\left(\begin{array}{c}
n^{11}_0\\
n^{12}_0
\end{array}\right)+\int_0^t \left(\begin{array}{c}
\frac{1}{2}\big(r(n^{12}_s)n^{11}_s+r(n^{11}_s)n^{12}_s\big)\ind_{n_s^{11}>0}\ind_{n^{12}_s>0}-dn^{11}_s\\
\frac{1}{2}\big(r(n^{12}_s)n^{11}_s+r(n^{11}_s)n^{12}_s\big)\ind_{n_s^{11}>0}\ind_{n^{12}_s>0}-dn^{12}_s
\end{array}\right)ds,
\end{align*}where $r(.)=\bar{r}$ for Wright's model and $r(.)$ is defined in \eqref{functionrN} for the dependence model. This gives with classical arguments on the regularity of solutions:
\begin{align}
\frac{dn^{11}_t}{dt}= & \frac{1}{2}\big(r(n^{12}_t)n^{11}_t+r(n^{11}_t)n^{12}_t\big)\ind_{n_t^{11}>0}\ind_{n^{12}_t>0}-dn^{11}_t\nonumber\\
\frac{dn^{12}_t}{dt}= & \frac{1}{2}\big(r(n^{12}_t)n^{11}_t+r(n^{11}_t)n^{12}_t\big)\ind_{n_t^{11}>0}\ind_{n^{12}_t>0}-dn^{12}_t,\label{ODEdistyle}
\end{align}with the initial conditions $(n^{11}_0,n^{12}_0)$. The roles of $n^{11}$ and $n^{12}$ are symmetric. It is easy to see that when $n^{11}$ vanishes, the solutions remain on the boundary $\{n^{11}=0\}$. There is existence and uniqueness of the solutions on $(\R_+^*)^2=(\R_+\setminus\{0\})^2$, $\{0\}\times \R_+$ and $\R_+\times \{0\}$, from which we deduce the existence of a unique solution for every initial condition $(n^{11}_0,n^{12}_0)$.

\subsection{Wright's model (Model 1) in a large population}\label{stationarysol:wright}

The system (\ref{ODEdistyle}) without the indicators in the right hand sides (r.h.s.) and with $r(.)=\bar{r}$ becomes:
\begin{align}
\frac{d}{dt}\left(\begin{array}{c}n^{11}_t\\
n^{12}_t\end{array}\right)=A\left(\begin{array}{c}n^{11}_t\\
n^{12}_t\end{array}\right)\quad \mbox{ where }\quad A=\left(\begin{array}{cc}
\frac{\bar{r}}{2}-d & \frac{\bar{r}}{2}\\
\frac{\bar{r}}{2} & \frac{\bar{r}}{2}-d
\end{array}\right)=\frac{\bar{r}}{2}J-dI,\label{equationsolutionA}
\end{align}where $J$ is the square $2\times 2$-matrix filled with ones and $I$ is the identity $2\times 2$ matrix.

\begin{proposition}\label{propODErconstant}There exists a unique solution for (\ref{ODEdistyle}), which coincides with the solution of (\ref{equationsolutionA}) for every $t\in \R_+$:
\begin{align}
n^{11}_t= & \frac{1}{2}\Big(n_0^{11}\big(e^{(\bar{r}-d)t}+e^{-dt}\big)+n_0^{12}\big(e^{(\bar{r}-d)t}-e^{-dt}\big)\Big)\nonumber\\
n^{12}_t= & \frac{1}{2}\Big(n_0^{11}\big(e^{(\bar{r}-d)t}-e^{-dt}\big)+n_0^{12}\big(e^{(\bar{r}-d)t}+e^{-dt}\big)\Big).\label{solutioneqA}
\end{align}and:
\begin{equation*}
\lim_{t\rightarrow +\infty}e^{-(\bar{r}-d)t}n^{11}_t=\lim_{t\rightarrow +\infty}e^{-(\bar{r}-d)t}n^{12}_t=(n_0^{11}+n_0^{12})/2.
\end{equation*}
\end{proposition}

\begin{figure}[!ht]
\begin{center}
\begin{tabular}{ccc}
(a) & (b) & (c) \\
\includegraphics[width=0.28\textwidth,height=0.22\textheight,angle=0,trim=1cm 1cm 1cm 1cm]{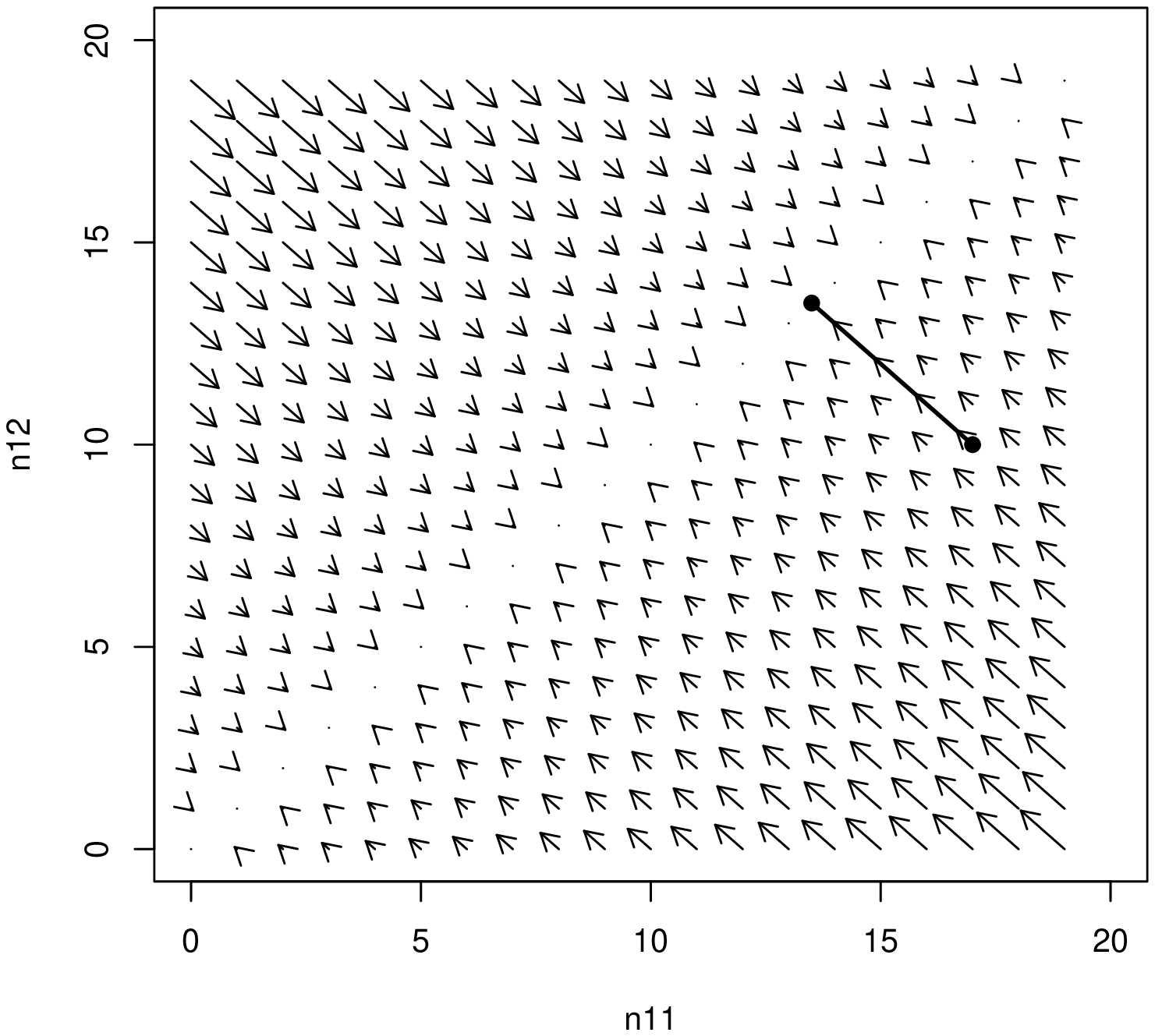}
 &
\includegraphics[width=0.28\textwidth,height=0.22\textheight,angle=0,trim=1cm 1cm 1cm 1cm]{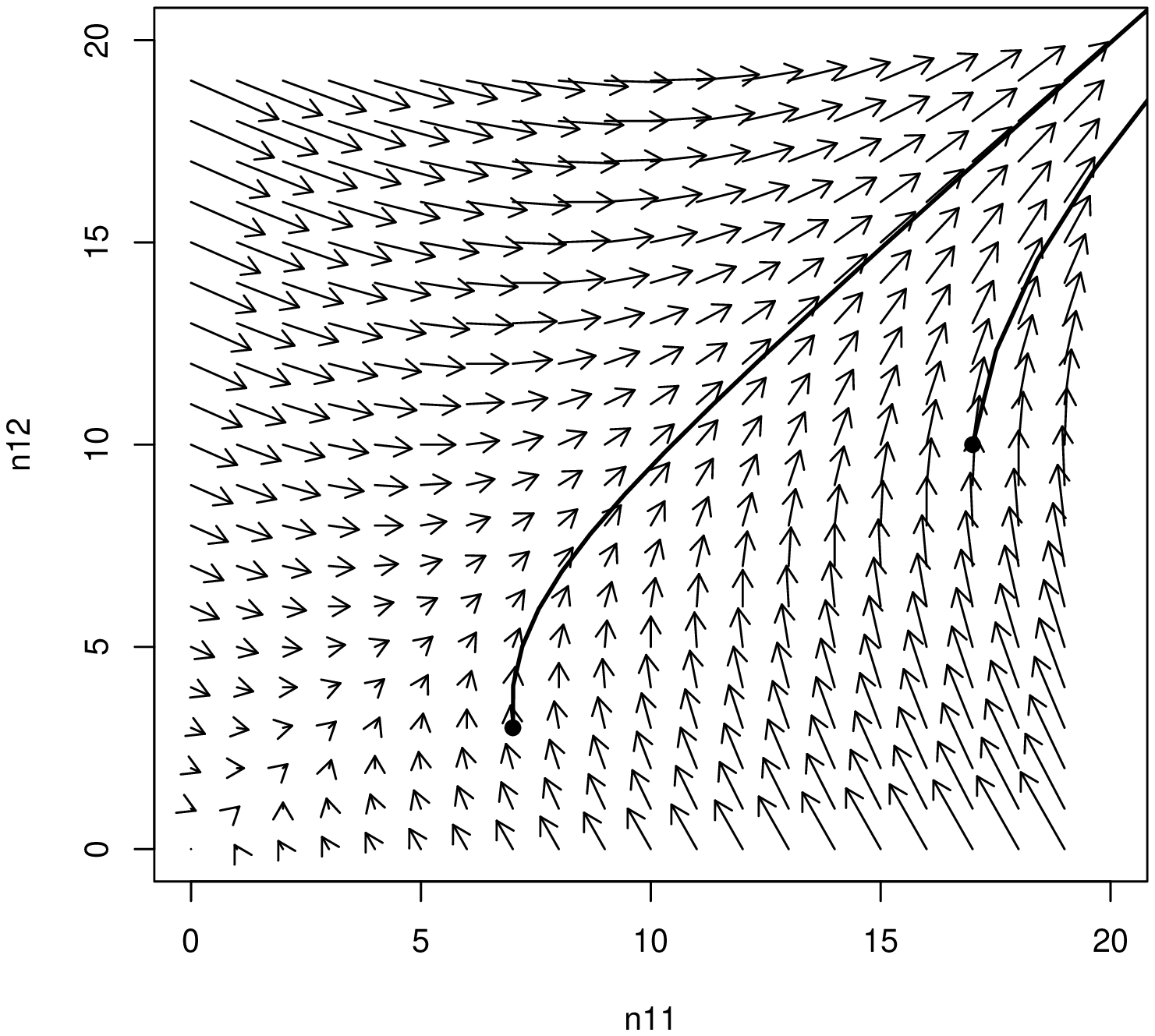}
 &
\includegraphics[width=0.28\textwidth,height=0.22\textheight,angle=0,trim=1cm 1cm 1cm 1cm]{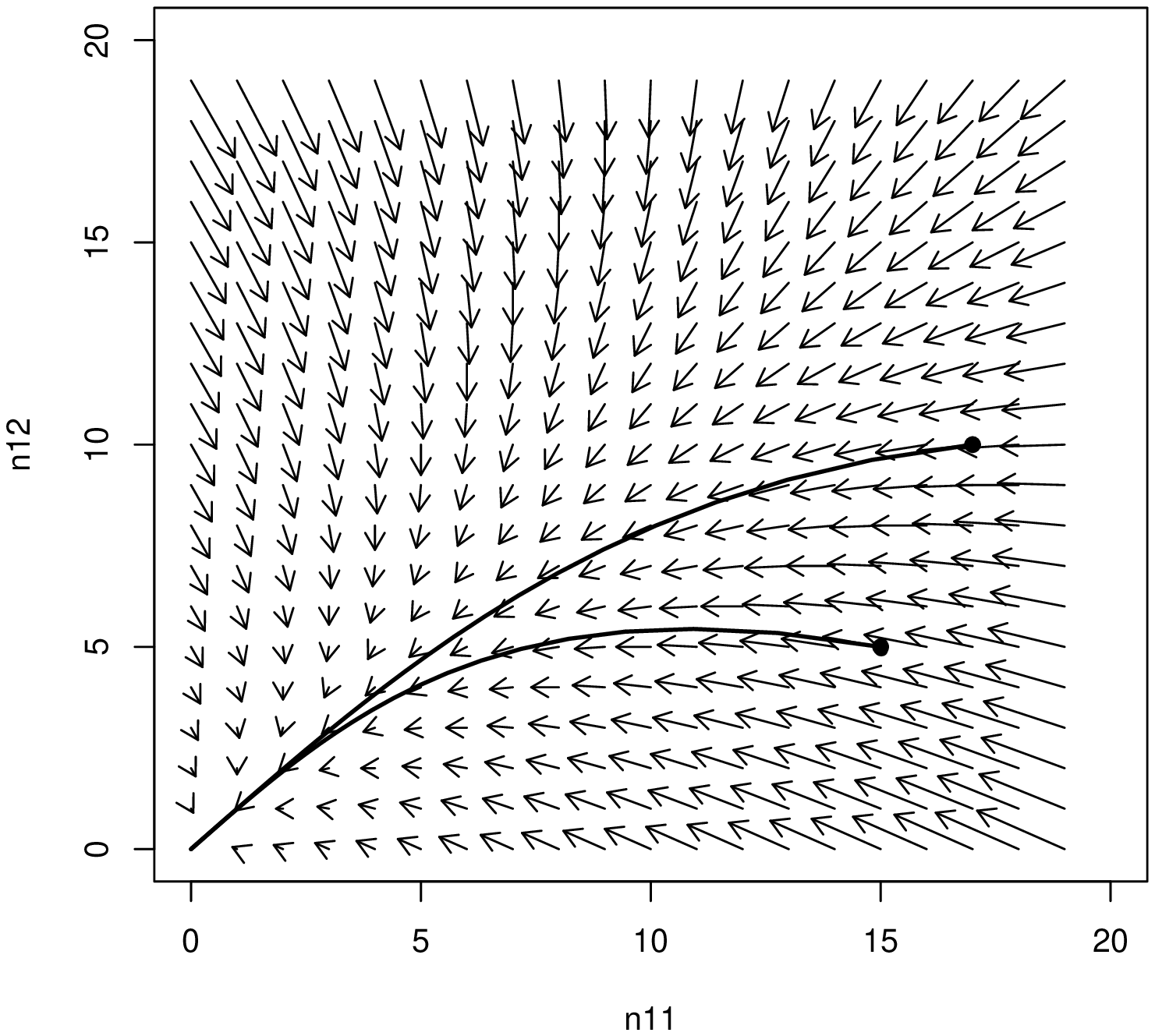}
\end{tabular}
\caption{\textit{Solutions $(n^{11}_t,n_t^{12})_{t\geq 0}$ of (\ref{ODEdistyle}) in the (a): critical ($\bar{r}=2$, $d=2$), (b): supercritical ($\bar{r}=4$, $d=3$) and (c): subcritical ($\bar{r}=2$, $d=3$) cases. }}\label{fig_solution_rconstant_distyle}
\end{center}
\end{figure}

\begin{corollary}We recover a natural disjunction between $\bar{r}>d$, $\bar{r}<d$ and $\bar{r}=d$. The first case is the supercritical case: the population survives and grows to infinite size. The second case is the subcritical case: there is asymptotic extinction of the population. In the third case, the critical case, the population size remains constant and the solution converges to a non-trivial equilibrium.
\end{corollary}

\begin{remark}
It is remarkable that in the large population approximation, for constant reproduction rates, the behavior is the same as for the case of compatible reproduction, in which the ODE for the population size is: $dn/dt=(\bar{r}-d)n.$ \qed
\end{remark}

\begin{proof}[Proof of Prop. \ref{propODErconstant}]
Existence and uniqueness of the solutions of (\ref{equationsolutionA}) in $\Co^{\infty}(\R_+,\R^2)$ hold by Cauchy-Lipschitz theorem. In every cases, the solutions of (\ref{ODEdistyle}) coincide with the solution of (\ref{equationsolutionA}) until one of the components equals zero. Let us denote by $t^0$ the time at which this happens and which may be infinite in case the trajectory of the solution of (\ref{equationsolutionA}) does not intersect the horizontal or vertical axes. This implies existence and uniqueness until the time $t^0$ of the solution of (\ref{ODEdistyle}).

\par Let us solve (\ref{equationsolutionA}). The matrix $A$ admits $r-d$ and $-d$ as eigenvalues respectively associated with the eigenvectors $(1/\sqrt{2},1/\sqrt{2})$ and $(1/\sqrt{2},-1/\sqrt{2})$. For the system (\ref{equationsolutionA}), there is a unique solution for every initial condition $(n^{11}_0,n^{12}_0)\in (0,+\infty)^2$ such that $t\mapsto (n^{11}_t,n^{12}_t)$ is of class $\Co^\infty$.
\begin{align}
\left(\begin{array}{c}
n^{11}_t\\
n^{12}_t
\end{array}\right)=\frac{1}{\sqrt{2}}\left(\begin{array}{cc}
1 &  1\\
1 & -1 \end{array}\right)\left(\begin{array}{c}
C_0 e^{(\bar{r}-d)t}\\
C_1 e^{-d.t}
\end{array}\right).\label{solutionA}
\end{align}The constants $C_0$ and $C_1$ are obtained from the initial condition. Solving (\ref{solutionA}) in $(C_0,C_1)$ for $t=0$:
\begin{equation}
C_0=\frac{\sqrt{2}}{2}\big(n_0^{11}+n_0^{12}\big),\qquad C_1=\frac{\sqrt{2}}{2}\big(n_0^{11}-n_0^{12}\big).
\end{equation}Hence this provides (\ref{solutioneqA}).

\par The question is whether $(\R_+^*)^2$ is positive invariant, \ie whether the trajectories of (\ref{solutioneqA}) remain in the positive quadrant.
\begin{align}
n^{11}_t=0 \quad \Leftrightarrow &\quad  n_0^{11}\big(e^{\bar{r}t}+1\big)+n_0^{12}\big(e^{\bar{r}t}-1\big)=0 \quad \Leftrightarrow \quad
e^{\bar{r}t}=\frac{n_0^{12}-n_0^{11}}{n_0^{11}+n_0^{12}}.
\end{align}This equation has no positive solution in $t$ since the right hand side is smaller than 1 (and possibly non positive).
Similar computation holds if one solves $n^{12}_t=0$. Hence, the solutions of (\ref{ODEdistyle}) and (\ref{equationsolutionA}) coincide on $\R_+$ and $t_0=+\infty$.

\par The long time behavior is obtained by noticing that whatever the case, the dominant factor in (\ref{solutioneqA}) is $\exp((\bar{r}-d)t)$.
\end{proof}

\subsection{Dependence model (Model 2) in a large population}\label{stationarysol:dependence}

We now turn to the case of a variable individual reproduction rate $r(.)$. First, we consider the general case of a bounded continuous nonnegative function $r(.)$, and then we focus on the functional form given in \eqref{functionrN}. The ODEs limits are given in \eqref{ODEdistyle}. Conclusions are summed up at the end of the subsection.\\

Let us start with a bounded continuous nonnegative function $r(.)$. The trivial solution $(0,0)$ is a stationary solution. We look for non trivial stationary solutions. Because of the symmetry in $n^{11}$ and $n^{12}$, if a nontrivial stationary solution exists, it satisfies $n^{11,*}=n^{12,*}:=n^*$. The latter value solves:
\begin{equation}
r(n^*)n^*=dn^*\qquad \Leftrightarrow \qquad r(n^*)=d.\label{solutionstationnaireODE}
\end{equation}The number of non trivial fixed points depends on the number of roots of (\ref{solutionstationnaireODE}).

\begin{remark}Notice that the total population size at equilibrium is then $2n^*$, which is twice the size at equilibrium in absence of SI. Indeed, in the latter case, the size of the population at equilibrium satisfies (\ref{solutionstationnaireODE}) and is thus $n^*$.\qed
\end{remark}

\noindent We now examine the stability of the stationary solutions. We refer to \citet{verhulstbook} for definitions and developments on the theory of dynamical systems.

\begin{proposition}\label{propstabilite00}The trivial equilibrium $(0,0)$ is:\\
(i) a positive attractor if $r(0)<d$,\\
(ii) a saddle point if $r(0)>d$.
\end{proposition}

\begin{proof}
Because of the indicators in (\ref{ODEdistyle}), we know that once one component has reached zero, it can not escape. We consider the stability of the trivial solution $(0,0)$ for the ODE (\ref{ODEdistyle}) without the indicators, as we know that before one of the components reaches zero, these systems have the same solutions. We use the classical linearization methods \cite[\eg][Chap. 3]{verhulstbook}. The linearization of the ODE without indicators around an equilibrium $(n^{11},n^{12})$ leads us to consider the Jacobian matrix of the system at this point:
\begin{align}
\mathfrak{J}(n^{11},n^{12})= & \left(\begin{array}{cc}
\frac{r(n^{12})}{2} + \frac{r'(n^{11})n^{12}}{2}-d & \frac{r'(n^{12})n^{11}}{2} + \frac{r(n^{11})}{2} \\
\frac{r(n^{12})}{2} + \frac{r'(n^{11})n^{12}}{2} & \frac{r'(n^{12})n^{11}}{2} + \frac{r(n^{11})}{2}-d
\end{array}\right)\label{matricejacobienne}
\end{align}
For the equilibrium $(0,0)$, with the notation $J$ and $I$ introduced after \eqref{equationsolutionA}:
\begin{equation}
\mathfrak{J}(0,0)=\left(\begin{array}{cc}
\frac{r(0)}{2} -d & \frac{r(0)}{2} \\
\frac{r(0)}{2}  & \frac{r(0)}{2}-d
\end{array}\right)=\frac{r(0)}{2}J-dI.\label{jacobienne(0,0)}
\end{equation}This matrix is the same as the matrix $A$ in (\ref{equationsolutionA}) and its eigenvalues are $r(0)-d$ and $-d$ respectively associated with the eigenvectors $(1/\sqrt{2},1/\sqrt{2})$ and $(1/\sqrt{2},-1/\sqrt{2})$. If $r(0)<d$, then both eigenvalues are negative and $(0,0)$ is a positive attractor for the system without indicators, and hence also for (\ref{ODEdistyle}). If $r(0)>d$ then there is a positive and a negative eigenvalue. In this case, $(0,0)$ is a saddle point for the system without indicators. This entails as in the proof of Prop \ref{propODErconstant} that $n^{11}_t+n^{12}_t$ converges to $+\infty$ while $n^{11}_t-n^{12}_t$ converges to zero. Thus, in the neighborhood of zero, when starting from points of the positive quadrant, the solutions are the same as for the system without indicator: there is no extinction.
\end{proof}


\begin{proposition}
For an equilibrium $(n^*,n^*)$ with $n^*>0$:\\
(i) If $r'(n^*)<0$, then the equilibrium is a positive attractor,\\
(ii) if $r'(n^*)>0$, then the equilibrium is a saddle point.
\end{proposition}

\begin{proof}
For the point $(n^*,n^*)$ the Jacobian matrix of (\ref{matricejacobienne}) becomes:
\begin{equation}
\mathfrak{J}(n^*,n^*)=
\Big(\frac{r(n^*)}{2} + \frac{r'(n^*)n^*}{2}\Big) \, J- d\,I.\label{etapestabilite}
\end{equation}This matrix is again of the same form as $A$ introduced in (\ref{equationsolutionA}). Its eigenvalues are $r(n^*) + r'(n^*)n^*-d$ and $-d$, respectively associated with the eigenvectors $(1/\sqrt{2},1/\sqrt{2})$ and $(1/\sqrt{2},-1/\sqrt{2})$.
As $r(n^*)=d$, we can simplify the expression of the first eigenvalue: $r(n^*) + r'(n^*)n^*-d=r'(n^*)n^*$, which is of the sign of $r'(n^*)$. As usual, the sign of the eigenvalues determines the nature of the equilibrium.
\end{proof}

Let us now focus on the particular form (\ref{functionrN}) for the reproduction rate. We compute the stationary solutions $(n^{*},n^{*})$ of (\ref{ODEdistyle}) by starting from (\ref{solutionstationnaireODE}):
\begin{align}
 r(n^*)=d \quad \Leftrightarrow \quad \overline{r}\frac{e^{\alpha n^*}}{e^{\alpha n^*}+\beta}=d
\quad \Leftrightarrow 
 \quad n^*=\frac{1}{\alpha}\log\Big(\frac{\beta d}{\overline{r}-d}\Big).\label{sol_stationnaire_ex3}
\end{align}Of course, we see that the log is well defined if and only if $\beta>0$ and $\overline{r}>d$. Moreover, $(n^*,n^*)$ belongs to the positive quadrant if and only if
\begin{align}
\frac{\beta d}{\overline{r}-d}>1 \qquad \Leftrightarrow \qquad \overline{r}<d(1+\beta).\label{condition}
\end{align}
\par Since $r(0)=\overline{r}/(1+\beta),$ we notice that the stability condition for the trivial equilibrium, $r(0)<d$, is equivalent to the condition (\ref{condition}) for the existence of a non-trivial stationary solution. Hence, if $r(0)<d$, $(0,0)$ is a positive attractor and there is no other equilibrium, and if $r(0)>d$, $(0,0)$ is a repulsive attractor. Indeed, $(0,0)$ is a saddle point for the system (\ref{ODEdistyle}) without the indicators. Since the stable manifold is the vectorial line of direction $(1,-1)$ which intersects the positive quadrant only at $(0,0)$, then for the system (\ref{ODEdistyle}), the equilibrium $(0,0)$ is a negative attractor.
\par Moreover, the equilibrium $(n^*,n^*)$ is always a saddle point as:
\begin{align}
r'(n)=\overline{r}\frac{\alpha e^{\alpha n}(\beta+e^{\alpha n})-\alpha e^{2\alpha n}}{(\beta+e^{\alpha n})^2}=\frac{\overline{r}\alpha \beta e^{\alpha n}}{(\beta+e^{\alpha n})^2},\label{calcul_r'}
\end{align}is always positive. The stable (resp. unstable) manifold is locally the affine line of direction $(1,-1)$ (resp. the affine line of direction $(1,1)$).
\par In conclusion:
\begin{itemize}
\item If $\bar{r}<d$, then every trajectory converges to $(0,0)$.
\item If $d<\bar{r}<d(1+\beta)$, then there exists a non-trivial equilibrium that is a saddle point. Trajectories converge to $(0,0)$ or $\lim_{t\rightarrow +\infty}n^{11}_t=\lim_{t\rightarrow +\infty}n^{12}_t=+\infty$.
\item If $\bar{r}>d(1+\beta)$, then there is no non-trivial equilibrium in the positive quadrant as the growth rate is too strong. $(0,0)$ is a negative attractor and $\lim_{t\rightarrow +\infty}n^{11}_t=\lim_{t\rightarrow +\infty}n^{12}_t=+\infty$.
\end{itemize}

\begin{figure}[!ht]
\begin{center}
\begin{tabular}{ccc}
(a) & (b) & (c) \\
\includegraphics[width=0.28\textwidth,height=0.22\textheight,angle=0,trim=1cm 1cm 1cm 1cm]{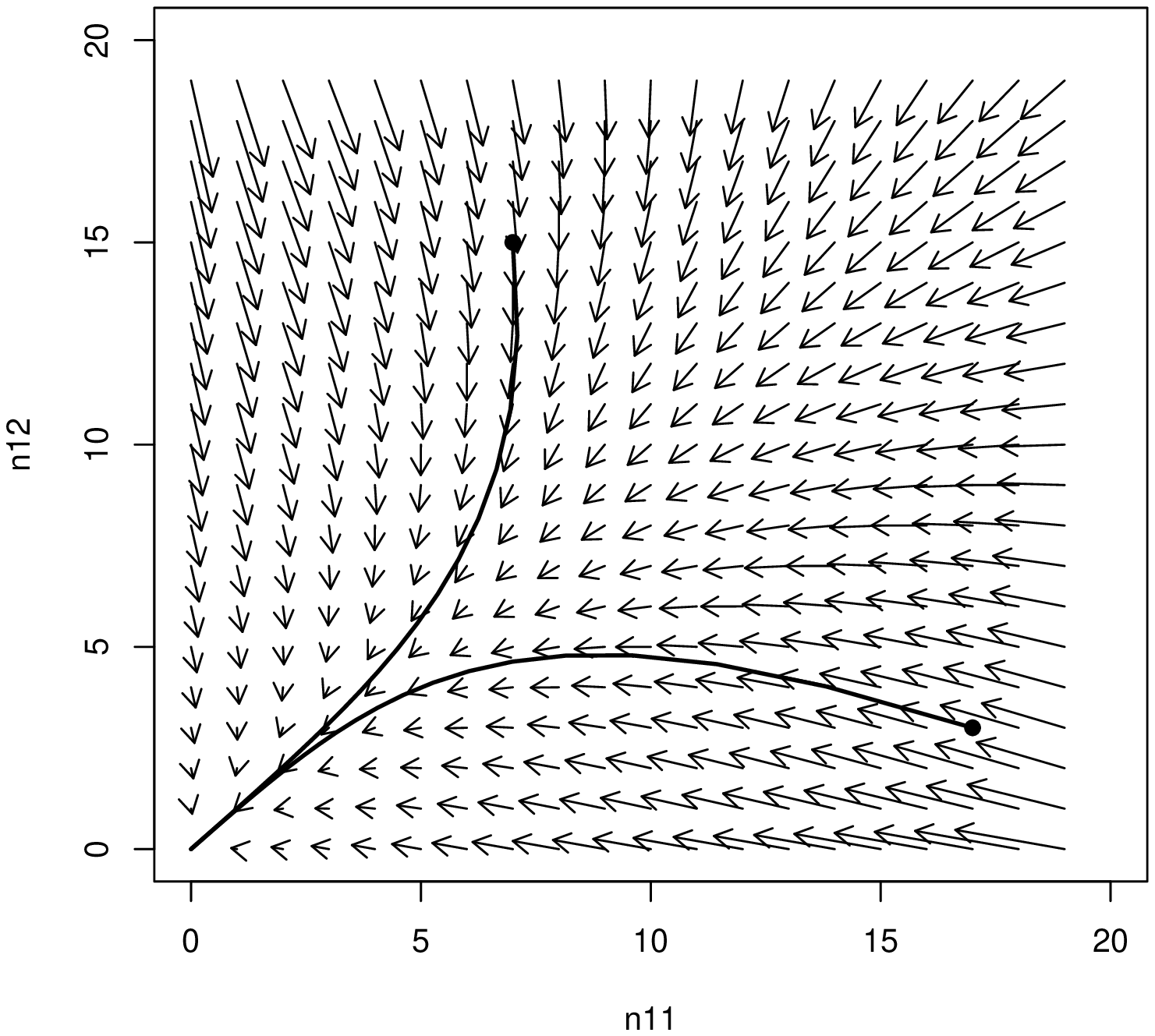}
 &
\includegraphics[width=0.28\textwidth,height=0.22\textheight,angle=0,trim=1cm 1cm 1cm 1cm]{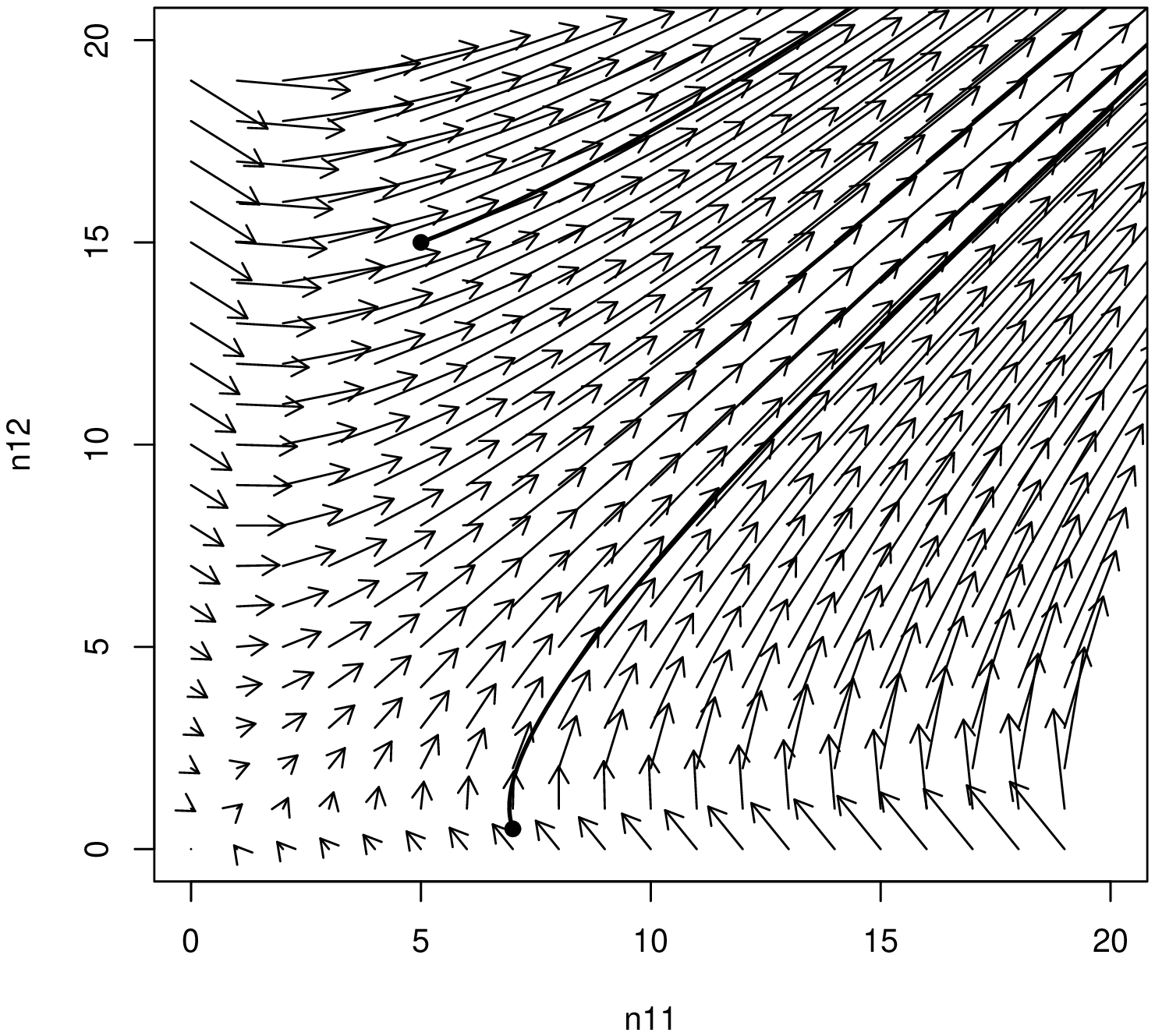}
&
\includegraphics[width=0.28\textwidth,height=0.22\textheight,angle=0,trim=1cm 1cm 1cm 1cm]{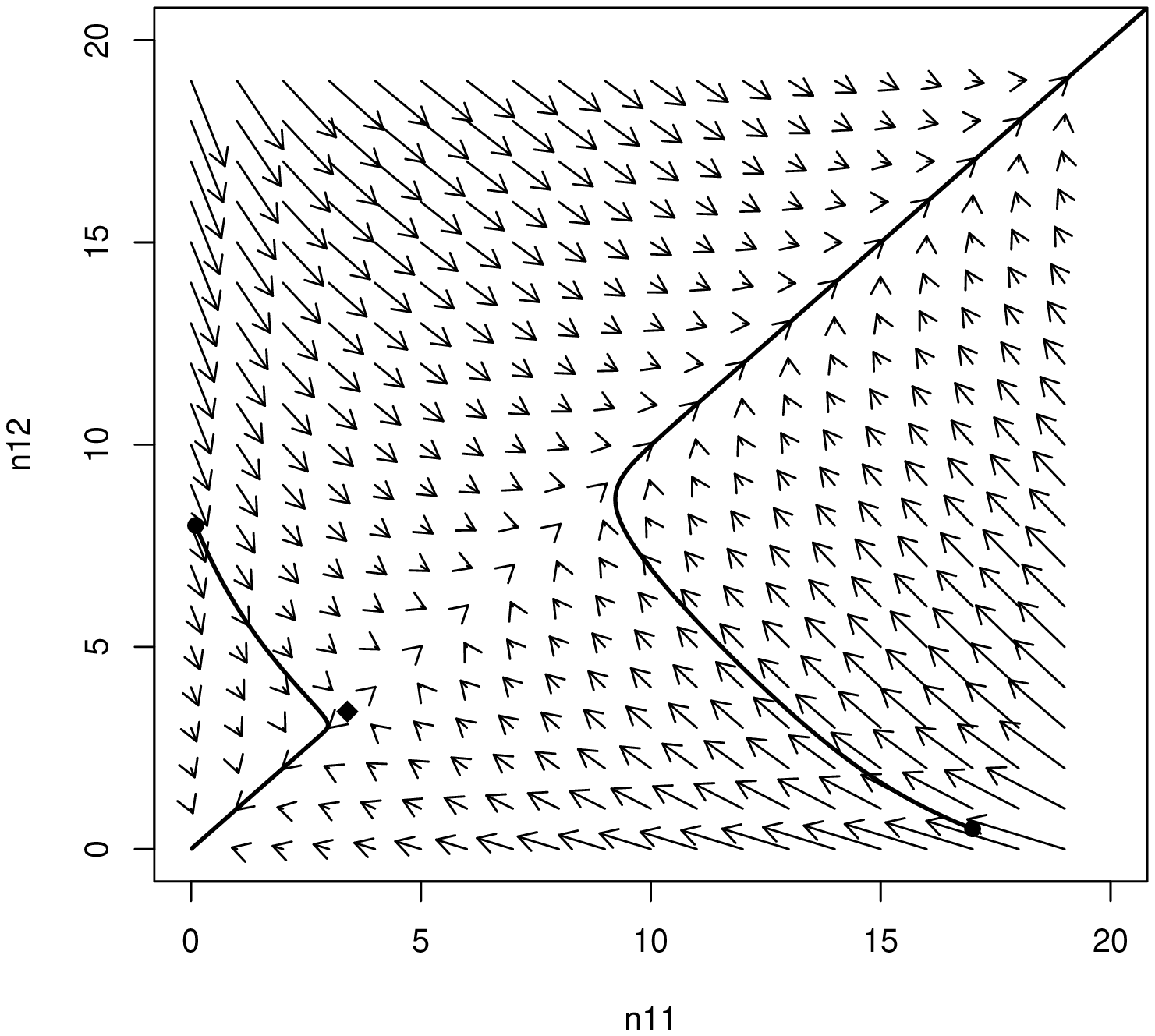}
\end{tabular}
\caption{\textit{Solutions $(n^{11}_t,n_t^{12})_{t\geq 0}$ of (\ref{ODEdistyle}) for a reproduction rate of the form (\ref{functionrN}) with $\alpha=\beta=1$. (a): $\overline{r}=2$, $d=3$, (b): $\overline{r}=7$, $d=3$. We have $n^*=-0.3$ and this equilibrium is not in the positive quadrant. $r(0)=3.5>d$ and $(0,0)$ is a negative attractor. Every solution tends to infinity. (c): $\overline{r}=3.1$, $d=3$. We obtain $n^*=3.4$. Since $r(0)=1.55$ is smaller than $d=3$, $(0,0)$ is a positive attractor in this case. The solutions hence either converge to zero or to infinity, given the initial condition. It is also seen that at the neighborhood of $2n^*$, for similar initial population sizes, populations with higher symmetry are more likely to survive than very asymmetric populations. }}\label{fig_solution_rvariable_distyle}
\end{center}
\end{figure}

\begin{remark}
Threshold phenomena appear: when $d<\bar{r}<d(1+\beta)$, we see that contrarily to the case of absence of pollen limitation (Wright's model), there may be either survival or extinction. If the size of the population is too small (if the initial condition belongs to the attracting domain of $(0,0)$) then there is extinction. There is a threshold implying that survival is possible only for sufficiently large population. \\
 On Fig. \ref{fig_solution_rvariable_distyle}, we see that the minimum size to ensure survival also depends on the composition of the population. It is smaller for population with $1/2$ individuals of genotype $\{1,1\}$ and $1/2$ individuals of genotype $\{1,2\}$. If we draw the affine line of direction $(1,-1)$ going through $(n^*,n^*)$ on Fig. \ref{fig_solution_rvariable_distyle} (c), it can be seen for a given initial condition $n_0=n^{11}_0+n^{12}_0$ close to $2n^*$ that the trajectories either lead to $+\infty$ or $0$ when $t\rightarrow +\infty$.
Thus, the symmetry in mating partners matters for survival or extinction (Fig. \ref{fig_solution_rvariable_distyle} (c)). To ensure survival whatever the initial condition, one needs a reproduction rate sufficiently large ($\bar{r}=d(1+\beta)$) to totally compensate the pollen limitation.\\
Both phenomena can be interpreted as Allee effects: due to pollen limitation, the population growth rate is an increasing function of the compatible population size and thresholds appear when the function is negative at the beginning.
\qed
\end{remark}

\subsection{Fecundity selection model (Model 3) in a large population}\label{stationarysol:fecundity}

We finally study Model 3. The rate of production of genotypes $\{1,1\}$ is $
\bar{r}\frac{N^{11}_tN^{12}_t}{N_t},$ which does not present any discontinuity on the boundaries $\{N^{11}_t=0\}$ and $\{N^{12}_t=0\}$ any more. The ODE approximation is now:
\begin{align}
\frac{dn^{11}_t}{dt}= & \bar{r}\frac{n^{11}_tn^{12}_t}{n^{11}_t+n^{12}_t}-d n^{11}_t,\qquad \qquad
\frac{dn^{12}_t}{dt}=  \bar{r}\frac{n^{11}_tn^{12}_t}{n^{11}_t+n^{12}_t}-dn^{12}_t,\label{ODEdistyle_fs}
\end{align}which admits a unique solution in $(\R_+^*)^2$.

\begin{proposition}
As for Wright's model, we have a disjunction into three cases:\\
(i) Subcritical case $\bar{r}<2d$: $\lim_{t\rightarrow +\infty}n^{11}_t=\lim_{t\rightarrow +\infty}n^{12}_t=0$.\\
(ii) Critical case $\bar{r}=2d$: the size $n^{11}_t+n^{12}_t$ remains constant and $(n^{11}_t,n^{12}_t)_{t\in \R_+}$ converges to an equilibrium $(n^*,n^*)\in (\R_+^*)^2$.\\
(iii) Supercritical case $\bar{r}>2d$: $\lim_{t\rightarrow +\infty}n^{11}_t=\lim_{t\rightarrow +\infty}n^{12}_t=+\infty$.\\
The additional coefficient 2 in the criteria comes from the fact that at equilibrium, the probability of mating success of an ovule is $1/2$ contrary to the case of (\ref{solutionstationnaireODE}) where it is 1.
\end{proposition}

\begin{proof}
By symmetry, a non-trivial equilibrim exists if and only if:
\begin{equation}
\bar{r}n^*=2dn^* \qquad \Leftrightarrow \qquad \bar{r}=2d.\label{solutionstationnaireODE_fs}
\end{equation}
To obtain the stability of the non-trivial equilibria, we compute the Jacobian matrix of the system (\ref{ODEdistyle_fs}) at $(n^{11},n^{12})$:
\begin{equation}
\mathfrak{J}(n^{11},n^{12})
=  \left(\begin{array}{cc}
\bar{r}\Big(\frac{n^{12}}{n^{11}+n^{12}}\Big)^2-d &
\bar{r}\Big(\frac{n^{11}}{n^{11}+n^{12}}\Big)^2\\
 \bar{r}\Big(\frac{n^{12}}{n^{11}+n^{12}}\Big)^2 &
 \bar{r}\Big(\frac{n^{11}}{n^{11}+n^{12}}\Big)^2-d
\end{array}\right)\label{matricejacobienne_fs}
\end{equation}For an equilibrium $(n^*,n^*)$, $\mathfrak{J}(n^*,n^*)=  \frac{\bar{r}}{4} J -dI,$ and we are led to a discussion similar to (\ref{etapestabilite}). A difficulty arises at the trivial equilibrium $(0,0)$ since the limit of $(n^{12}/(n^{11}+n^{12}))^2$ at $(0,0)$ in $\mathfrak{J}(0,0)$ is not defined.
To remedy for this, we consider the ODEs satisfied by the total size $n_t=n^{11}_t+n^{12}_t$ and by $z_t=n^{11}_t \times n^{12}_t$:
\begin{align}
\frac{dn_t}{dt}= & 2\bar{r}\frac{n^{11}_tn^{12}_t}{n^{11}_t+n^{12}_t}-d n_t,\qquad
\frac{dz_t}{dt}=  \frac{dn^{11}_t}{dt}n^{12}_t+n^{11}_t\frac{n^{12}_t}{dt}=  \Big(\bar{r}-2d\Big)z_t.
\end{align}As a consequence, when $\bar{r}<2d$, there is asymptotic extinction of at least one of the genotypes $\{1,1\}$ or $\{1,2\}$, let us assume that it is $n^{11}$. Since:
\begin{equation}
\frac{dn^{12}_t}{dt}\leq \bar{r}n^{11}_t-d\,n^{12}_t \leq \varepsilon-d\, n^{12}_t
\end{equation}for all $\varepsilon>0$ and for sufficiently large times, $n^{12}$ also converges to zero. \\
If $\bar{r}>2d$,$$\lim_{t\rightarrow +\infty}n^{11}_t \times n^{12}_t=+\infty.$$One at least of the two subpopulation sizes $\{1,1\}$ or $\{1,2\}$ tends to infinity. Let us assume that it is $n^{11}$ and that $n^{11}_t\geq n^{12}_t$ for sufficiently large $t$ (implying that $n^{11}_t/n_t\geq 1/2$). Writing:
\begin{equation}
\frac{dn^{12}_t}{dt}\geq \frac{\bar{r}}{2}n^{12}_t-d\,n^{12}_t,
\end{equation}we see that $n^{12}$ also necessarily tends to infinity.\\
Finally, when $\bar{r}=2d$. If $n^{11}_t<n^{12}_t$, then
$n^{11}_t/2<n^{11}_tn^{12}_t/n_t<n^{12}_t/2$, and
\begin{align*}
 \frac{dn^{11}_t}{dt}> \Big(\frac{\bar{r}}{2}-d\Big)n^{11}_t=0\qquad \mbox{ and }\qquad
 \frac{dn^{12}_t}{dt}< \Big(\frac{\bar{r}}{2}-d\Big)n^{12}_t=0.
\end{align*}Thus $n^{11}$ increases and $n^{12}$ decreases until they reach a state where $n^{11}=n^{12}$.
\end{proof}

\subsection{Comparison with self-compatibility with inbreeding depression (SCID)}\label{section:comparaison}

Now that our three models of interest have been considered, we wish to carry out comparisons with the SCID cases. Under which conditions do our models predict that SSI are advantageous (in some sense defined in the sequel) with respect to self-fertilization ?

\subsubsection{Self-compatible distylous population with inbreeding depression}\label{section:SC}

We consider here reproduction rates of the form (\ref{forme_r(Ncomp)_inbreedingdepr}). In the case of self-compatible populations with inbreeding depression, the size $n_t$ of the population in a large population limit follows the ODE:
\begin{equation}
\frac{dn_t}{dt}=\Big[(1-\delta)s \bar{r}+(1-s)\bar{r}\frac{e^{\alpha n_t}}{\beta+e^{\alpha n_t}}-d\Big]n_t.\label{equation_inbreeding}
\end{equation}

The trivial solution defines an equilibrium and:
\begin{proposition}
A non-trivial equilibrium to \eqref{equation_inbreeding} exists if and only if
\begin{equation}
\beta\big[d-\overline{r}s(1-\delta)\big]>\big[\overline{r}(1-\delta s)-d\big]>0.\label{condition3}
\end{equation}
(i) If (\ref{condition3}) is satisfied, the non-trivial equilibrium $n^*$ is given by:
 \begin{equation}n^*=\frac{1}{\alpha}\log\Big(\frac{\beta \big(d-(1-\delta)s \bar{r}\big)}{\overline{r}(1-\delta s)-d}\Big)\label{equilibre_inbreeding}\end{equation}This equilibrium is repulsive and $0$ is a locally positive attractor. If $n_0<n^*$ then the solution converges to zero, else, it grows to infinity.\\
(ii) Else, there is no non-trivial equilibrium in the positive quadrant and there are three possibilities:
  \begin{itemize}
  \item If $\bar{r}(1-\delta s)-d<0 ~ \Leftrightarrow ~ \bar{r}(1-\delta s)<d$, then there is extinction ($d$ is too large).
  \item If $ \beta[d-\bar{r}s(1-\delta)]<0 ~ \Leftrightarrow ~ \bar{r}s(1-\delta)>d$, or if $\big[\overline{r}(1-\delta s)-d\big]>\beta\big[d-\overline{r}s(1-\delta)\big]>0$, then there is growth to infinity.
\end{itemize}
\end{proposition}

\begin{proof}
A non-trivial equilibrium $n^*$ of (\ref{equation_inbreeding}) satisfies necessarily:
\begin{align}
(1-s)\bar{r}\frac{e^{\alpha n^*}}{\beta+e^{\alpha n^*}}=d-(1-\delta)s \bar{r}\quad \Leftrightarrow &\quad n^*=\frac{1}{\alpha}\log\Big(\frac{\beta \big(d-(1-\delta)s \bar{r}\big)}{(1-s)\bar{r}-d+(1-\delta)s \bar{r}}\Big)\nonumber
\end{align}which provides (\ref{equilibre_inbreeding}) if the log is well defined:
\begin{equation}
(1-\delta)s\overline{r}<d<(1-\delta s)\overline{r}.
\label{condition_existence}
\end{equation}By comparison arguments, we can prove that for $d>(1-\delta s)\overline{r}$, every solution converges to 0 as the square bracket in the r.h.s. of (\ref{equation_inbreeding}) is strictly negative. For $d<(1-\delta)s\overline{r}$, this bracket remains strictly positive and the solutions converge to $+\infty$. Under the condition (\ref{condition_existence}), $n^*$ is strictly positive if and only if
\begin{align}
\beta(d-(1-\delta)s\overline{r})>(1-\delta s)\overline{r}-d ~ \Leftrightarrow ~ & (1+\beta)d>\overline{r}\big(1-\delta s +\beta\, s(1-\delta) \big)\nonumber\\
\Leftrightarrow ~  & \beta\big[d-\overline{r}s(1-\delta)\big]>\big[\overline{r}(1-\delta s)-d\big].\label{condition2}
\end{align}When \eqref{condition_existence} is satisfied, the brackets in (\ref{condition2}) are positive. Equation \eqref{condition2} says whether the parameter $d$ in (\ref{condition_existence}) is closer to the lower bound of \eqref{condition_existence} (\ie $0<\beta(d-\bar{r}s(1-\delta))<\bar{r}(1-\delta s)-d$, implying growth of the population size to infinity) or of its upper bound (\ie \eqref{condition3}, under which there exists a non-trivial equilibrium). Hence, we obtain \eqref{condition3} as sufficient and necessary condition for the existence of a non-trivial equilibrium in the positive quarter plane.\\

 With \eqref{condition3}, the situation is similar to (\ref{sol_stationnaire_ex3}) with $d-(1-\delta)s\overline{r}$ instead of $d$ and $(1-s)\overline{r}$ instead of $\overline{r}$. Self-fertilization amounts to a reduction of the natural mortality since individuals can at least mate with themselves. However, this introduces a limitation of the maximal number of individuals produced by outcrossings. \\

To study the stability of the equilibria $0$ and $n^*$, we linearize the system and repeat the computation of (\ref{jacobienne(0,0)}) and (\ref{etapestabilite}). The stability of $n^*$ depends on the sign of $r(n^*)+r'(n^*)n^*-d=r'(n^*)n^*$ which is here always positive by (\ref{calcul_r'}) and by the remark of the preceeding paragraph. Hence, $n^*$ is a negative attractor. The trivial equilibrium $0$ is a positive attractor if and only if
\begin{align*}
r(0)<d ~ \Leftrightarrow ~ & (1-\delta)s\overline{r}+\frac{(1-s)\overline{r}}{1+\beta}<d ~ \Leftrightarrow ~ \big(1-\delta s +\beta s (1-\delta)\big)\overline{r}<d(1+\beta).
\end{align*}We recognize here the condition (\ref{condition2}) of existence of a non-trivial equilibrium, which leads to the announced result (i). In the case where \eqref{condition3} is not fulfilled, the behavior of the solutions is obtained by comparisons with simple ODEs.
\end{proof}

\subsubsection{Comparison with the self-incompatible case}

It is natural to compare the self-compatible model of Section \ref{section:SC} with the SI distylous model of Section \ref{stationarysol:dependence}, with the same reproduction rate \eqref{functionrN}. To carry the comparison, we introduce the following criteria:
\begin{itemize}
 \item We compare, when they exist, the population sizes at the non-trivial equilibria. Heuristically, these sizes provide the limit between the two extreme behaviors that are growth to infinity and extinction. When this equilibrium size is high, large populations are needed to avoid extinction and the population is more ``fragile''. The corresponding model will be said to be less advantageous.
\item We can also compare the range of parameters $\bar{r}$ and $d$ (for fixed $s$ and $\delta$) for which the population goes extinct. We will say that with respect to this second criterion, the more advantageaous conditions correspond to the smaller ranges of such parameters.
\end{itemize}

\noindent We sum up the results of Sections \ref{stationarysol:dependence} and \ref{section:SC} in Table \ref{tab1}.\\
\begin{table}[!ht]
\begin{center}
\begin{tabular}{|l|l|l|}
\hline
Case & Case description & Behavior of the population size\\
\hline
\multicolumn{3}{|c|}{Wright's model (Model 1)}\\
\hline
(a.1) & $\bar{r}<d$ & Convergence to 0\\
(b.1) & $\bar{r}=d$ & Convergence to a non-trivial fixed point\\
(c.1) & $\bar{r}>d$ & Divergence to $+\infty$\\
\hline
\multicolumn{3}{|c|}{Dependence model (Model 2)}\\
\hline
(a.2) & $\bar{r}<d$ & Convergence to 0\\
(b.2) & $d<\bar{r}<d(1+\beta)$ & Existence of a saddle point\\
(c.2) & $\bar{r}>d(1+\beta)$ & Divergence to $+\infty$\\
\hline
\multicolumn{3}{|c|}{Fecundity selection model (Model 3)}\\
\hline
(a.3) & $\bar{r}<2d$ & Convergence to 0\\
(b.3) & $\bar{r}=2d$ & Convergence to a non-trivial fixed point\\
(c.3) & $\bar{r}>2d$ & Divergence to $+\infty$\\
\hline
\multicolumn{3}{|c|}{Self-compatible model without inbreeding depression (Model 4)}\\
\hline
(a.4) & $\bar{r}<d$ & Convergence to 0\\
(b.4) & $\bar{r}=d$ & Convergence to a non-trivial fixed point\\
(c.4) & $\bar{r}>d$ & Divergence to $+\infty$\\
\hline
\multicolumn{3}{|c|}{Self-compatible with inbreeding depression (SCID) model (Model 5)}\\
\hline
(a.5) & $(1-\delta s) \bar{r}<d$  & Convergence to 0\\
(b.5) & $\beta(d-\bar{r}s(1-\delta))>\bar{r}(1-\delta s)-d>0$ & Existence of a repulsive equilibrium \\
(c.5) & $\bar{r}>d \min\big(\frac{1}{(1-\delta)s},\frac{1+\beta}{1+\beta s -s\delta (1+\beta)}\big)$ & Divergence to $+\infty$\\
\hline
\end{tabular}
\caption{{\small \textit{Summary of the behavior of the population size in the Sections \ref{stationarysol:dependence} and \ref{section:SC} depending on the respective values of $\bar{r}$, $d$, $s$ and $\delta$. The condition for Case (c.5) is obtained by saying that we have divergence to $+\infty$ in the SCID model if $(1-\delta)s \bar{r}>d$ or if $(1-\delta)s \bar{r}<d<\bar{r}s(1-\delta)$ and $\bar{r}(1-\delta s)-d>\beta(d-\bar{r}s(1-\delta))>0$.}}}\label{tab1}
\end{center}
\end{table}

Models 1, 3 and 4 are similar. We have convergence to zero or divergence to infinity except in the particular cases when $\bar{r}=d$ or $\bar{r}=2d$. This shows that under fecundity selection, distylous species are more fragile than under Wright's model or self-compatibility without inbreeding depression. This is expected since there is no pollen limitation or self-incompatibility in the two latter models. Notice also that in absence of pollen limitation, the equilibria and critical, sub- and supercritical regions correspond in Models 1 and 4.\\
In the sequel, we compare Models 2 and 5, for which there exists a range of parameters for which extinction or divergence to infinity coincide, depending on the initial condition. In both cases, the pollen limitation is modeled similarly (see \eqref{functionrN}). The difference relies on the penalization by SI in the first case and by inbreeding depression in the second case. This comparison provides conditions on the parameters under which distyly is advantageous on self-fertilization.

\paragraph{Comparison with respect to the range of parameters}

The cases (a.2), \dots, (c.2), (a.5), \dots (c.5) are defined in Table \ref{tab1}.
Let us consider the set of parameters for which the population goes extinct whatever its initial condition (Cases (a.2) and (a.5)). Since $\delta>0$ and $s>0$, we see that this region is larger for the SCID model which in this respect appears as less avantadgeous than the SI model. The parameter $(1-\delta s)$ in Case (a.5) can indeed be interpreted as the proportion of seeds which survived, \ie when excluding the fraction of non-viable seeds produced by self-fertilization. This term thus appears as an extra death parameter that is not present in the SI model. In our large population setting, this penalty is more important than the loss of partners that may face individuals in the SI model. \\

We now turn to the Cases (c.2) and (c.5) where the population size diverges to infinity. The SI model is advantageous on the SCID model if and only if:
\begin{align*}
 & (1+\beta)\leq \min\big(\frac{1}{(1-\delta)s},\frac{1+\beta}{1+\beta s -s\delta (1+\beta)}\big)\\
\Leftrightarrow \quad & \frac{1}{1+\beta}\geq \min\big(s(1-\delta),1-\delta\big)=s(1-\delta).
\end{align*}This condition shows that for low self-fertilization efficiency, \ie for large inbreeding depression, it is less advantageous than SI. The latter efficiency is expressed by comparing the fraction of viable seeds produced by self-fertilization to the initial fraction $1/(1+\beta)=r(0)/\bar{r}$ of seeds produced without self-fertilization when only an infinitesimal quantity of compatible individuals is present.

\paragraph{Comparison with respect to the sizes of the population at equilibrium}

Let us now consider the cases when both (b.2) and (b.5) are satisfied. When there exists a non-trivial equilibrium, we have seen that the behavior of the solution is determined by its initial condition. The size of the equilibrium provides an idea of how many individuals are necessary to allow survival, even if in cases as in \eqref{sol_stationnaire_ex3}, the symmetry of the initial condition may matter. Let us thus compare the sizes at equilibrium in (\ref{sol_stationnaire_ex3}) and (\ref{equilibre_inbreeding}). The size of the population at equilibrium in the SCID model is equal to the size of the population in the SI model when:
\begin{align}
 \frac{1}{\alpha}\log\Big(\frac{\beta \big(d-(1-\delta)s \bar{r}\big)}{(1-\delta s)\bar{r}-d}\Big) =  \frac{2}{\alpha}\log\Big(\frac{\beta d}{\bar{r}-d}\Big)~ \Leftrightarrow ~ & \frac{\beta \big(d-(1-\delta)s \bar{r}\big)}{(1-\delta s)\bar{r}-d}= \Big(\frac{\beta d}{\bar{r}-d}\Big)^2.\label{comparaison}
\end{align}Let us study:
\begin{equation}
f(s,\delta)=\frac{\beta d - \beta \overline{r}(1-\delta) s}{(\overline{r}-d)-\overline{r}\delta s}.
\end{equation}Notice that $s=0$ corresponds to the self-compatible case without inbreeding depression: $f(0,\delta)=\beta d/(\overline{r}-d)$. The case $s=1$ provides a model of compatible population with reproduction rate $\bar{r}(1-\delta)$ and we have $f(1,\delta)=-\beta$. For any given $\delta\in [0,1]$, $f(.,\delta)$ is a rational function of $s$, defined on $[0,1]\setminus \{s_0\}$ with $s_0:=(\overline{r}-d)/\overline{r}\delta$. The latter value belongs to $[0,1]$ if and only if
\begin{equation}
d<\overline{r}<\frac{d}{1-\delta}.
\end{equation}
On the domain $[0,1]\setminus \{s_0\}$:
\begin{equation}
\partial_s f(s,\delta)=-\frac{\beta \overline{r}(\overline{r}(1-\delta)-d)}{((\overline{r}-d)-\overline{r}\delta s)^2},
\end{equation}which has the same sign as $d-\overline{r}(1-\delta)$.
\par In a nutshell,
\begin{itemize}
\item If $d<\overline{r}(1-\delta)$, the $s\mapsto f(s,\delta)$ is a decreasing continuous function on $[0,1]$, bounded above by $f(0,\delta)=\beta d/(\overline{r}-d)<(\beta d/(\overline{r}-d))^2$ (which is larger than 1 in Case (b.2)). In this case, there is no solution to (\ref{comparaison}). The size at equilibrium of the model with self-fertilization is always lower than its counterpart with self-incompatible reproduction: the cost of self-incompatibility may not be counterbalanced in this case.
\item If $d>\overline{r}(1-\delta)$, then $f(.,\delta)$ is an increasing function from $[0,s_0)$ into $[f(0,\delta),+\infty)$ and from $(s_0,1]$ to $(-\infty,-\beta]$. Thus, $s_0$ is the highest fraction of self-fertilization that allows the existence of an equilibrium (\ref{equilibre_inbreeding}) and we will thus only consider $s<s_0$. There exists a unique solution $r_0\in [0,s_0)$ to \eqref{comparaison} given by:
$$r_0=\frac{\Big(\frac{\beta d}{\bar{r}-d}\Big)^2(\bar{r}-d)-\beta d}{\bar{r}\Big(\delta \Big(\frac{\beta d}{\bar{r}-d}\Big)^2-\beta(1-\delta)\Big)}.$$
For any $s\in (r_0,s_0)$, the size of the equilibrium in the self-incompatible case is smaller than its counterpart with possible self-inbreeding. This means that if the fraction of offspring produced by self-fertilization is too high, then it is advantageous to switch to self-incompatibility.
\end{itemize}

\section{Small distylous populations}\label{section:smallpop}

We now consider cases where the deterministic approximation may not be taken and stick with random walks. In this section, we focus on the cases with constant rate of ovule production (Wright's model (Model 1), the fecundity selection model (Model 3) and their counterpart with possible self-fertilization (Model 4)) for which computations are tractable. Simulations for the other models are carried in Section \ref{section:simulations}. We are interested in the study of extinction probabilities. A difficulty comes from the fact that the random walks are inhomogeneous. After providing an equation for the extinction probabilities, we proceed by couplings to obtain approximations of these quantities. We construct processes defined on the same probability space as the \textit{original process} $(N^{11}_t,N^{12}_t)_{t\in \R_+}$. These processes will be called \textit{auxiliary processes}. We look for couplings such that the extinction probabilities for the auxiliary processes are easier to compute. For lectures on coupling, see \eg \cite{lindvall}. As in Section \ref{section:comparaison}, we close the section with a table that sums up our results.

We denote by $(T_k)_{k\in \N}$ the successive jump times of the process $(N^{11}_t,N^{12}_t)_{t\geq 0}$, with the convention $T_0=0$ and $T_{k+1}=+\infty$ if $N^{11}_{T_k}=N^{12}_{T_k}=0$. We will be led to consider the continuous time process $(N^{11}_t,N^{12}_t)_{t\geq 0}$ as well as the discrete time process $(N^{11}_{T_k},N^{12}_{T_k})_{k\in \N}$.\\
Recall also that the time at which the process reaches the horizontal or vertical axes has been denoted by $\tau$ (see \eqref{def:tau}).

The denominations for super and subcritical cases are used for small populations but with some slight differences with the cases considered in Section \ref{section:largepop}: even when the reproduction rate is very high in comparison of the death rate, there may always be a probability of extinction by demographic stochasticity. With the terminology of branching processes \citep[\eg][]{athreyaney}, we say that we are in the supercritical case if there is a positive probability of survival. In the subcritical case, there is almost sure extinction.

\subsection{Wright's model (Model 1) in a small population}

\unitlength=1cm
\begin{figure}[ht]
  \begin{center}
\begin{minipage}[b]{.45\textwidth}\centering
    \begin{picture}(4,4)
    \put(1,1){\vector(1,0){3}}
    \put(1,1){\vector(0,1){3}}
    \put(3.8,0.6){$N^{11}$}
    \put(0.3,3.6){$N^{12}$}
    \put(3,3){\vector(1,0){1}}
    \put(3,3){\vector(-1,0){1}}
    \put(3,3){\vector(0,1){1}}
    \put(3,3){\vector(0,-1){1}}
    \put(3.8,2.6){$\bar{r}\frac{i+j}{2}$}
    \put(2,2.6){$d\,i$}
    \put(3.1,2){$d\,j$}
    \put(3.1,3.8){$\bar{r}\frac{i+j}{2}$}
    \dottedline(1,3)(3,3)
    \dottedline(3,1)(3,3)
    \put(0.4,2.8){$j$}
    \put(2.8,0.6){$i$}
    \end{picture}
\end{minipage}
\begin{minipage}[b]{.45\textwidth}\centering
    \begin{picture}(4,4)
    \put(1,1){\vector(1,0){3}}
    \put(1,1){\vector(0,1){3}}
    \put(3.8,0.6){$N^{11}$}
    \put(0.3,3.6){$N^{12}$}
    \put(3,3){\vector(1,0){1}}
    \put(3,3){\vector(-1,0){1}}
    \put(3,3){\vector(0,1){1}}
    \put(3,3){\vector(0,-1){1}}
    \put(4.2,3){$\frac{\bar{r}}{2(\bar{r}+d)}$}
    \put(1.1,2.6){$\frac{d\,i}{(\bar{r}+d)(i+j)}$}
    \put(3.1,2){$\frac{d\,j}{(\bar{r}+d)(i+j)}$}
    \put(3.1,3.8){$\frac{\bar{r}}{2(\bar{r}+d)}$}
    \dottedline(1,3)(3,3)
    \dottedline(3,1)(3,3)
    \put(0.4,2.8){$j$}
    \put(2.8,0.6){$i$}
    \end{picture}\end{minipage}
  \vspace{-0.3cm}
  \caption{\textit{Evolution of the distylous system $(N^{11}_t,N^{12}_t)_{t\in \R_+}$ in Wright's model (Model 1). Rates of events (for positive $i$'s and $j$'s) are pictured on the left, while transition probabilities of the embedded Markov chain are represented on the right.}}\label{Fig4}
  \end{center}
\end{figure}
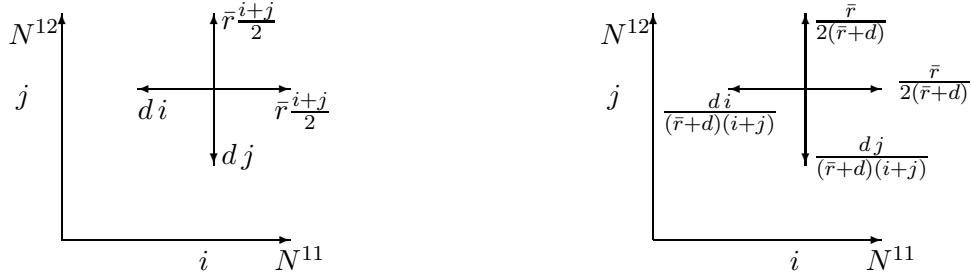

\subsubsection{Extinction probabilities: recurrence equation}\label{section:recurrence}

A difficulty comes from the fact that the transition rates of $(N^{11}_t,N^{12}_t)_{t\in \R_+}$ and the transition probabilities of the associated discrete time Markov chain vary with the state of the population (see Fig. \ref{Fig4}). Techniques developed in the literature of random walks on positive quadrants usually focus on homogeneous random walks \citep[\eg][]{kleinhaneveldpittenger,fayollemalyshevmenshikov,kurkovaraschel}. \\


For $i,j\in \N$, we denote by $p_{i,j}=\P_{ij}(\exists t\geq 0,\, N^{11}_t=0 \mbox{ or }N^{12}_t=0)$, where $\P_{ij}$ means that we start with the initial condition $N^{11}_0=i$ and $N^{12}_0=j$. By symmetry arguments, we have:
\begin{equation}
\forall i,j\in \N,\quad p_{i,j}=p_{j,i}.\label{symmetry}
\end{equation}Moreover,
\begin{equation}
\mbox{when }i=0\mbox{ or }j=0,\quad p_{i,j}=1.\label{boundary_condition_dirichlet}
\end{equation}
We begin with a recurrence equation satisfied by these extinction probabilities:
\begin{proposition}\label{prop_dirichlet}(i) The extinction probabilities $p_{i,j}$ for $i,j\in \N^*$ satisfy the following recurrence equations:
\begin{align}
p_{i,j}= & \frac{di}{(\bar{r}+d)(i+j)}p_{i-1,j}+\frac{dj}{(\bar{r}+d)(i+j)}p_{i,j-1}\nonumber\\
+ & \frac{\bar{r}}{2(\bar{r}+d)}p_{i,j+1}+\frac{\bar{r}}{2(\bar{r}+d)}p_{i+1,j}.\label{eq3}
\end{align}
The family $(p_{i,j})_{i,j\in \N}$ is a solution of the Dirichlet problem (\ref{eq3}) with boundary condition (\ref{boundary_condition_dirichlet}). Uniqueness of the solution may not hold, but the extinction probabilities $(p_{i,j})_{i,j\in \N}$ define
the smallest positive solution of this problem.\\
(ii) Assume that the probabilities $p_{i,1}$ for $i\in \N^*$ are given. Then, the other probabilities $p_{i,j}$ are completely determined.
\end{proposition}

\begin{proof}We refer to \citet{lafitterascheltran} for the proof, which follows classical arguments found in \citet{baldimazliakpriouret} or \citet{revuz}: for \eqref{eq3}, we consider the embedded Markov chain and apply the Markov property at $t=1$.
\end{proof}

The recurrence equation (\ref{eq3}) may not admit a unique positive solution. Point (ii) of Prop. \ref{prop_dirichlet} tells us that to every boundary condition $(p_{i,1},i\in \N)$ corresponds a solution. Point (i) tells us that the smallest positive solution gives the extinction probabilities $p_{i,j}$. The computation of the $p_{i,j}$'s and the determination of the probabilities $p_{i,1}$'s is a work in progress by \citet*{lafitterascheltran}.\\
By coupling and comparison techniques, we manage to define subcritical, critical and supercritical regimes, where there is almost sure extinction or positive probability of survival of the random walk killed at the boundary.

\subsubsection{Extinction probabilities: coupling approach}

It is difficult to solve \eqref{eq3}. Our purpose is to approximate the extinction probability with couplings. Lower and upper bounds are provided by:
\begin{proposition}\label{propencadrement}
Assume that $(N_0^{11},N_0^{12})=(i,j)$ is a pair of positive integers.  \\
(i) The population goes extinct in finite time almost surely if and only if $\bar{r}\leq d$.\\
\noindent (ii) If $\bar{r}>d$, there is a strictly positive probability of survival $1-p_{i,j}$ when starting from any $i,j\in \N^*$ and the extinction probability $p_{i,j}$ satisfies:
\begin{equation}
\left(\frac{d}{\bar{r}}\right)^{i+j} \leq p_{i,j}\leq \left(\frac{d}{\bar{r}}\right)^{i}+\left(\frac{d}{\bar{r}}\right)^{j}-\left(\frac{d}{\bar{r}}\right)^{i+j}\label{encadrement}
\end{equation}
\end{proposition}

\begin{figure}[!ht]
\begin{center}
\begin{tabular}{cc}
(a) & (b)  \\
 & \\
 \includegraphics[width=0.40\textwidth,height=0.23\textheight,angle=0,trim=1cm 1cm 1cm 1cm]{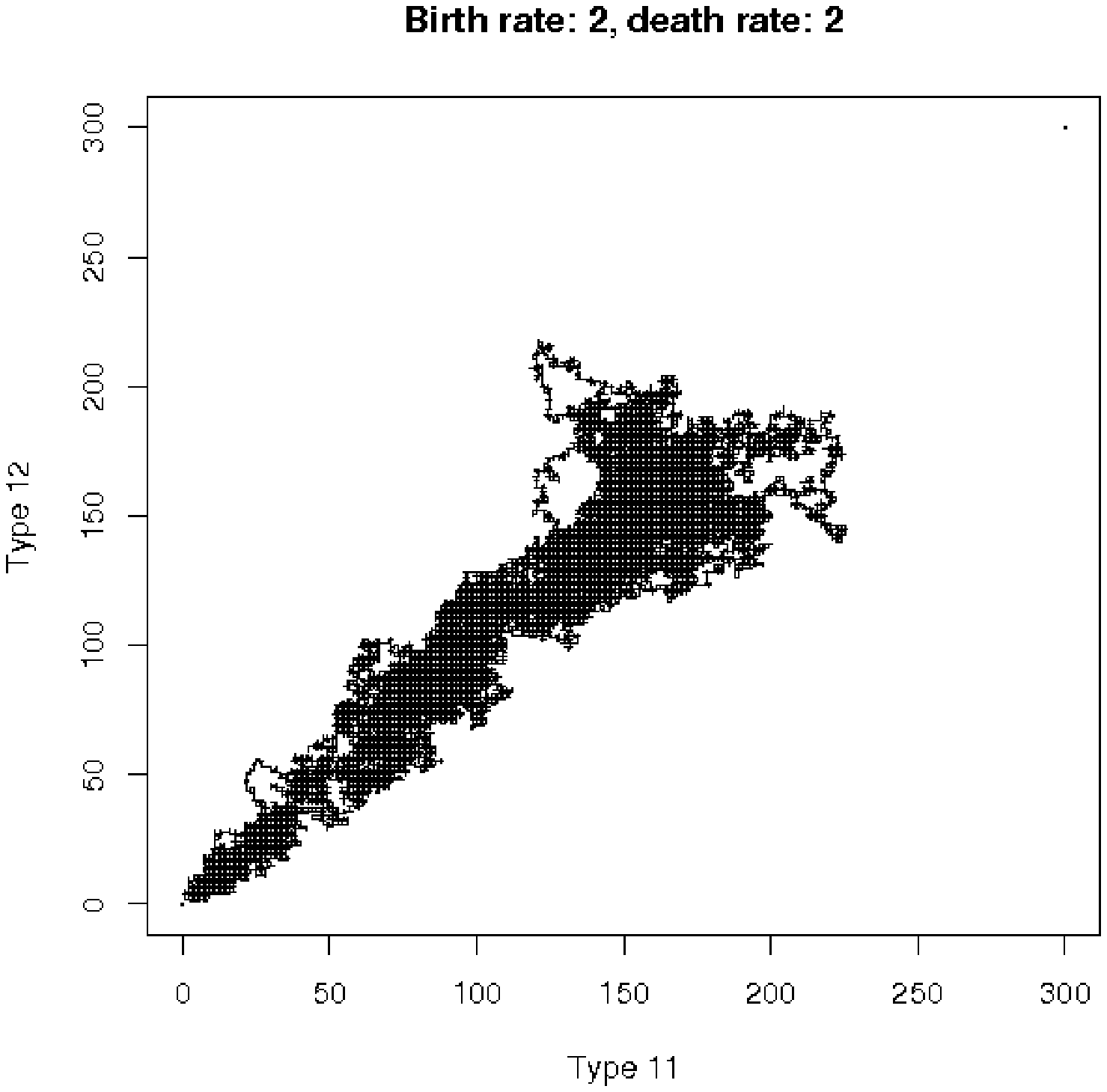}
 &
\includegraphics[width=0.40\textwidth,height=0.23\textheight,angle=0,trim=1cm 1cm 1cm 1cm]{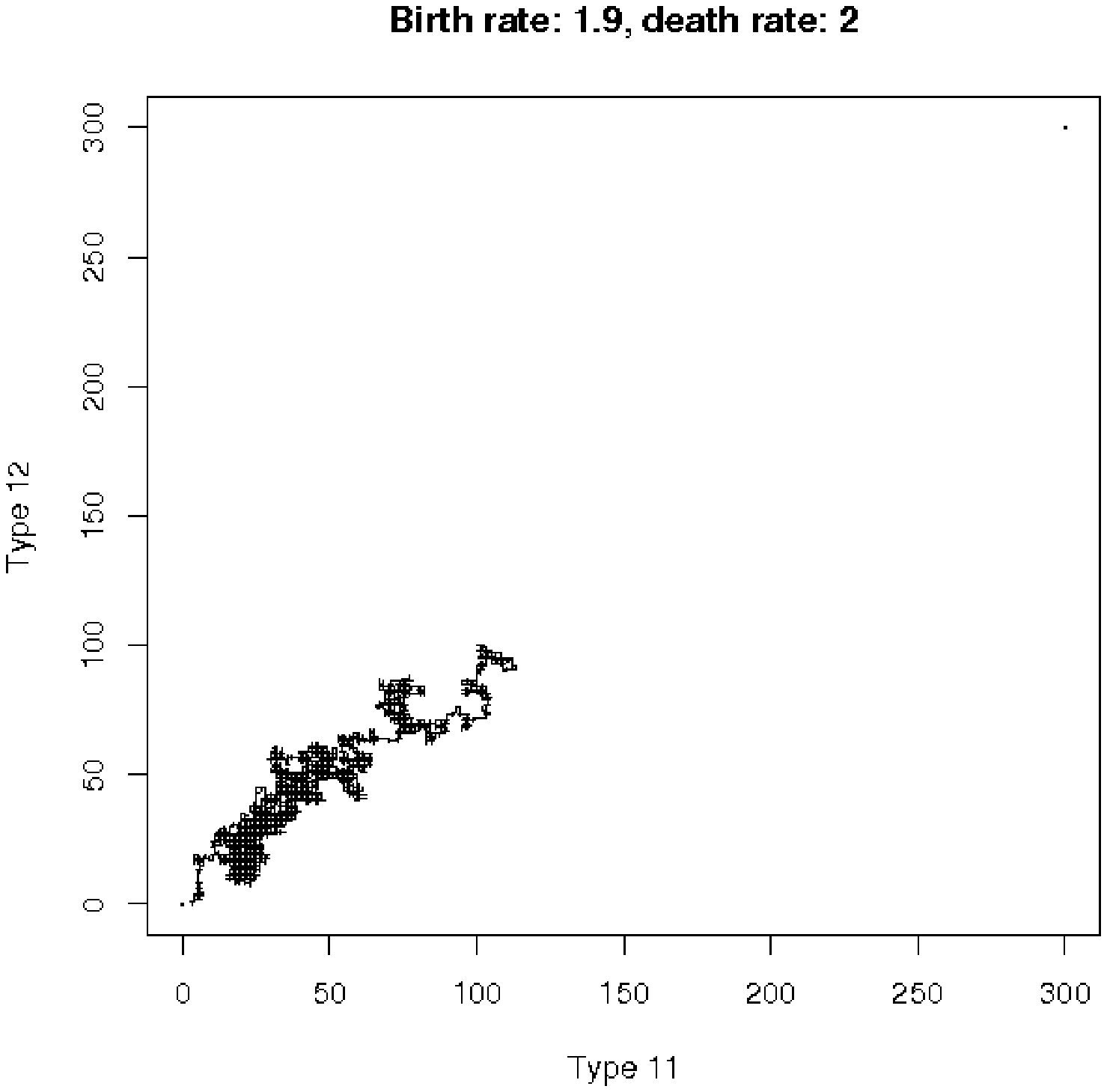}\\
 & \\
 (c) & (d) \\
 & \\
 \includegraphics[width=0.40\textwidth,height=0.23\textheight,angle=0,trim=1cm 1cm 1cm 1cm]{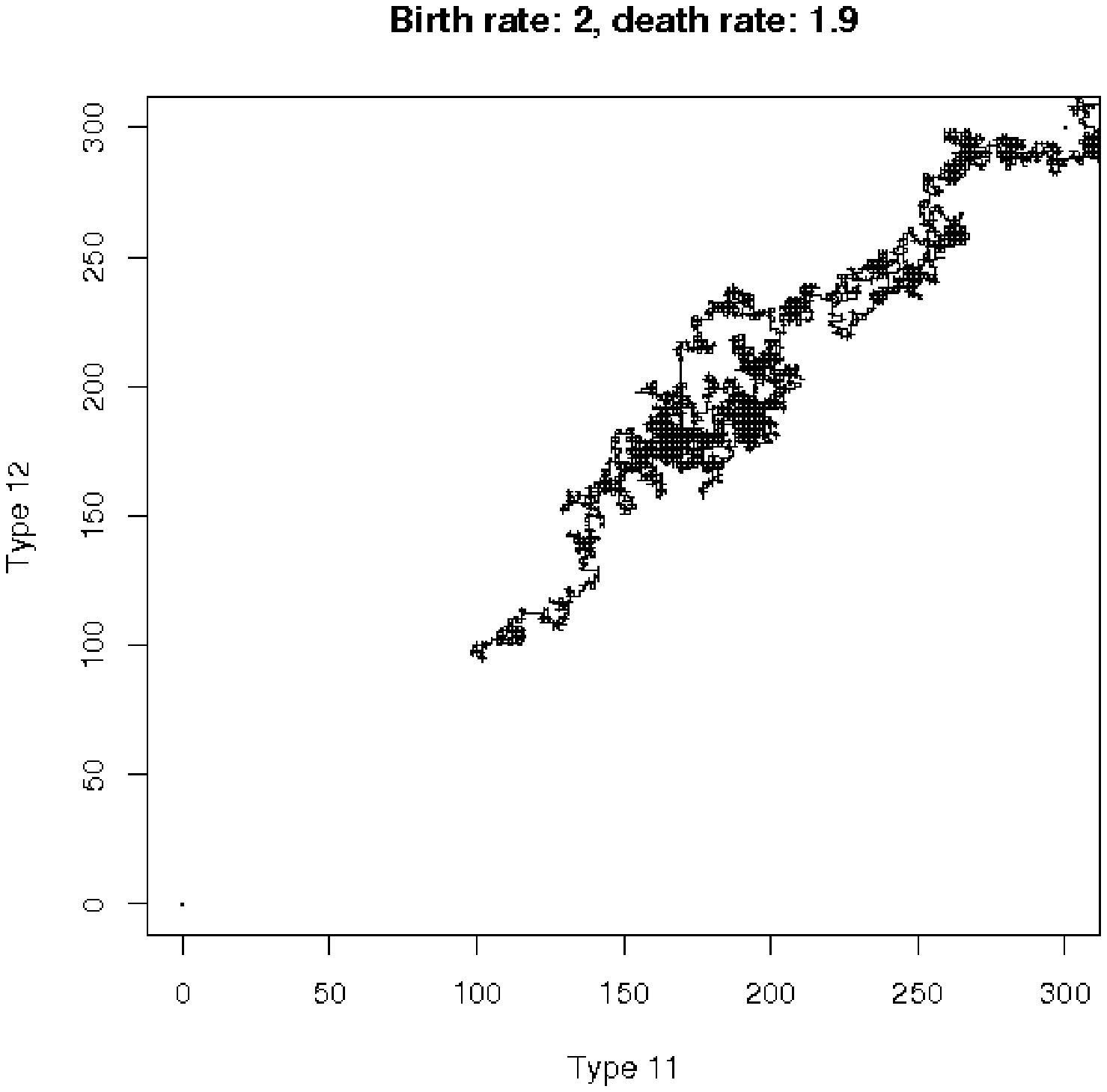}
 &
\includegraphics[width=0.40\textwidth,height=0.23\textheight,angle=0,trim=1cm 1cm 1cm 1cm]{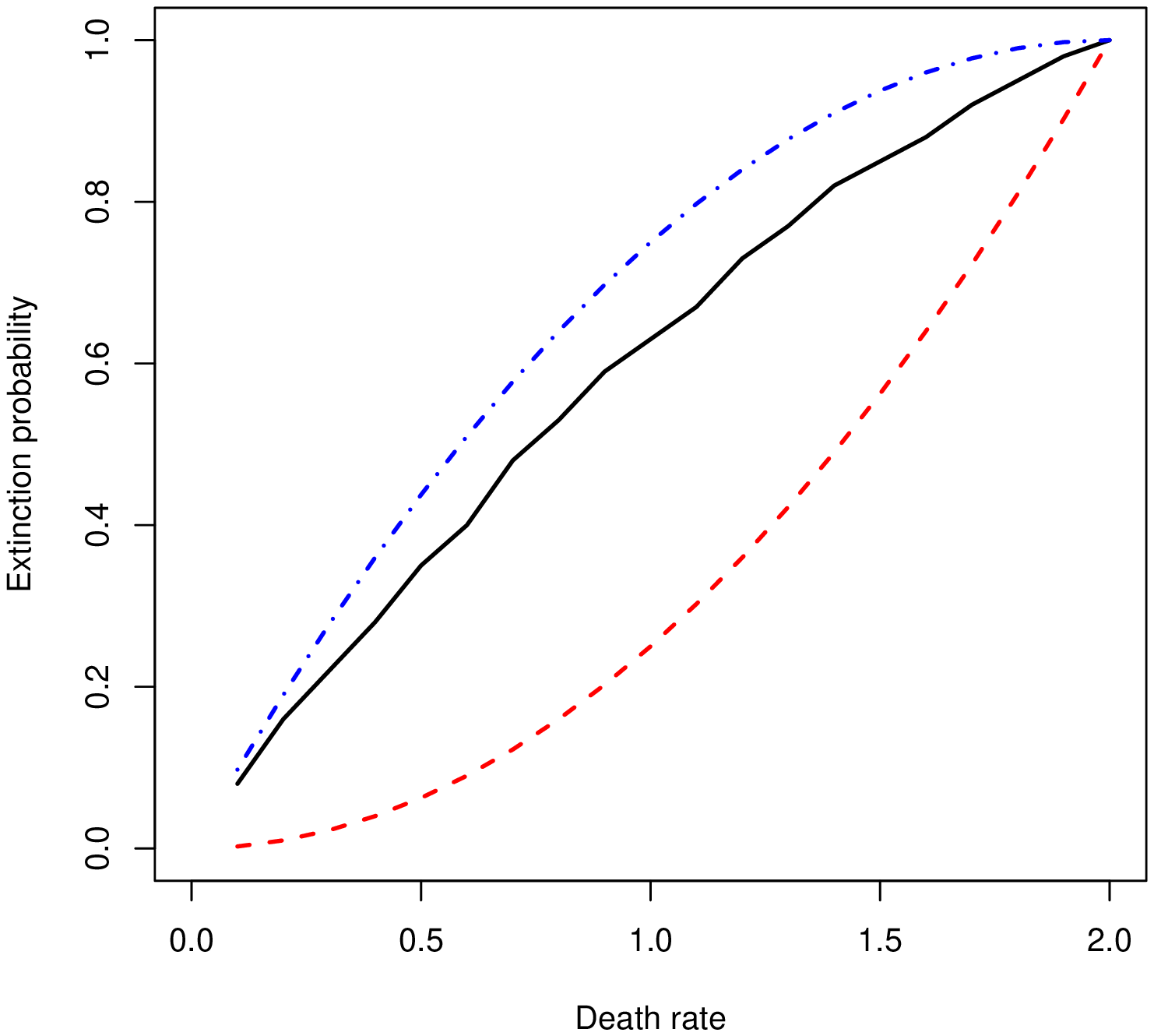}
 \end{tabular}
\caption{\textit{(a)-(c): Simulation of the paths $(N^{11}_t,N_t^{12})_{t\geq 0}$ in the critical, subcritical and supercritical cases. The number of individuals of genotype $\{1,1\}$ is in abscissa, the number of individuals of genotype $\{1,2\}$ is in ordinate. For the three simulations, the initial condition is $(100,100)$. (d): Simulated extinction probability for $(N^{11},N^{12})$ in plain line, together with the extinction probability of $(\wN^{11},\wN^{12})$ in dashed red line and of $(\widehat{N}^{11}, \widehat{N}^{12})$ in blue dash-dots. We start with $(N^{11}_0,N^{12}_0)=(1,1)$.}}\label{fig_simu}
\end{center}
\end{figure}


To prove Proposition \ref{propencadrement}, we introduce processes that dominate and bound below $(N^{11}_t, N^{12}_t)_{t\in \R_+}$, and that allow to obtain the bounds in (\ref{encadrement}).

\subsubsection{Domination and proof of the lower bound}\label{section:dominationsto}
In Wright's model, the random walks (\ref{equationsdistyles}) can be dominated by the random walks $\widetilde{N}^{11}$ and $\widetilde{N}^{12}$ with initial conditions $N^{11}_0$ and $N^{12}_0$, and reproduction rates:
\begin{equation}\frac{\bar{r}}{2}\big( \widetilde{N}^{11}_s +\widetilde{N}^{12}_s \big)\mbox{ instead of }\frac{\bar{r}}{2}\big(N^{11}_s + N^{12}_s \big)\ind_{N^{11}_s>0;N^{12}_s>0}\label{tauxdominants}\end{equation}
such that until time $\tau$, the processes $(N_t^{11}, N_t^{12})_{t\in \R_+}$ and $(\widetilde{N}_t^{11}, \widetilde{N}_t^{12})_{t\in \R_+}$ have the same paths (this can be obtained by using the same Poisson measures for $(N_t^{11}, N_t^{12})_{t\in \R_+}$ and $(\wN_t^{11}, \wN_t^{12})_{t\in \R_+}$, see Appendix \ref{section:SDEdetaillees}). For $(\widetilde{N}^{11}_t,\widetilde{N}^{12}_t)_{t\in \R_+}$:
\begin{align}
\widetilde{N}_t^{11}= & N_0^{11}+\int_0^t \left(\frac{\bar{r}}{2}\big(\widetilde{N}^{11}_s + \widetilde{N}^{12}_s \big)-d\times \widetilde{N}^{11}_s\right)ds+\widetilde{M}^{11}_t\nonumber\\
\widetilde{N}_t^{12}= & N_0^{12}+\int_0^t \left(\frac{\bar{r}}{2}\big(\widetilde{N}^{11}_s +\widetilde{N}^{12}_s \big)-d\times \widetilde{N}^{12}_s\right)ds+\widetilde{M}^{12}_t.\label{equationsdistylestilde}
\end{align}
We define by $\widetilde{N}_t=\widetilde{N}_t^{11}+\widetilde{N}t^{12}$ the total size of the population ruled by (\ref{tauxdominants}). As $(N_t^{11}, N_t^{12})_{t\in \R_+}$ and $(\wN_t^{11}, \wN_t^{12})_{t\in \R_+}$ coincide until time $\tau$, they reach the axes at the same stopping time $\tau$. 
Moreover, $(\widetilde{N}^{11}_t, \widetilde{N}^{12}_t)_{t\in \R_+}$ stochastically dominates $(N^{11}_t, N^{12}_t)_{t\in \R_+}$ since it allows rebirths of the disappeared type on the boundaries $\{N^{11}=0\}\cup\{N^{12}=0\}$ provided the total population size is strictly positive. The coupling presented here also works for general division rates $r(N)$ instead of $\bar{r}$.\\

\par The process $(\widetilde{N}^{11}_t,\widetilde{N}^{12}_t)_{t\in \R_+}$ is a particular case of two-type Galton-Watson process \citep[\eg][Chap. V]{athreyaney} where each individual lives during an exponential time of rate $\bar{r}+d$. At death, a particle is replaced by:
\begin{itemize}
\item zero offspring with probability $d/(\bar{r}+d)$: this corresponds to the case of a real death.
\item two offspring of the same type as their mother with probability $\bar{r}/(2(\bar{r}+d))$: one is the mother and the other is her daughter, of the same type.
\item two offspring of different types $\{1,1\}$ and $\{1,2\}$ with probability $\bar{r}/(2(\bar{r}+d))$: one is the mother and the other is her daughter with the other type.
\end{itemize}
Hence, once the trajectories have reached the horizontal axis, for instance, the extinct genotype $\{1,2\}$ may be regenerated from birth of individuals of genotype $\{1,2\}$ from individuals of genotype $\{1,1\}$. Notice also that the total size of the population $\widetilde{N}_t=\widetilde{N}^{11}_t+ \widetilde{N}^{12}_t$ is a continuous time birth and death process. Let:
\begin{equation}
\sigma=\inf\{t\geq 0,\, \widetilde{N}_t=0\}
\end{equation}be the extinction time of the dominating process.

\begin{proposition}\label{propmomentntilde}
Let us consider the processes $(\widetilde{N}^{11}, \widetilde{N}^{12})$ starting from the initial condition $(N^{11}_0,N^{12}_0)$. We set $N_0=N^{11}_0+N^{12}_0$.\\
(i) The total population is a continuous time birth and death process with birth and death rates $\bar{r}$ and $d$ respectively, and:
\begin{align}
 \mathbb{E}\big(\widetilde{N}_t\big)=  \mathbb{E}\big(N_0\big)e^{(\bar{r}-d)t},\qquad
  \mathbb{P}(\sigma<+\infty\, |\, \wN_0=1)=\left\{\begin{array}{l}
d/\bar{r}\mbox{ if }\bar{r}>d\\
1 \mbox{ otherwise}.
\end{array}\right.\label{ENttilde}\end{align}
(ii) For the population of type $\{S^1,S^2\}\in \{\{1,1\},\{1,2\}\}$ we have:
\begin{align}
\mathbb{E}\big(\widetilde{N}^{S^1S^2}_t\big)=\frac{\mathbb{E}(N_0)e^{(\bar{r}-d)t}}{2}+\left(\mathbb{E}\big(N^{S^1S^2}_0-\frac{E(N_0)}{2}\big)\right)e^{-d.t}.
\end{align}
\end{proposition}

\begin{proof}Once that it has been noticed that $(\wN_t)_{t\geq 0}$ is a continuous time birth and death process, the results are standard. It is indeed classical to consider the generating functions for which \citep[\eg][Chap. III 4-5]{athreyaney}:
\begin{equation}
 F(s,t)=\sum_{k=0}^{+\infty}\P(\wN_t=k\, |\, \wN_0=1)s^k = \left\{\begin{array}{l}
\frac{d(s-1)-e^{-(\bar{r}-d)t}(\bar{r}s-d)}{\bar{r}(s-1)-e^{-(\bar{r}-d)t}(\bar{r}s-d)}\quad \mbox{ if }\bar{r}\not= d\\
\frac{s-\bar{r}t(s-1)}{1-\bar{r}t(s-1)}\quad \mbox{ if }\bar{r}=d.
\end{array}\right.\label{expressionF(s,t)}
\end{equation}
Since $\mathbb{P}(\sigma<+\infty\, |\, \wN_0=1)=\lim_{t\rightarrow +\infty}F(0,t)$ and since:
\begin{eqnarray*}
\P(\sigma\leq t \, |\, \wN_0=1)=\P(\wN_t=0\, |\, \wN_0=1)=F(0,t)= & \frac{d}{\bar{r}}\Big(\frac{1-e^{-(\bar{r}-d)t}}{1-\frac{d}{\bar{r}}e^{-(\bar{r}-d)t}}\Big) & \mbox{ if }\bar{r}>d\\
= &  \frac{\bar{r}t}{1+\bar{r}t} & \mbox{ if }\bar{r}=d\\
= & \frac{e^{(d-\bar{r})t-1}}{e^{(d-\bar{r})t-\frac{\bar{r}}{d}}} & \mbox{ if }\bar{r}<d,
\end{eqnarray*}we deduce the result by taking the limit in $t$. The expectations $\E(\wN^{11}_t)$, $\E(\wN^{12}_t)$ and $\E(\wN_t)$ are obtained by noticing that $t\mapsto (\E(\wN^{11}_t),\E(\wN^{12}_t))$ solves the system (\ref{solutioneqA}) which has been studied in a previous part.
\end{proof}

In conclusion, we can distinguish three regimes:
\begin{proposition}\label{prop4.4}
(i) In the subcritical case $\bar{r}<d$, the population $(\wN^{11}_t,\wN^{12}_t)_{t\geq 0}$ goes extinct with probability 1 and so does $(N^{11}_t,N^{12}_t)_{t\geq 0}$.\\
(ii) In the critical case $\bar{r}=d$, the population $(\wN^{11}_t,\wN^{12}_t)_{t\geq 0}$ goes extinct with probability 1, and so does $(N^{11}_t,N^{12}_t)_{t\geq 0}$, but the expectation of the population size $\mathbb{E}\big(\widetilde{N}_t\big)$ remains constant and the extinction time is not integrable.\\
(iii) In the supercritical case $\bar{r}>d$, there is a positive probability of survival for $(\widetilde{N}^{11}_t,\widetilde{N}^{12}_t)_{t\in \R_+}$, equal to $(d/\bar{r})^{N^{11}_0+N^{12}_0}$ that
provides the lower bound in (\ref{encadrement}). Moreover, in this case, the mean size $\mathbb{E}\big(\widetilde{N}_t\big)$ tends to infinity with:
$$\lim_{t\rightarrow +\infty}\frac{\mathbb{E}\big(\widetilde{N}^{11}_t\big)}{\mathbb{E}(\wN_t)}=\frac{1}{2}.$$
\end{proposition}
\begin{proof}Point (i) is clear. For Point (ii), we use that:
$$\E(\sigma)=\int_0^{+\infty}\P(\sigma>t)dt=\int_0^{+\infty}\frac{dt}{1+\bar{r}t}=+\infty.$$
For Point (iii), we use (\ref{ENttilde}) and the branching property.
\end{proof}

\subsubsection{Minoration and proof of the upper bound}\label{section:stochasticminoration}

\paragraph{Comparison with an homogeneous random walk}
To obtain the upper bound in (\ref{encadrement}), we have to find a process that has an extinction probability higher than $(N^{11}_t,N^{12}_t)_{t\in \R_+}$. A natural process is the following random walk $(\widehat{N}^{11}_t,\widehat{N}^{12}_t)_{t\in \R_+}$ on $\N^2$: when the process is in $(i,j)\in (\N^*)^2$
\begin{itemize}
 \item individuals of genotype $\{1,1\}$ or $\{1,2\}$ are produced with rate $\bar{r}(i+j)/2$,
 \item individuals of genotype $\{1,1\}$ or $\{1,2\}$ die with rate $d(i+j)/2$.
\end{itemize}
If we consider the associated discrete time Markov chain, we obtain transition rates which are homogeneous: whatever the state $(i,j)\in (\N^*)^2$, given the occurrence of an event, the chain goes to the north or east with probability $\bar{r}/(2(\bar{r}+d))$ and to the south or west with probability $d/(2(\bar{r}+d))$ (see Fig. \ref{Fig7}). We begin with computing the extinction probability for this process. Then, we show that a naive coupling between $(N^{11}_t,N^{12}_t)_{t\in \R_+}$ and $(\widehat{N}^{11}_t,\widehat{N}^{12}_t)_{t\in \R_+}$ does not work and exploit the symmetry of the problem to obtain the upper bound in \eqref{encadrement}. For this, we introduce a coupling of the original process $(N^{11}_t,N^{12}_t)_{t\in \R_+}$ with two auxiliary processes $(\widehat{N}^{11,+}_t,\widehat{N}^{12,+}_t)_{t\in \R_+}$ and $(\widehat{N}^{11,-}_t,\widehat{N}^{12,-}_t)_{t\in \R_+}$.

\unitlength=1cm
\begin{figure}[ht]
  \begin{center}
  \begin{minipage}[b]{.45\textwidth}\centering
    \begin{picture}(4,4)
    \put(1,1){\vector(1,0){3}}
    \put(1,1){\vector(0,1){3}}
    \put(3.8,0.6){$\widehat{N}^{11,+}$}
    \put(0.3,3.6){$\widehat{N}^{12,+}$}
    \put(3,3){\vector(1,0){1}}
    \put(3,3){\vector(-1,0){1}}
    \put(3,3){\vector(0,1){1}}
    \put(3,3){\vector(0,-1){1}}
    \put(3.8,2.6){$\bar{r}\frac{i+j}{2}$}
    \put(2,2.6){$d\frac{i+j}{2}$}
    \put(3.1,2){$d\frac{i+j}{2}$}
    \put(3.1,3.8){$\bar{r}\frac{i+j}{2}$}
    \dottedline(1,3)(3,3)
    \dottedline(3,1)(3,3)
    \put(0.4,2.8){$j$}
    \put(2.8,0.6){$i$}
    \end{picture}
\end{minipage}
\begin{minipage}[b]{.45\textwidth}\centering
    \begin{picture}(4,4)
    \put(1,1){\vector(1,0){3}}
    \put(1,1){\vector(0,1){3}}
    \put(3.8,0.6){$\widehat{N}^{11,+}$}
    \put(0.3,3.6){$\widehat{N}^{12,+}$}
    \put(3,3){\vector(1,0){1}}
    \put(3,3){\vector(-1,0){1}}
    \put(3,3){\vector(0,1){1}}
    \put(3,3){\vector(0,-1){1}}
    \put(4.2,3){$\frac{\bar{r}}{2(\bar{r}+d)}$}
    \put(1.1,2.6){$\frac{d}{2(\bar{r}+d)}$}
    \put(3.1,2){$\frac{d}{2(\bar{r}+d)}$}
    \put(3.1,3.8){$\frac{\bar{r}}{2(\bar{r}+d)}$}
    \dottedline(1,3)(3,3)
    \dottedline(3,1)(3,3)
    \put(0.4,2.8){$j$}
    \put(2.8,0.6){$i$}
    \end{picture}\end{minipage}
  \vspace{-0.3cm}
  \caption{\textit{Rates of events (left) and transition probabilities (right) of $(\widehat{N}^{11,+}_{t},\widehat{N}^{12,+}_t)_{t\in \R_+}$ and of its discrete time skeleton.}}\label{Fig7}
  \end{center}
\end{figure}
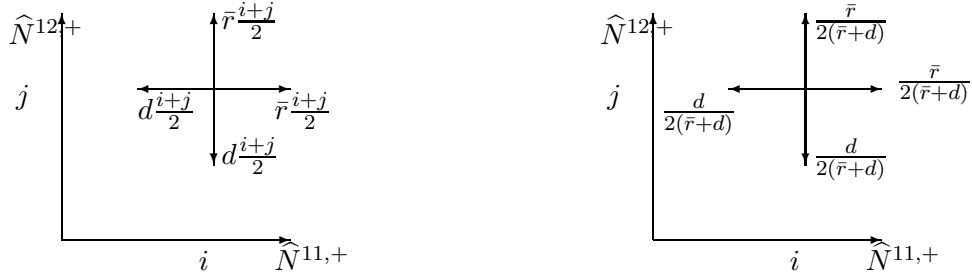

Let us give some explanation. $(N^{11}_t,N^{12}_t)_{t\in \R_+}$ and $(\widehat{N}^{11}_t,\widehat{N}^{12}_t)_{t\in \R_+}$ have the same birth rates. When we are in the upper octant $\{(i,j)\in \N^2,\ i<j\}$, then $d \,i<d(i+j)/2<d\,j$. In this case, there are higher (resp. lower) death rates for the $j$ individuals $\{1,1\}$ (resp. the $i$ individuals $\{1,2\}$) in $(\widehat{N}^{11}_t,\widehat{N}^{12}_t)_{t\in \R_+}$ compared with $(N^{11}_t,N^{12}_t)_{t\in \R_+}$. In the lower octant $\{i>j\}$, the reverse holds. Heuristically, for the homogeneous random walk $(\widehat{N}^{11}_t,\widehat{N}^{12}_t)_{t\in \R_+}$, the less abundant genotype is pushed more strongly towards the axis, which should entail a higher extinction probability.

\begin{proposition}\label{prop_couplinga2}
For the process $(\widehat{N}^{11}_t,\widehat{N}^{12}_t)_{t\in \R_+}$, we have when $\bar{r}>d$:
 \begin{align}
  \P_{ij}\Big(\exists t\in \R_+,\, \widehat{N}^{11}_t=0\mbox{ or }\widehat{N}^{12}_t=0\Big)=\Big(\frac{d}{\bar{r}}\Big)^i+\Big(\frac{d}{\bar{r}}\Big)^j-\Big(\frac{d}{\bar{r}}\Big)^{i+j}.\label{etape9}
 \end{align}When $\bar{r}\leq d$, this probability equals 1.
\end{proposition}

\begin{proof}When $\bar{r}\leq d$, the result is obtain with arguments similar to the ones developed in Section \ref{section:dominationsto}. Let us consider the case $\bar{r}>d$. We denote by $(T_k)_{k\in \N}$ the jump times of $(\widehat{N}^{11}_t,\widehat{N}^{12}_t)_{t\in \R_+}$. For the associated Markov chain, $(\widehat{N}^{11}_{T_k})_{k\in \N}$ and $(\widehat{N}^{12}_{T_k})_{k\in \N}$ move independently. Hence, the survival probability is the product of the survival probabilities of each of these processes:
\begin{equation}
\P_{ij}\Big(\forall t\in \R_+,\, \widehat{N}^{11}_t>0\mbox{ and }\widehat{N}^{12}_t>0\Big)=\left(1-\Big(\frac{d}{\bar{r}}\Big)^i\right)\left(1-\Big(\frac{d}{\bar{r}}\Big)^j\right).
\end{equation}This hence provides \eqref{etape9}.\hfill $\Box$
\end{proof}

\paragraph{A first naive coupling}
For the original process $(N^{11}_t,N^{12}_t)_{t\in \R_+}$ as well as for the random walk $(\widehat{N}^{11}_t,\widehat{N}^{12}_t)_{t\in \R_+}$, when in the state $(i,j)$, the jump rates to the north and east on the one hand, and to the south and west on the other hand, depend only on the sum $i+j$ and are the same for the two processes. Thus, a first (naive) idea of coupling of these two processes is to have the jumps to the north/east or to the south/west occur simultaneously for both process. All the jumps of the original process are copied by the auxiliary process except the following modifications:
 \begin{itemize}
 \item When we are in the upper octant $\{i\leq j\}$, a jump of $(N^{11}_t,N^{12}_t)_{t\in \R_+}$ to the south gives:
 \begin{itemize}
 \item a jump of $(\widehat{N}^{11}_t,\widehat{N}^{12}_t)_{t\in \R_+}$ to the south with probability $(i+j)/(2j)$,
 \item a jump of $(\widehat{N}^{11}_t,\widehat{N}^{12}_t)_{t\in \R_+}$ to the west with probability $(j-i)/(2j)$.
 \end{itemize}
 \item When we are in the lower octant $\{i>j\}$, a jump of $(N^{11}_t,N^{12}_t)_{t\in \R_+}$ to the west gives:
 \begin{itemize}
 \item a jump of $(\widehat{N}^{11}_t,\widehat{N}^{12}_t)_{t\in \R_+}$ to the west with probability $(i+j)/(2i)$,
 \item a jump of $(\widehat{N}^{11}_t,\widehat{N}^{12}_t)_{t\in \R_+}$ to the south with probability $(i-j)/(2i)$.
 \end{itemize}
 \end{itemize}
One notices that as long as none of the two processes has reached a boundary of the positive quadrant: $N^{11}_t+N^{12}_t=\widehat{N}^{11}_t+\widehat{N}^{12}_t$. This entails that the jump rates and times to the north/east and south/west remain the same for both processes. \\
Unfortunately, this coupling is not sufficient as there exist paths where the original process goes extinct before $(\widehat{N}^{11}_t,\widehat{N}^{12}_t)_{t\in \R_+}$: see \eg Fig. \ref{Fig8} (a).

\unitlength=1cm
\begin{figure}[ht]
  \begin{center}
  \begin{minipage}[b]{.45\textwidth}\centering
    \begin{picture}(6,6)
    \put(1,1){\vector(1,0){4.5}}
    \put(1,1){\vector(0,1){4.5}}
    \put(5.3,0.6){$i$}
    \put(0.3,5.1){$j$}
    \linethickness{0.5mm}
    \put(1,4){\color{blue}{\line(1,0){4}}}
    \put(1,4){\circle*{0.1}}
    \put(2,4){\circle*{0.1}}
    \put(3,4){\circle*{0.1}}
    \put(4,4){\circle*{0.1}}
    \put(5,4){\circle*{0.1}}
    \put(5.15,4.15){$(N^{11}_0,N^{12}_0)$}
    \put(-1,4.15){$(N^{11}_{T_4},N^{12}_{T_4})$}
    \put(1.2,3.3){$(\widehat{N}^{11}_{T_4},\widehat{N}^{12}_{T_4})$}
    \linethickness{0.1mm}
    \color{red}{
    \dashline{0.5}(2,3)(5,3)
    \dashline{0.5}(5,4)(5,3)
    }
    \put(2,3){\circle*{0.1}}
    \put(3,3){\circle*{0.1}}
    \put(4,3){\circle*{0.1}}
    \put(5,3){\circle*{0.1}}
    \color{black}{\dashline{0.2}(1,1)(5,5)}
    \end{picture}
\end{minipage}
  \begin{minipage}[b]{.45\textwidth}\centering
    \begin{picture}(6,6)
    \put(1.1,1){\vector(1,0){4.5}}
    \put(1.1,1){\vector(0,1){4.5}}
    \put(5.3,0.6){$i$}
    \put(0.3,5.1){$j$}
    \linethickness{0.5mm}
    \put(2,4){\color{blue}{\line(1,0){3}}}
    \put(2,4){\circle*{0.1}}
    \put(3,4){\circle*{0.1}}
    \put(4,4){\circle*{0.1}}
    \put(5,4){\circle*{0.1}}
    \put(5.15,4.15){$(N^{11}_0,N^{12}_0)$}
    \put(1.95,4.3){$(N^{11}_{T_3},N^{12}_{T_3})$}
    \put(5.2,1.3){$(\widehat{N}^{11,-}_{T_3},\widehat{N}^{12,-}_{T_3})$}
    \put(1.15,5.3){$(\widehat{N}^{11,+}_{T_3},\widehat{N}^{12,+}_{T_3})$}
    \linethickness{0.1mm}
    \color{red}{
    \dashline{0.5}(5,4)(5,1)
    \dashline{0.5}(4,5)(1,5)
    }
    \put(5,3){\circle*{0.1}}
    \put(5,2){\circle*{0.1}}
    \put(5,1){\circle*{0.1}}
    \put(1,5){\circle*{0.1}}
    \put(2,5){\circle*{0.1}}
    \put(3,5){\circle*{0.1}}
    \put(4,5){\circle*{0.1}}
    \color{black}{\dashline{0.2}(1,1)(5,5)}
    \color{black}{\dashline{0.05}(1,4.8)(4.8,1)}
    \end{picture}
\end{minipage}
  \vspace{-0.3cm}
  \caption{\textit{(a) An example of path for the naive coupling between $(N^{11}_t,N^{12}_t)_{t\in \R_+}$ (in plain blue line) and $(\widehat{N}^{11}_t,\widehat{N}^{12}_t)_{t\in \R_+}$ (in dashed red line) where the original process reaches the boundary before $(\widehat{N}^{11}_t,\widehat{N}^{12}_t)_{t\in \R_+}$. The two processes are started in the lower octant $\{i>j\}$. For the first step, the west move of the original process is replaced by a south move for the homogeneous random walk. Then, the original process keeps moving west, and since it is in the upper octant $\{i>j\}$, these moves are the same for the auxiliary process. The original process reaches the vertical axis first. (b) When we use the coupling with the two processes $(\widehat{N}^{11,+}_t,\widehat{N}^{12,+}_t)_{t\in \R_+}$ and $(\widehat{N}^{11,-}_t,\widehat{N}^{12,-}_t)_{t\in \R_+}$ (dashed red line started in the upper and lower octants respectively), things happen as follows. For the first move, since the original process is in the lower octant, it is coupled with $(\widehat{N}^{11,-}_t,\widehat{N}^{12,-}_t)_{t\in \R_+}$ and the west move of the original process is here replaced with a south move of the latter process ; by symmetry, the auxiliary process $(\widehat{N}^{11,+}_t,\widehat{N}^{12,+}_t)_{t\in \R_+}$  in the upper octant makes a jump to the west. After this, the process $(\widehat{N}^{11}_t,\widehat{N}^{12}_t)_{t\in \R_+}$ evolves in the upper octant and is coupled with $(\widehat{N}^{11,+}_t,\widehat{N}^{12,+}_t)_{t\in \R_+}$: all the west moves are unchanged. The auxiliary processes reach the boundaries first.}}\label{Fig8}
  \end{center}
\end{figure}
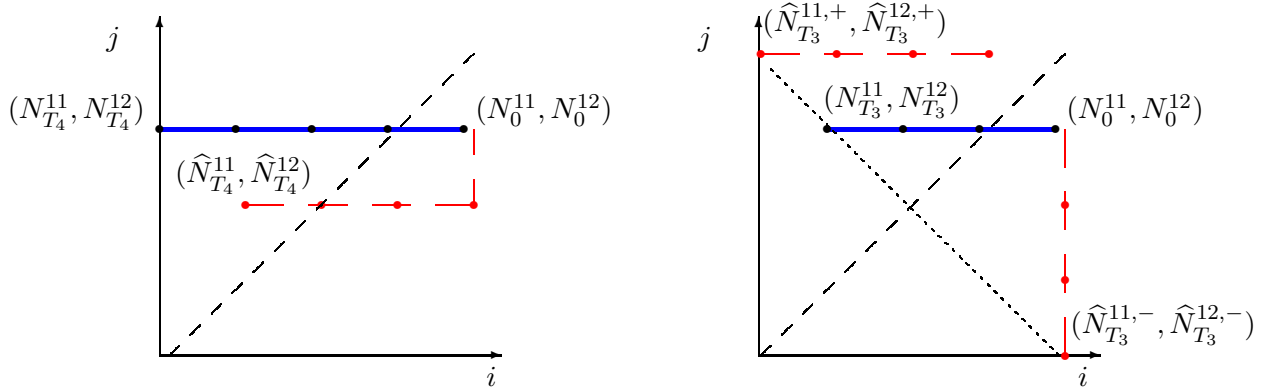

\paragraph{Coupling with two copies of $(\widehat{N}^{11}_t,\widehat{N}^{12}_t)_{t\in \R_+}$} To prove the upper bound in \eqref{encadrement}, we introduce a coupling of $(N^{11}_t,N^{12}_t)_{t\in \R_+}$ with 2 copies $(\widehat{N}^{11,+}_t,\widehat{N}^{12,+}_t)_{t\in \R_+}$ and $(\widehat{N}^{11,-}_t,\widehat{N}^{12,-}_t)_{t\in \R_+}$ of $(\widehat{N}^{11}_t,\widehat{N}^{12}_t)_{t\in \R_+}$.
Their respective initial conditions are:
\begin{align*}
& (\widehat{N}^{11,+}_0,\widehat{N}^{12,+}_0)=(\min(N^{11}_0,N^{12}_0),\max(N^{11}_0,N^{12}_0))\\
& (\widehat{N}^{11,-}_0,\widehat{N}^{12,-}_0)=(\max(N^{11}_0,N^{12}_0),\min(N^{11}_0,N^{12}_0)).
\end{align*}
As for the process $(\widehat{N}^{11}_t,\widehat{N}^{12}_t)_{t\in \R_+}$, there exists a coupling (see Appendix \ref{section:SDEdetaillees}) such that:
\begin{itemize}
 \item the two processes $(\widehat{N}^{11,+}_t,\widehat{N}^{12,+}_t)_{t\in \R_+}$ and $(\widehat{N}^{11,-}_t,\widehat{N}^{12,-}_t)_{t\in \R_+}$ are symmetric with respect to the line $\{i=j\}$, and started respectively in the upper and lower octants, $\{i\leq j\}$ and $\{i\geq j\}$. Thus, there is at every time one of these processes in each octant $\{i\leq j\}$ and $\{i\geq j\}$. Of course, since these processes are copies of $(\widehat{N}^{11}_t,\widehat{N}^{12}_t)_{t\in \R_+}$ they may change octant, but this is done simultaneously and they meet on the line $\{i=j\}$. Moreover, these processes both reach the axes at the same time.
 \item the jump times to the north/east or south/west are the same for $(N^{11}_t,N^{12}_t)_{t\in \R_+}$ and the two auxiliary processes. Indeed, when in the state $(i,j)$, these jump rates are $\bar{r}(i+j)$ and $d(i+j)$, and as long as the north/east and south/west jumps of the three processes occur simultaneously:
\begin{align*}
 N^{11}_t+N^{12}_t=\widehat{N}^{11,+}_t+\widehat{N}^{12,+}_t=\widehat{N}^{11,-}_t+\widehat{N}^{12,-}_t.
  \end{align*}
 \item the coupling is between the original process $(N^{11}_t,N^{12}_t)_{t\in \R_+}$ and the auxiliary process that belongs to the same octant, and it depends on this octant. The other auxiliary process is obtained by symmetry. If at a time $t$, $N^{11}_t\leq N^{12}_t$ (we are in the octant $\{i\leq j\}$), the coupling determines as follows the behavior of the auxiliary process that is found in this octant $\{i\leq j\}$ at $t$ (when  $N^{11}_t=N^{12}_t$, we choose $(\widehat{N}^{11,+}_t,\widehat{N}^{12,+}_t)$ by convention):
 \begin{itemize}
 \item If the original process moves to the north, east or west, so does the chosen auxiliary process. If the original process moves to the south, then the chosen auxiliary process moves to the south with probability $(N^{11}_t+N^{12}_t)/(2N^{12}_t)$ and to the west with probability $(N^{12}_t-N^{11}_t)/(2N^{12}_t)$.
  \item The behavior of the other auxiliary process is obtained by symmetry with the first bisector.
 \end{itemize}
 If at time $t$, $N^{11}_t> N^{12}_t$, the coupling determines the behavior of the auxiliary process that is found in this octant $\{i> j\}$ at $t$ as follows:
 \begin{itemize}
 \item If the original process moves to the north, east or south, so does the chosen auxiliary process. If the original process moves to the west, then the chosen auxiliary process moves to the west with probability $(N^{11}_t+N^{12}_t)/(2N^{11}_t)$ and to the south with probability $(N^{11}_t-N^{12}_t)/(2N^{11}_t)$.
  \item The behavior of the other auxiliary process is obtained by symmetry with the first bisector.
 \end{itemize}
It can be easily checked that the auxiliary processes have the same distribution as $(\widehat{N}^{11}_t,\widehat{N}^{12}_t)_{t\in \R_+}$ (see Appendix) and that
for every time $t$, as long as neither of the three process has reached the boundaries $\{i=0\}$ or $\{j=0\}$:
\begin{align*}
  & \min(\widehat{N}^{11,+}_t,\widehat{N}^{11,-}_t)\leq N^{11}_t\leq \max(\widehat{N}^{11,+}_t,\widehat{N}^{11,-}_t)\\
  & \min(\widehat{N}^{12,+}_t,\widehat{N}^{12,-}_t)\leq N^{12}_t\leq \max(\widehat{N}^{12,+}_t,\widehat{N}^{12,-}_t).
\end{align*}
\end{itemize}
As a consequence,
\begin{align}
\{\exists  & t\in \R_+,\, N^{11}_t= 0  \mbox{ or }N^{12}_t=0\}\nonumber\\
\subset & \{\exists t\in \R_+,\, \widehat{N}^{11,+}_t=0\mbox{ or }\widehat{N}^{12,+}_t=0\}\cup \{\exists t\in \R_+,\, \widehat{N}^{11,-}_t=0\mbox{ or }\widehat{N}^{12,-}_t=0\}\nonumber\\
 = & \{\exists t\in \R_+,\, \widehat{N}^{11,+}_t=0\mbox{ or }\widehat{N}^{12,+}_t=0\},\label{etape7}
\end{align}since the two processes $(\widehat{N}^{11,+}_t,\widehat{N}^{12,+}_t)_{t\in \R_+}$ and $(\widehat{N}^{11,-}_t,\widehat{N}^{12,-}_t)_{t\in \R_+}$ reach the axes at the same time. The computation of the probability of the event on the r.h.s. of \eqref{etape7} is given by Prop. \ref{prop_couplinga2} and proves the upper bound in \eqref{encadrement}.

\subsection{Fecundity selection model (Model 3) in a small population}

\unitlength=1cm
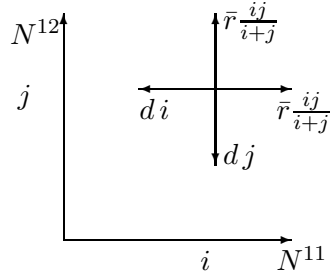
\begin{figure}[ht]
  \begin{center}
\begin{minipage}[b]{.45\textwidth}\centering
    \begin{picture}(4,4)
    \put(1,1){\vector(1,0){3}}
    \put(1,1){\vector(0,1){3}}
    \put(3.8,0.6){$N^{11}$}
    \put(0.3,3.6){$N^{12}$}
    \put(3,3){\vector(1,0){1}}
    \put(3,3){\vector(-1,0){1}}
    \put(3,3){\vector(0,1){1}}
    \put(3,3){\vector(0,-1){1}}
    \put(3.8,2.6){$\bar{r}\frac{ij}{i+j}$}
    \put(2,2.6){$d\,i$}
    \put(3.1,2){$d\,j$}
    \put(3.1,3.8){$\bar{r}\frac{ij}{i+j}$}
    \dottedline(1,3)(3,3)
    \dottedline(3,1)(3,3)
    \put(0.4,2.8){$j$}
    \put(2.8,0.6){$i$}
    \end{picture}
\end{minipage}
  \vspace{-0.3cm}
  \caption{\textit{Evolution of the distylous system $(N^{11}_t,N^{12}_t)_{t\in \R_+}$ in the fecundity selection model.
  }}\label{fig:small-fec}
  \end{center}
\end{figure}

For the fecundity selection model, we will prove that the same criteria as in large population hold for separating the subcritical, critical and supercritical cases:
\begin{proposition}Assume that $(N^{11}_0,N^{12}_0)=(i,j)$.\\
 (i) If $r\leq 2d$, then we have almost sure extinction in finite time.\\
 (ii) If $r>2d$, then we have a positive probability of survival $1-p_{i,j}$ such that:
 \begin{equation}
  \Big(\frac{d}{\bar{r}}\Big)^{i+j} \leq p_{i,j}\leq \Big(\frac{2d}{\bar{r}}\Big)^{\min(i,j)} .\label{etape8} \end{equation}
\end{proposition}

\begin{proof}We begin with the proof of (i). We first assume that $\bar{r}\leq 2d$. When in the state $(i,j)
\in \N^2\setminus \{(0,0)\}$, with $i\leq j$, the total death rate of the population is $d(i+j)$ and the total birth rate is:
\begin{align*}
 2\bar{r}\frac{ij}{i+j}=  \bar{r}\frac{j}{i+j} i+\bar{r}\frac{i}{i+j} j
 = &   \bar{r}\big(\frac{1}{2}+\frac{|j-i|}{2(i+j)}\big)i+\bar{r}\big(\frac{1}{2}-\frac{|j-i|}{2(i+j)}\big)j\\
 =  & \frac{\bar{r}}{2}(i+j)+\frac{r}{2}\frac{|j-i|(i-j)}{i+j}
 \leq   \frac{\bar{r}}{2}(i+j).
\end{align*}Hence, it is possible to dominate the total population size by a continuous time Galton-Watson process with individual birth and death rates $\bar{r}/2$ and $d$. For $\bar{r}\leq 2d$, the latter process is critical or subcritical and extinction is almost sure.\\

We now prove (ii), and assume that $\bar{r}>2d$. Since the birth rates are bounded above by $\bar{r}i$ and $\bar{r}j$, $(N^{11}_t,N^{12}_t)_{t\in \R_+}$ is bounded above by the process $(\widetilde{N}^{11}_t,\widetilde{N}^{12}_t)$ of Section \ref{section:dominationsto}. We turn to the minoration. When $(N^{11}_t,N^{12}_t)=(i,j)$ with $i\leq j$, then the birth rate for the population $\{1,1\}$ is:
$$\bar{r}\frac{ij}{i+j}\geq \frac{\bar{r}}{2}i.$$Hence, as long as $N^{11}_t\leq N^{12}_t$, $N^{11}$ is stochastically lower bounded by a supercritical continuous birth and death process with individual birth and death rates $\bar{r}$ and $d$. Similarly, when $j< i$, $N^{12}$ is stochastically lower bounded by the same process.\\
As a consequence, the total size of the population $(N^{11}_t + N^{12}_t)_{t\in \R_+}$ is stochastically lower bounded by the supercritical continuous time birth and death process introduced above, started at $\min(N^{11}_0,N^{12}_0)$, and which survives with probability $1-(2d/\bar{r})^{\min(N^{11}_0,N^{12}_0)}$ (see Prop. \ref{prop4.4}).
\end{proof}

\begin{remark}
 Notice that the upper and lower bounds in \ref{etape8} are not tight and the couplings used in the proof are not likely to be optimal. For the lower bound, the bound 1 for $i/(i+j)$ and $j/(i+j)$ is rough. For the upper bound, we have only considered the less represented component in the population.\qed
\end{remark}

\subsection{Compatible population without pollen limitation nor inbreeding depression (Model 4)}\label{section:birthdeathprocess}

As a reference for the previous simulations, we provide the results for populations where there is no self-incompatibility. This corresponds to the case where $\forall x\in \R^n,\, \Phi(x)=0$. In this situation, we can forget the phenotype of the individuals. This leads us to consider a population where individuals reproduce with an individual rate $R(N)$ and die with a rate $d$, $N$ being the size of the population. Notice that this corresponds to the process $(\wN_t)_{t\geq 0}$ studied in Proposition \ref{propmomentntilde}. \\

The Markov chain embedded in the continuous time branching process is the chain on $\N$, with transition probabilities $q_{i,j}$ from $i$ to $j$ defined for $i>0$ by:
\begin{align}
& q_{i,i+1}=\frac{R(i)}{R(i)+d},\qquad q_{i,i-1}=\frac{d}{R(i)+d} ,\qquad q_{0,0}=1\quad \mbox{ and else }\quad q_{i,j}=0.
\end{align}
Let $p_i$ be the extinction probability when the initial population is of size $i$.

\begin{proposition}Let $p_1$ be given. For $i\geq 1$,
\begin{align}
p_{i+1}=p_1\left(1+\sum_{j=1}^i \frac{d^j}{R(1)\dots R(j)}\right)-\sum_{j=1}^i \frac{d^j}{R(1)\dots R(j)}.
\end{align}
\end{proposition}

\begin{proof}
Using the strong Markov property at the time of the first event \citep[see \eg][]{baldimazliakpriouret}:
\begin{align}
p_i=\frac{d}{R(i)+d}p_{i-1}+\frac{R(i)}{R(i)+d}p_{i+1},
\end{align}where by convention $p_0=1$. We deduce from this that:
\begin{align}
p_{i+1}-p_i = & \frac{d}{R(i)}(p_i-p_{i-1}) = \frac{d^i}{\prod_{j=1}^i R(j)}(p_1-1).
\end{align}by recursion. The result follows by using $p_{i+1}=p_1+\sum_{j=1}^i(p_{j+1}-p_j)$.
\end{proof}

\begin{example}\label{ex:model4-fin}In the case where the individual reproduction rate is constant $R(i)=\bar{r}$,
\begin{align}
\sum_{j=1}^i \left(\frac{d}{\bar{r}}\right)^j=\frac{d}{\bar{r}-d}\left(1-\Big(\frac{d}{\bar{r}}\Big)^i\right).
\end{align}In this case, $(\widetilde{N}_t)_{t\geq 0}$ is moreover a one-dimensional continuous time Markov branching process \citep[\eg][Chap. III]{athreyaney} and we know that $p_1$ solves $g(s)=s$ where $g(s)$ is the generating function of the offspring distribution:
$$g(s)=\frac{d}{d+\bar{r}}+\frac{\bar{r}}{d+\bar{r}}s^2.$$This gives
\begin{equation}\label{SCsmall}
p_1=\frac{d}{\bar{r}}\qquad\mbox{ and hence }\qquad\forall i\geq 1,\, p_{i}=\left(\frac{d}{\bar{r}}\right)^i,
 \end{equation}which is expected since the branching property holds in this case, contrary to cases where $R(.)$ is not constant and where there is interaction between the individuals.
\end{example}

In a nutshell, we have obtained the following estimates for the $p_{i,j}$'s:
\begin{table}[!ht]
\begin{center}
\begin{tabular}{|l|l|l|}
\hline
Case & Case description & Extinction probabilities (or bounds) when started from $(i,j)$\\
\hline
\multicolumn{3}{|c|}{Wright's model (Model 1)}\\
\hline
(a-b.1) & $\bar{r}\leq d$ & $p_{i,j}=1$\\
(c.1) & $\bar{r}>d$ & $(d/\bar{r})^{i+j}\leq p_{i,j}\leq (d/\bar{r})^i+(d/\bar{r})^j-(d/\bar{r})^{i+j}$\\
\hline
\multicolumn{3}{|c|}{Fecundity selection model (Model 3)}\\
\hline
(a-b.3) & $\bar{r}\leq 2d$ & $p_{i,j}=1$\\
(c.3) & $\bar{r}>2d$ & $(d/\bar{r})^{i+j}\leq p_{i,j}\leq (2d/\bar{r})^{\min(i,j)}$\\
\hline
\multicolumn{3}{|c|}{Self-compatible model without inbreeding depression (Model 4)}\\
\hline
(a-b.4) & $\bar{r}\leq d$ & $p_{i,j}=1$\\
(c.4) & $\bar{r}>d$ & $p_{i,j}=(d/\bar{r})^{i+j}$.\\
\hline
\end{tabular}
\caption{{\small \textit{Summary of the bounds on the extinction probabilities $p_{i,j}$ defined in Section \ref{section:recurrence}. It is seen that the disjunction between sub, super and critical regimes is the same as for large population. In the super-critical case, there remains however always a positive probability of extinction.}}}\label{tab1}
\end{center}
\end{table}

\section{Simulations}\label{section:simulations}

\subsection{Simulation algorithm}\label{algosimu}
The population dynamics described in the previous Sections \ref{inddynadescr} and \ref{sectionreproductionrate} can be simulated with the following algorithm \citep[see][]{fourniermeleard}. Notice that the algorithm that we propose is exact (in the sense that it describes exactly the dynamics described above without approximation scheme). Assume that the population is known at time $t$. Then:
\begin{enumerate}
\item We define the total event rate at the population level by:
\begin{equation}
C_t=\sum_{1\leq u\leq v\leq n} R(\overline{N}_t^{uv},N_t)\overline{N}_t^{uv} + N_t \, d.
\end{equation}
\item The next event time is $t'=t+\tau$ where $\tau$ is an independent random variable that is exponentially distributed with parameter $C_t$.
\item We then draw an independent uniform random variable $\theta$. \\
If $0\leq \theta\leq \sum_{1\leq u\leq v\leq n}R(\overline{N}_t^{uv},N_t)\overline{N}_t^{uv}/C_t$ then a birth happens:
\begin{enumerate}
\item The ovule is of type $\{u,v\}$ with probability $$R(\overline{N}_t^{uv},N_t)\overline{N}_t^{uv}/\sum_{1\leq u\leq v\leq n} R(\overline{N}_t^{uv},N_t)\overline{N}_t^{uv}.$$
\item The pollen is then of type $\{u',v'\}$ with probability $\overline{p}_t^{uv}(u',v')$.
\item The offspring is then of genotype $\{u,u'\}$, $\{u,v'\}$, $\{v,u'\}$ or $\{v,v'\}$ with probability 1/4.
\end{enumerate}
If $\theta> \sum_{1\leq u\leq v\leq n} R(\overline{N}_t^{uv},N_t)\overline{N}_t^{uv}/C_t$ then an individual dies. This individual is drawn uniformly among the living individuals.
\end{enumerate}

\subsection{Simulations performed in the case of a distylous species}\label{simu}
Each simulated curve is obtained as the average on 5000 simulations of paths $(N_t^{11}, N^{12}_t)_{t\in [0,10000]}$. The estimated extinction probability
is obtained as the mean number of extinctions before time 10000.
The simulations were run for varying number of different genotypes as initial conditions when considering Wright's model (Model 1, Section \ref{sectionreproductionrate}), while $N_0^{11}=N_0^{12}=1$ when considering the dependence model (Model 2). \\

In Fig. \ref{fig_table1} the extinction probability for a constant reproduction rate $\overline r = 2$ is shown for an increasing death rate, for both the distylous and the self-compatible case, for different initial numbers of the two genotypes, under the Wright's model. Fig. \ref{fig_table1} shows that in small population, the extinction probabilities are higher for the distylous than for the self-compatible population. As expected, the difference between the distylous and the self-compatible populations is lower when the initial population size increases, as shown in proposition \ref{propODErconstant} for Wright's model. The results shown for Wright's model shows the importance of the absorbing effect that increases the extinction probability for the distylous population. We also see that when $\alpha < +\infty$ in \eqref{functionrN}, the extinction probabilities are higher than in Wright's model, both for distylous and self-compatible populations. This reflects that when pollen limitation is high, it is more difficult to encounter a mate and to produce offspring.

\begin{figure}[!h]
\begin{center}
\begin{tabular}{cc}
(a) & (b) \\
\includegraphics[width=0.4\textwidth,height=0.27\textheight,angle=-90,trim=1cm 1cm 1cm 1cm]{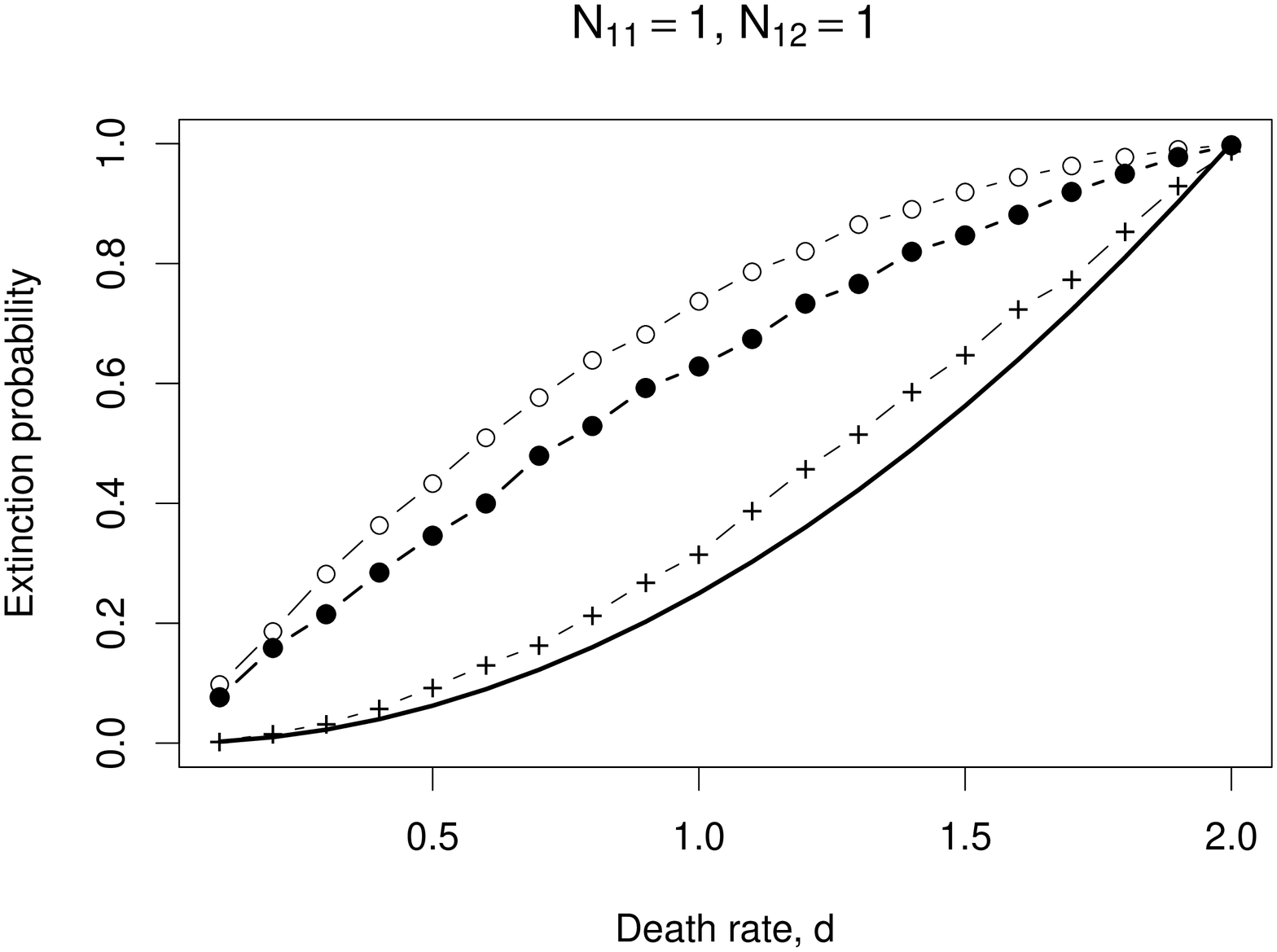}
 &
\includegraphics[width=0.4\textwidth,height=0.27\textheight,angle=-90,trim=1cm 1cm 1cm 1cm]{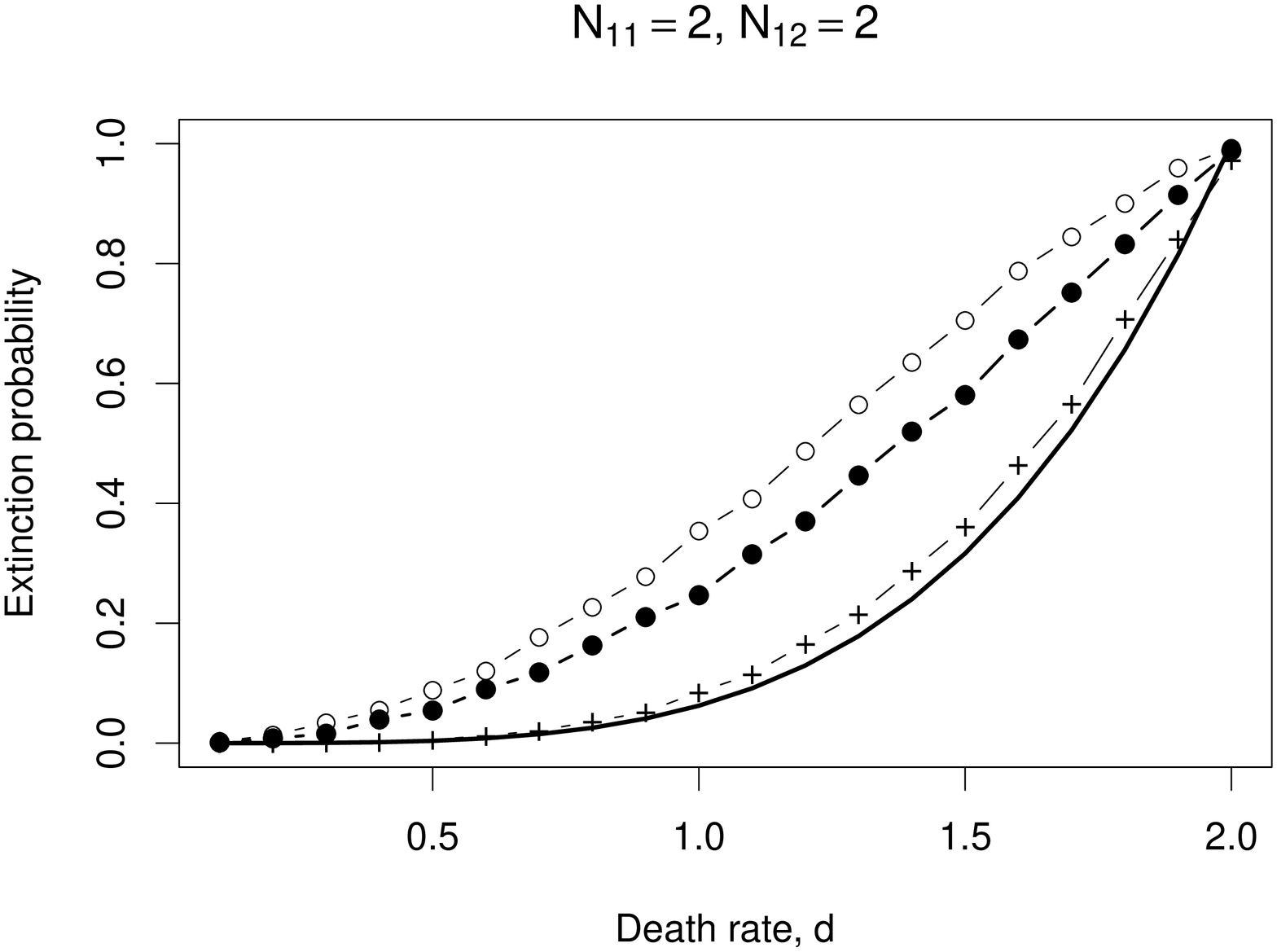}\\
\\
(c) & (d)\\
\includegraphics[width=0.4\textwidth,height=0.27\textheight,angle=-90,trim=1cm 1cm 1cm 1cm]{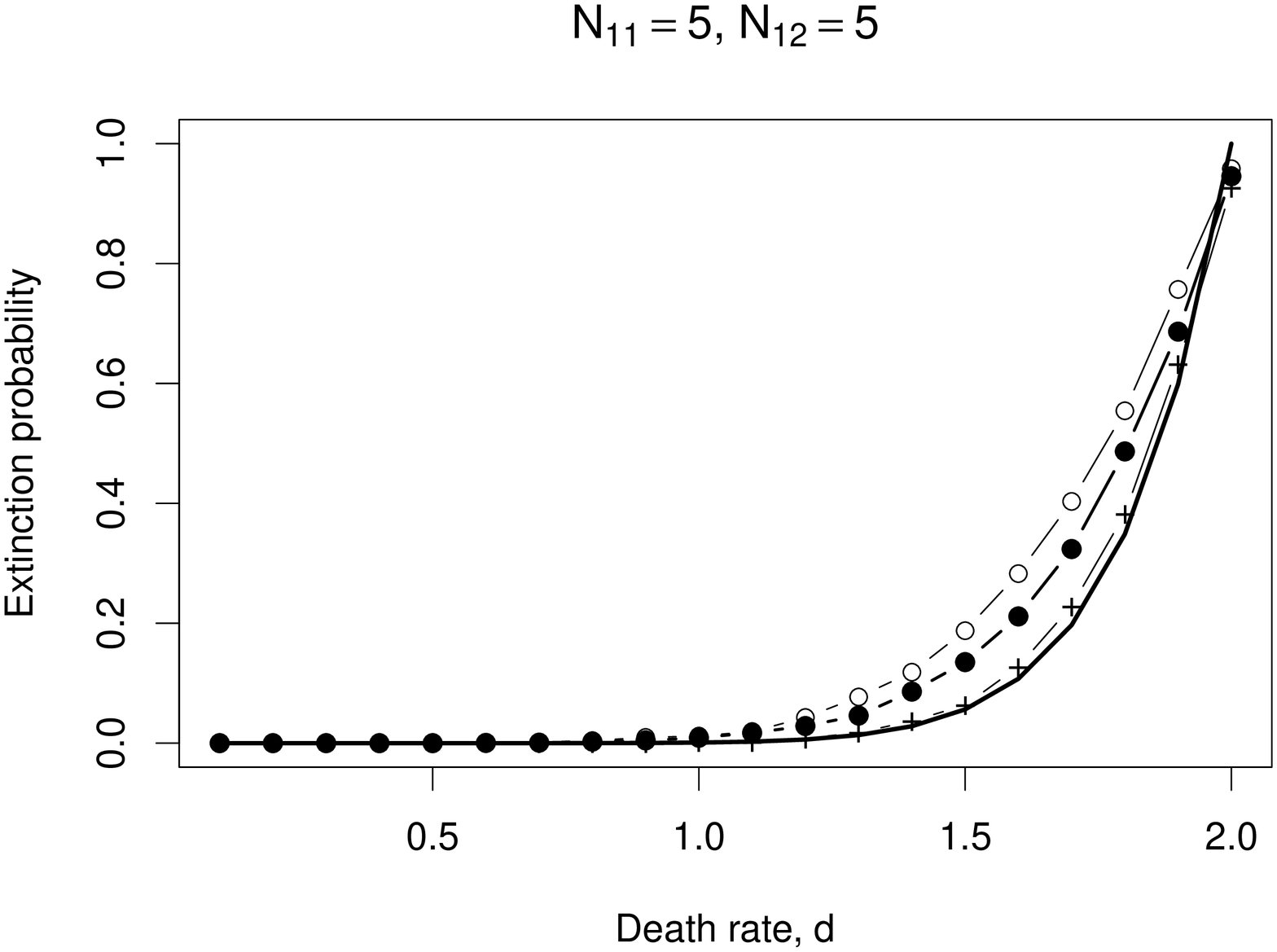}
 &
\includegraphics[width=0.4\textwidth,height=0.27\textheight,angle=-90,trim=1cm 1cm 1cm 1cm]{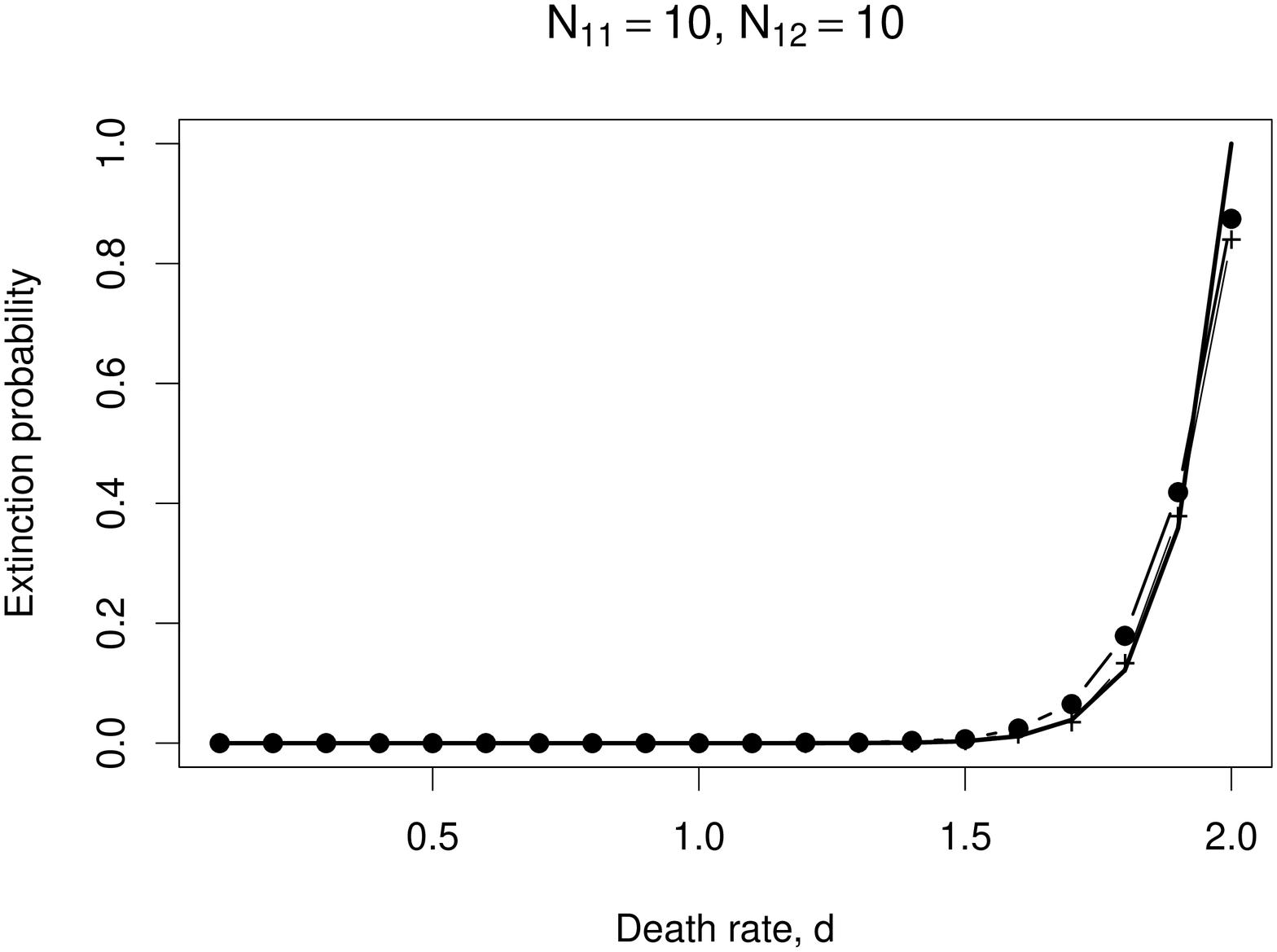}\\
\end{tabular}
\caption{\textit{Comparison of the extinction probabilities of $(N^{11}_t,N_t^{12})_{t\geq 0}$ for a constant reproduction rate $\bar{r}=2$ when the initial number of the different genotypes varies (indicated in the top of the box). Thick line: Extinction probability in the self-compatible case (from equation \ref{SCsmall} ; Crosses: self-compatible case with pollen limitation ($\alpha=\beta=1$) ; Full circles: distylous case with Wright's model ; Circles: distylous case with the dependence model ($\alpha=\beta=1$).}} \label{fig_table1}
\end{center}
\end{figure}

In Fig. \ref{fig_ratio}, we estimated the ratio $\rho$ of the extinction probability in the self-compatible case on the extinction probability in the SI case, which are obtained
from our simulations results. The Fig. \ref{fig_ratio}(a) shows that the ratio $\rho$ gets smaller when the initial population size is higher, which means that the extinction probabilities for the self-compatible population decreases more rapidly when $N$ increases than for the distylous population. This effect is even worse when the initial population is asymmetric, that is when a given genotype is more frequent than the other (compare filled with empty symbols). Once again, the differences between distylous and self-compatible populations shown in Fig. \ref{fig_ratio}(a) are only due to the absorbing effects. Fig. \ref{fig_ratio}(b) shows the ratio $\rho$ for different initial conditions under the dependence model with $\alpha=1$ and $\beta=100$. Fig. \ref{fig_ratio}(c) shows similar results but for different strength of pollen limitation ($\alpha < +\infty$) with $N_0^{11}=N_0^{12}=1$. It is remarkable in this figure that when the strength of the pollen limitation is high ($\alpha$ is small and $\beta$ is large) then the ratio $\rho$ is higher, which means that the higher the pollen limitation, the lower the difference between distylous and self-compatible populations.

\begin{figure}[!h]
\begin{center}
\begin{tabular}{cc}
(a) & (b)  \\
\includegraphics[width=0.4\textwidth,height=0.27\textheight,angle=-90,trim=1cm 1cm 1cm 1cm]{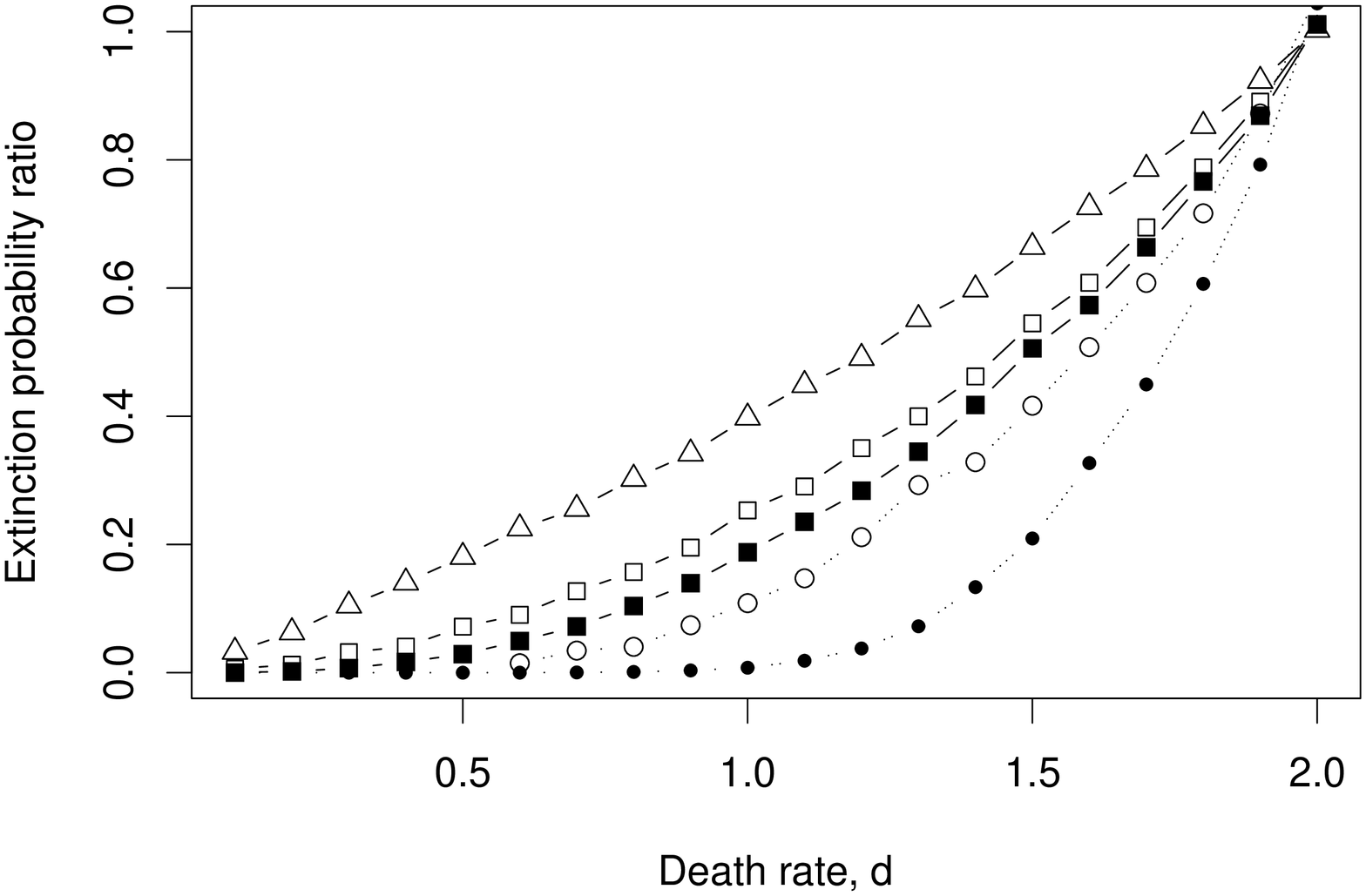}
 &
\includegraphics[width=0.4\textwidth,height=0.27\textheight,angle=-90,trim=1cm 1cm 1cm 1cm]{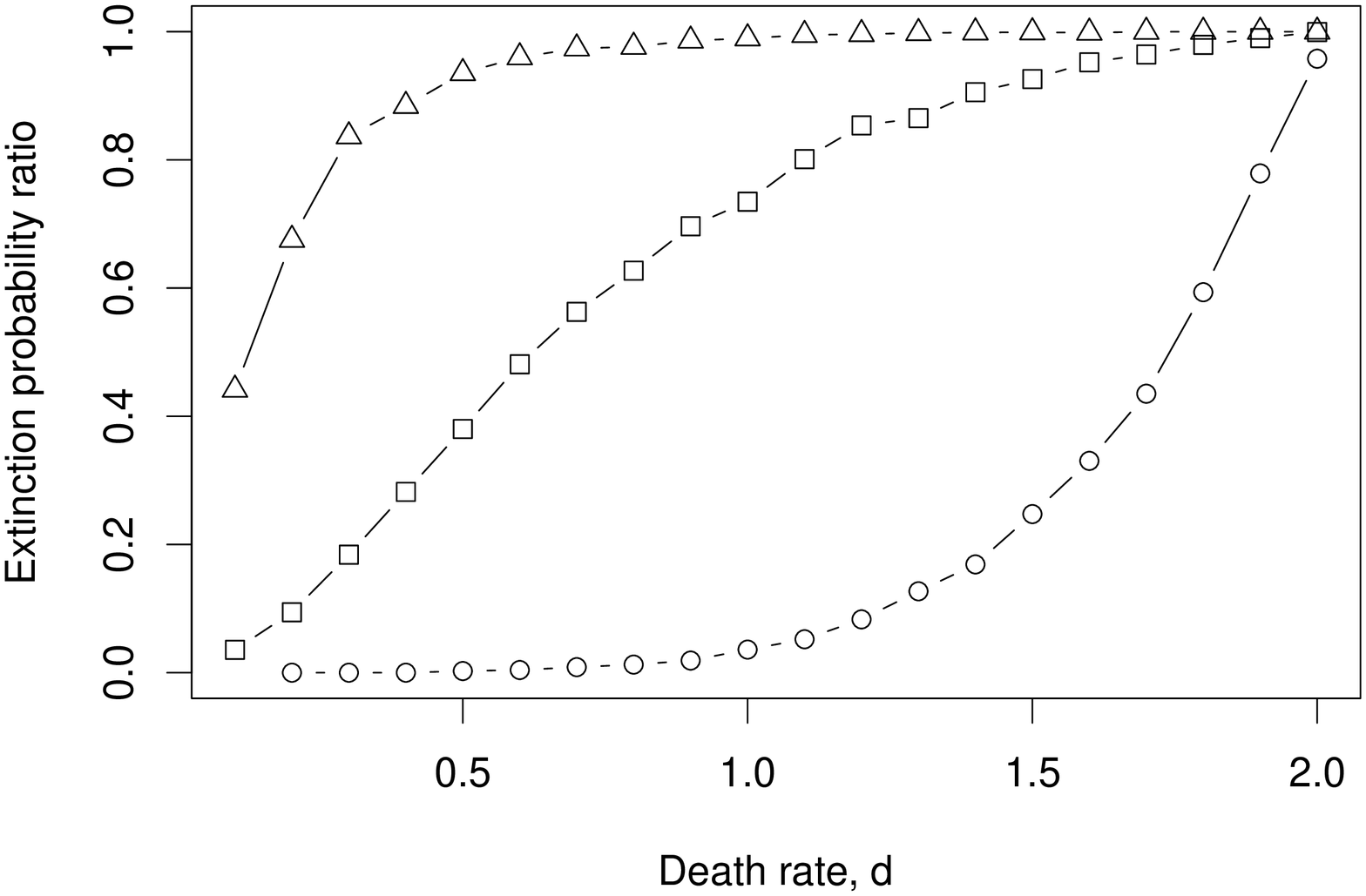}
 \\
\multicolumn{2}{
c}{ (c)} \\
\multicolumn{2}{c}{\includegraphics[width=0.4\textwidth,height=0.27\textheight,angle=-90,trim=1cm 1cm 1cm 1cm]{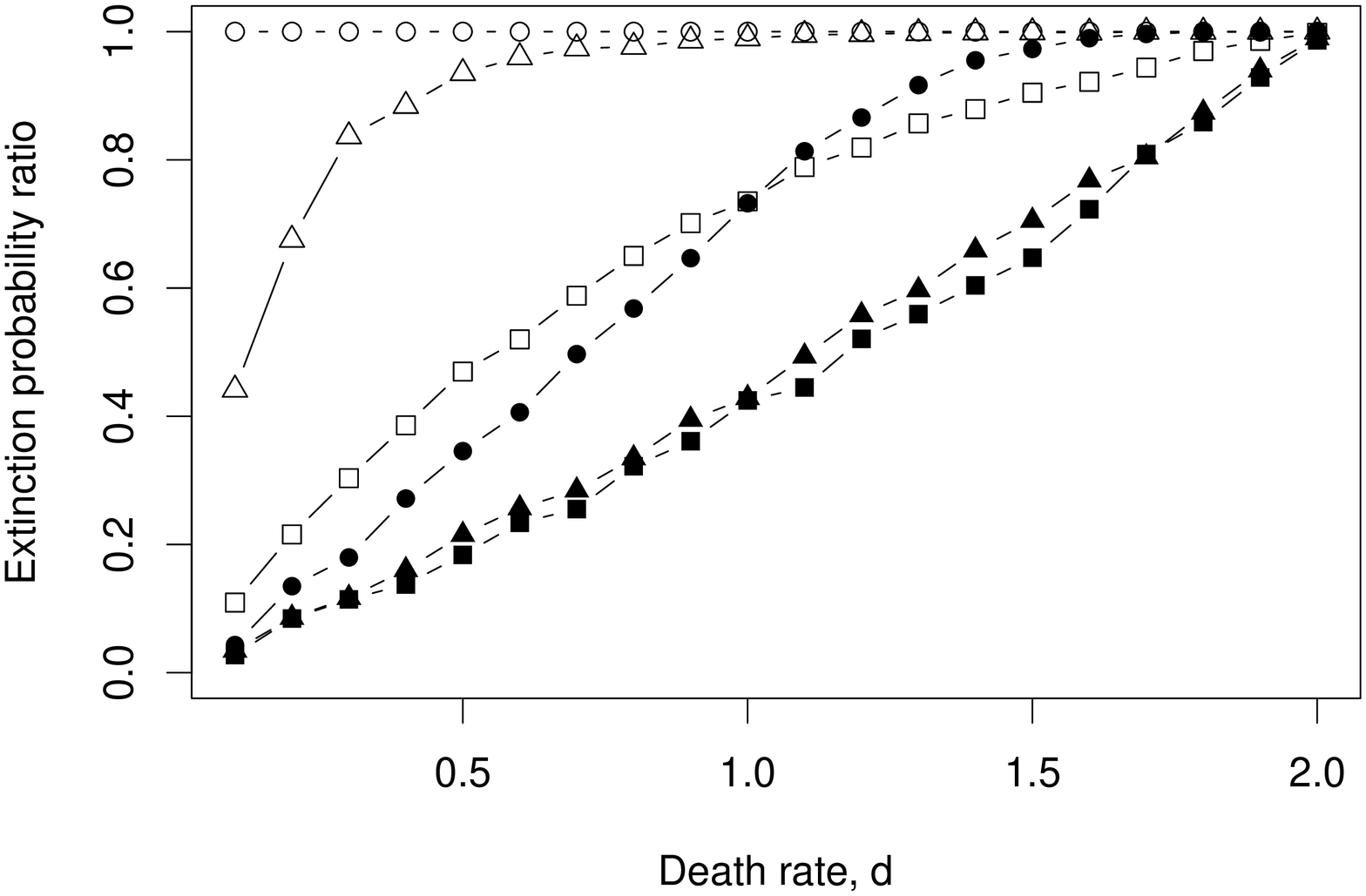}}
\end{tabular}
\caption{\textit{Ratio of extinction probabilities in self-compatible cases over distylous case with $\overline{r}=2$, varying initial conditions and varying $\alpha$ and $\beta$. (a): Wright's model, $N^{11}_0=N_0^{12}=1$ (triangles); $N^{11}_0=N_0^{12}=2$ (empty squares); $N^{11}_0=N_0^{12}=5$ (empty circles); $N^{11}_0=3$ and $N_0^{12}=1$ (filled squares); $N^{11}_0=9$ and $N_0^{12}=1$ (filled circles). (b): Dependence model with $\alpha=1$ and $\beta=100$ and varying initial conditions: $N^{11}_0=N_0^{12}=1$ (triangles), $N^{11}_0=N_0^{12}=2$ (squares),$N^{11}_0=N_0^{12}=5$ (circles), $N^{11}_0=N_0^{12}=10$ (crosses). (c): Dependence model with $N^{11}_0=N_0^{12}=1$ and $\beta=1$ (filled symbols), $\beta=100$ (empty symbols),  $\alpha=0.1$ (circles), $\alpha=1$ (triangles), $\alpha=2$ (squares).}}\label{fig_ratio}
\end{center}
\end{figure}

\section{Discussion}

\paragraph{Impact of pollen limitation versus demographic stochasticity and boundary effects}
Thanks to the three relationships between the compatible population size and the reproductive rates we assumed (see \ref{sectionreproductionrate} and Fig. \ref{fig:billiardtranSSI}), we are able to disentangle the relative effects of demographic stochasticity and pollen limitation on the fate of the populations. Indeed, when Wright's model is assumed, there is no pollen limitation. In other words all individuals of a given mating types receive enough pollen to fertilize all their ovules, as long as at least one compatible individual is present in the population. This assumption introduces a discontinuity in the individual rate of seeds production. On the other extreme, the fecundity selection model has a rate of seeds production that vanishes continuously at the boundaries. 
Finally, we investigated the combined effect of pollen limitation and demographic stochasticity thanks to the dependence model.\\

We have first studied large populations. We exhibited three different regimes, depending on the relationships between the birth and death parameters. Under the subcritical regime, the population goes extinct when $t\rightarrow +\infty$, while its size tends to infinity under the supercritical regime.
As expected, 
when there is no pollen limitation, a distylous population behaves like a self-compatible population. When there exists an equilibrium, the population size is $n^*$ in both cases. To the contrary, under the fecundity selection model, the size at equilibrium is $2n^*$ showing that pollen limitation can have a large effect on the dynamics of distylous populations, even in large ones. In the dependence model, a saddle point appears that makes condition for the maintenance of the population more scarce than in Wright's model and favors populations with symmetry in the initial conditions.
\par In small populations, the separation into subcritical, critical or supercritical regimes is the same as for large populations. In supercritical cases, we are able to find lower and upper bounds of the probability of extinction of a distylous population under the Wright's and fecundity selection models. Again, the main difference between these models relies on the fact the critical condition for the survival of the population is twofold higher in the fecundity selection model. \\
Our simulations show that when pollen limitation is low ($\alpha=1$ and $\beta=1$, see Fig. \ref{fig_table1}), the increase in the probability of extinction is small relative to the impact of demographic stochasticity. When we compare the extinction
probabilities between a distylous and a self-compatible population by the measure $\rho$, we see that the higher the pollen limitation (when $\alpha$ decreases), the lower the difference between them since $\rho$ increases (see Fig. \ref{fig_ratio}).
Those results have important ecological and evolutionary consequences since they suggest pollen limitation plays a minor role in small populations relative to demographic stochasticity, and especially the stochastic loss of one of the two mating types. When pollen limitation is very large (when $\alpha$ is low), small populations of self-compatible and self-incompatible species tend to behave similarly. In short, our results suggest that pollen limitation plays a major role in large population only.

\paragraph{Mate Finding Allee effect, inbreeding depression and the evolution of SI}
One of the most intriguing and long-standing problem in evolutionary biology resides in the existence and maintenance of SI. Indeed, the establishment probability is lower in the case of SI compared to self-compatible species since a single individual is sufficient to colonize empty spaces (the Baker's law, \citet{pannellbarrett98}). Furthermore, as shown here, the extinction probability is higher for SI than for self-compatible species in absence of inbreeding depression. It has been also shown that in infinite populations, the conditions for the invasion of SI populations by a self-compatible mutant are less stringent when there is strong pollen limitation \citep[\eg][]{porcherlande_evolution}. The process generally invoked to explain the existence and maintenance of SI is inbreeding depression: SI can be advantageous relatively to self-compatibility when the cost of inbreeding depression caused by self-fertilization is high. Inbreeding depression can play a role in two ways: by preventing the invasion of self-compatible mutant \citep{porcherlande_evolution}, or by increasing the extinction rate of obligate selfing species relatively to SI species \citep{lynchetal95}, what is suggested in data since self-compatible species are most often localised at the leaves of phylogenies, especially in plant families where distyly is present (\eg in Amsinckia \citet{schoenetal97}, in Narcissus \citet{perezbarralesetal06} and in Psychotria \citet{sakaiwright08}).  Here, we proposed a very simple model taking into account inbreeding depression caused by self-fertilization, to investigate if there are conditions under which the extinction probabilities are higher for SC populations than for SI populations. We found that there are some conditions in large populations where the size at equilibrium is lower in SI populations than in self-compatible populations: under strong inbreeding depression, SI populations may be less sensitive to extinction.

\paragraph{What about more complex dominance networks?}
Although we developed a general model for the population dynamics of SSI species, we mainly investigated  the dynamics of a distylous species. Distyly is however the case where the impact of the existence of mating types is the highest since only two mating types exist. When the number of mating types increases, the proportion of compatible individuals also increases. It would be interesting to investigate the impact of the dominance relationships among S-alleles on the extinction of population, to check if there are dominance interactions patterns that are less sensitive to extinction than others, as highlighted by \citet{kirchnerrobertcolas}: they showed that populations where all S-alleles are codominant have a higher extinction rate than populations where a linear hierarchy of dominance exist between S alleles (the DOM model, see \eg \citet{billiardcastricvekemans}). Our model could be used to investigate the effect of these dominance interactions more precisely.

\appendix

\section{SDEs and Proofs of the Section \ref{sectionODE} Propositions}

\subsection{SDEs}
Following \citet{fourniermeleard}, we present a SDE describing the evolution of $(Z_t)_{t\in \R_+}$.

\begin{definition}\label{defsde}
 Let $Q(ds,dg',d\theta)$ be a Poisson point measure on $\R_+\times E\times \R_+$ with intensity $q(ds,dg',d\theta)=ds\,dn(g')\,d\theta$, where $ds$ and $d\theta$ are Lebesgue measures on $\R_+$ and where $dn(g')$ is the counting measure on $E$. To each atom from $Q(ds,dg',d\theta)$ are hence associated a time of possible event $s$, the genotype $g'$ that either appears or dies and an auxiliary variable $\theta$ that decides what happens (with a role similar to the variable $\theta$ in Point 3 of the algorithm of Section \ref{algosimu}):
 \begin{multline}
 Z_t(dg)=Z_0(dg) +\int_0^t \int_{E\times \R_+} \delta_{g'}(dg) \Big(\ind_{\theta\leq r^{g'}(Z_{s_-})}\\
 -  \ind_{r^{g'}(Z_{s_-})<\theta\leq r^{g'}(Z_{s_-})+d\times N^{g'}_{s_-}}\Big) Q(ds,dg',d\theta),\label{sde}
 \end{multline}and where $r^{g'}(Z_t)$ has been defined in (\ref{tauxreprod(uv)}).\qed
\end{definition}

Existence and uniqueness of a solution of (\ref{sde}) are stated in the next proposition. Moreover, it is possible, for a given test function $f$, to derive from (\ref{sde}) equations for the evolution of $(\langle Z_t,f\rangle)_{t\in \R_+}$.

\begin{proposition}\label{propsde}
(i) If $\E(N_0)<+\infty$, then there exists a unique solution to SDE (\ref{sde}).\\
(ii) If additionally $\E(N_0^2)<+\infty$ then for any bounded test function $f$ on $\E$,
\begin{align}
\langle Z_t,f\rangle= & \langle Z_0,f\rangle +\int_0^t  \sum_{\{u,v\}\in E}\left(r^{uv}(Z_s) -d\times N^{uv}_s\right)f(\{u,v\}) ds+M^f_t
\end{align}where $(M^f_t)_{t\in \R}$ is a square integrable martingale starting from 0 with quadratic variation:
\begin{align}
\langle M^f\rangle_t= & \int_0^t  \sum_{\{u,v\}\in E}  \big(r^{uv}(Z_s) + d\times N^{uv}_s\big)f^2(u,v) ds.
\end{align}
\end{proposition}

\begin{proof}The rate $r^{uv}$ defined in (\ref{tauxreprod(uv)}) is bounded by a linear function in $N_t$. The moment condition in (i) allows us to prove that there is no explosion. The proofs then follows the ones developed in \citet[Th. 3.1 and Prop. 3.4]{fourniermeleard} for a model of plant with asexual reproduction.
\end{proof}

\subsection{Couplings in Wright's model for the proofs of Section \ref{section:smallpop}}\label{section:SDEdetaillees}

In this section, we give the expression of the processes $(N^{11}_t,N^{12}_t)_{t\in \R_+}$, $(\widetilde{N}^{11}_t,\widetilde{N}^{12}_t)_{t\in \R_+}$ and $(\widehat{N}^{11}_t,\widehat{N}^{12}_t)_{t\in \R_+}$ that appear in Section \ref{section:smallpop}.\\

Let us rewrite Equations \eqref{equationsdistyles} thanks to \eqref{sde}:
\begin{align*}
 N^{11}_t=N^{11}_0 +\int_0^t \int_{E\times \R_+} & \ind_{g'=\{1,1\}}\Big(\ind_{N^{11}_{s_-}>0 ; N^{12}_{s_-}>0}\ind_{\theta\leq \frac{\bar{r}(N^{11}_{s_-}+N^{12}_{s_-})}{2}}\\
  & -\ind_{\frac{\bar{r}(N^{11}_{s_-}+N^{12}_{s_-})}{2}<\theta\leq \frac{\bar{r}(N^{11}_{s_-}+N^{12}_{s_-})}{2}+d\times N^{11}_{s_-}}\Big) Q(ds,dg',d\theta)
\end{align*}and similarly for $(N^{12}_t)_{t\in \R_+}$.\\

For the stochastic domination, the process $(\widetilde{N}^{11}_t,\widetilde{N}^{12}_t)_{t\in \R_+}$ introduced in \eqref{tauxdominants} can be rewritten as:
\begin{align*}
 \widetilde{N}^{11}_t=N^{11}_0 +\int_0^t \int_{E\times \R_+} & \ind_{g'=\{1,1\}}\Big(\ind_{\theta\leq \frac{\bar{r}(\widetilde{N}^{11}_{s_-}+\widetilde{N}^{12}_{s_-})}{2}}\\
 & -\ind_{\frac{\bar{r}(\widetilde{N}^{11}_{s_-}+\widetilde{N}^{12}_{s_-})}{2}<\theta\leq \frac{\bar{r}(\widetilde{N}^{11}_{s_-}+\widetilde{N}^{12}_{s_-})}{2}+d\times \widetilde{N}^{11}_{s_-}}\Big) Q(ds,dg',d\theta)
\end{align*}and similarly for $(\widetilde{N}^{12}_t)_{t\in \R_+}$. If we start at a point where $(N^{11}_t,N^{12}_t)=(\widetilde{N}^{11}_t,\widetilde{N}^{12}_t)$ then it is clear that these processes have the same births and deaths. If one of the components is null, then there is no birth in $(N^{11}_t,N^{12}_t)_{t\in \R_+}$ while there is still births in $(\widetilde{N}^{11}_t,\widetilde{N}^{12}_t)_{t\in \R_+}$.\\

For the stochastic minoration, we have used in Section
\ref{section:stochasticminoration} the processes $(\widehat{N}^{11,+}_t,\widehat{N}^{12,+}_t)_{t\in \R_+}$ and $(\widehat{N}^{11,-}_t,\widehat{N}^{12,-}_t)_{t\in \R_+}$. Let us give an SDE to describe their dynamics.
\begin{align}
\widehat{N}^{11,+}_t= & \min(N^{11}_0,N^{12}_0) +\int_0^t \int_{E\times \R_+} \Big[\ind_{(N^{11}_{s_-}-N^{12}_{s_-})(\widehat{N}^{11,+}_{s_-}-\widehat{N}^{12,+}_{s_-})>0}\times \nonumber\\ 
 \Big\{ & \ind_{N^{11}_{s_-}<N^{12}_{s_-}}\Big(\ind_{N^{11}_{s}-N^{11}_{s_-}>0}-\ind_{N^{11}_{s}-N^{11}_{s_-}<0}\nonumber\\
  & \hspace{2cm}-\ind_{g'=\{1,2\}}
 \ind_{\frac{(\bar{r}+d)(N^{11}_{s_-}+N^{12}_{s_-})}{2}<\theta\leq\frac{\bar{r}(N^{11}_{s_-}+N^{12}_{s_-})}{2}+d N^{12}_{s_-}}\Big)\nonumber\\ 
& + \ind_{N^{11}_{s_-}>N^{12}_{s_-}} \Big(\ind_{N^{11}_{s}-N^{11}_{s_-}>0}
-\ind_{g'=\{1,1\}}\ind_{\frac{\bar{r}(N^{11}_{s_-}+N^{12}_{s_-})}{2}<\theta\leq\frac{(\bar{r}+d)(N^{11}_{s_-}+N^{12}_{s_-})}{2}}\Big)\Big\}\nonumber\\ 
 + & \ind_{N^{11}_{s_-}=N^{12}_{s_-}} \Big(\ind_{N^{11}_{s}-N^{11}_{s_-}>0}-\ind_{N^{11}_{s}-N^{11}_{s_-}<0}\Big)\nonumber \\ 
+ &\ind_{(N^{11}_{s_-}-N^{12}_{s_-})(\widehat{N}^{11,+}_{s_-}-\widehat{N}^{12,+}_{s_-})<0}\times \nonumber\\ 
& \Big\{ \ind_{N^{11}_{s_-}<N^{12}_{s_-}}\Big(\ind_{N^{12}_{s}-N^{12}_{s_-}>0}-\ind_{g'=\{1,2\}}
 \ind_{\frac{\bar{r}(N^{11}_{s_-}+N^{12}_{s_-})}{2}<\theta\leq\frac{(\bar{r}+d)(N^{11}_{s_-}+N^{12}_{s_-})}{2}}\Big)\nonumber\\ 
& + \ind_{N^{11}_{s_-}>N^{12}_{s_-}} \Big(\ind_{N^{12}_{s}-N^{12}_{s_-}>0}
-\ind_{N^{12}_{s}-N^{12}_{s_-}<0}\nonumber\\
 & \hspace{2cm}-\ind_{g'=\{1,1\}}\ind_{\frac{(\bar{r}+d)(N^{11}_{s_-}+N^{12}_{s_-})}{2}<\theta\leq\frac{\bar{r}(N^{11}_{s_-}+N^{12}_{s_-})}{2}+d N^{11}_{s_-}}\Big)\Big\}\Big]\nonumber\\
& \hspace{7.5cm} Q(ds,dg',d\theta)\label{SDEN+}
\end{align}
Similar SDEs can be written for $(\widehat{N}^{12,+}_t)_{t\in \R_+}$, $(\widehat{N}^{11,-}_t)_{t\in \R_+}$ and $(\widehat{N}^{12,-}_t)_{t\in \R_+}$.
Notice that the process $(N^{11}_t,N^{12}_t)_{t\in \R_+}$ appears in the definition of the auxiliary processes because of the coupling.\\
The behavior of $(\widehat{N}^{11,+}_t)_{t\in \R_+}$ depend on whether $(N^{11}_t,N^{12}_t)_{t\in \R_+}$ and $(\widehat{N}^{11,+}_t,\widehat{N}^{12,+}_t)_{t\in \R_+}$ belong to the same octant and second on which octant it is. These processes belong to the same octant at time $t$ if $(N^{11}_t-N^{12}_t)(\widehat{N}^{11,+}_t-\widehat{N}^{12,+}_t)>0$. In this case (3 first lines of \eqref{SDEN+}):
\begin{itemize}
\item Births of individuals $\{1,1\}$ in the auxiliary process occur as soon as an individual $\{1,1\}$ is born in the original process $(N^{11}_t,N^{12}_t)_{t\in \R_+}$,
\item If we are in the octant $\{i<j\}$, then every death of individuals $\{1,1\}$ for the original process entails a death of an individual $\{1,1\}$ for $(\widehat{N}^{11,+}_t,\widehat{N}^{12,+}_t)_{t\in \R_+}$. Additionally, some deaths of individuals $\{1,2\}$ for the original process are changed into deaths of individuals $\{1,1\}$. This happens (second line of \eqref{SDEN+}) when:
$$\frac{(\bar{r}+d)(N^{11}_{s_-}+N^{12}_{s_-})}{2}<\theta\leq \frac{\bar{r}(N^{11}_{s_-}+N^{12}_{s_-})}{2}+dN^{12}_{s_-}$$which corresponds to a rate $$\frac{\bar{r}(N^{11}_{s_-}+N^{12}_{s_-})}{2}+dN^{12}_{s_-}-\frac{(\bar{r}+d)(N^{11}_{s_-}+N^{12}_{s_-})}{2}=\frac{d(N^{12}_{s_-}-N^{12}_{s_-})}{2}.$$Hence individuals $\{1,1\}$ in $\widehat{N}^{11,+}$ die with rate:
$$\frac{d(N^{12}_{s_-}-N^{11}_{s_-})}{2}+d \ N^{11}_{s_-}=\frac{d(N^{11}_{s_-}+N^{12}_{s_-})}{2}=\frac{d(\widehat{N}^{11,+}_{s_-}+\widehat{N}^{12,+}_{s_-})}{2}.$$
In this case, the process $(\widehat{N}^{11,+}_t,\widehat{N}^{12,+}_t)_{t\in \R_+}$ has the rates announced in Section \ref{section:stochasticminoration}.
\item If we are in the octant $\{i>j\}$, then deaths of individuals $\{1,1\}$ for $(\widehat{N}^{11,+}_t,\widehat{N}^{12,+}_t)_{t\in \R_+}$ originates from the deaths of the original process such that:
    $$\frac{\bar{r}(N^{11}_{s_-}+N^{12}_{s_-})}{2}<\theta\leq\frac{(\bar{r}+d)(N^{11}_{s_-}+N^{12}_{s_-})}{2},$$
    The death rate for the population $\{1,1\}$ is thus $d(N^{11}_{s_-}+N^{12}_{s_-})/2$.
\end{itemize}
When the original process $(N^{11}_t,N^{12}_t)_{t\in \R_+}$ is on the line $\{i=j\}$, the death rates are $d\ N^{11}_t=d\ N^{12}_t=d(N^{11}_t+N^{12}_t)/2$ and no correction is needed for $(\widehat{N}^{11,+}_t,\widehat{N}^{12,+}_t)_{t\in \R_+}$.\\
Finally, when the original process $(N^{11}_t,N^{12}_t)_{t\in \R_+}$ and the auxiliary process $(\widehat{N}^{11,+}_t,\widehat{N}^{12,+}_t)_{t\in \R_+}$ do not belong to the same octant (lines 7-9 of \eqref{SDEN+}), the auxiliary process $(\widehat{N}^{11,-}_t,\widehat{N}^{12,-}_t)_{t\in \R_+}$ is constructed as above and the process $(\widehat{N}^{11,+}_t,\widehat{N}^{12,+}_t)_{t\in \R_+}$ is deduced by symmetry.

\subsection{Large population limits and proof of Proposition \ref{proprenormgdepop2}}

\begin{proof}[Sketch of Proof of Prop. \ref{proprenormgdepop2}]
\par We begin with (ii) by assuming existence of a solution to:
\begin{equation}
\langle \xi_t,f\rangle=  \langle \xi_0,f\rangle +\int_0^t  \left(\sum_{\{u,v\}\in E}r^{uv}(\xi_s) f(\{u,v\})-d\, \langle \xi_s,f\rangle\right) ds,\label{limitedeterministe}
\end{equation}where $f$ is a bounded test function on $E$. Since all finite measures of $\mathcal{M}_F(E)$ have the form (\ref{defxi}), it remains to prove that the functions $n^{u'v'}$ satisfy (\ref{ODExi}). Let us choose $f(\{u,v\})=\ind_{\{u',v'\}}(\{u,v\})$ and integrate $f$ with respect to $\xi_t$. Since $\langle \xi_t,\ind_{\{u',v'\}}\rangle=n^{u'v'}_t$, (\ref{limitedeterministe}) gives:
\begin{equation}
n^{u'v'}_t=n^{u'v'}_0+\int_0^t\big( r^{u'v'}(\xi_s) -d \, n^{u'v'}_s\big) ds.\label{etape6}
\end{equation}
Notice that for any finite measure $\xi$ on $E$ of the form (\ref{defxi}), $r^{uv}(\xi)$ is the sum of terms of the form $\bar{r}n^{uu'}p^{uu'}(v,v')$ for $u,u',v,v'$ in $\lbrac 1, n\rbrac$. In Wright's model, there may be discontinuities when several of the possible genotypes reach a size 0. Else, the function $r^{uv}(.)$ that we consider is locally Lipschitz continuous as products of locally Lipschitz continuous functions. Since $\xi$ is continuous and since $r^{u'v'}(.)$ is continuous, the integrand in (\ref{etape6}) is continuous, which entails that $n^{u'v'}$ is of class $\Co^1$ and hence of class $\Co^\infty$ by direct recursion. Taking the derivative with respect to time gives (\ref{ODExi}) and achieves the proof. Moreover, the conditions on $r^{uv}$ imply by the Cauchy-Lipschitz theorem that there exists a unique solution to (\ref{ODExi}) and hence to (\ref{limitedeterministe}). Let us now prove existence, which is a consequence of (i).\\

First, let us notice that for given $K\in \N^*$ and real test function $f$ on $E$, the process $(Z^{(K)}_t)_{t\in \R_+}$ satisfies the following evolution equation:
\begin{align}
\langle Z^{(K)},f\rangle= & \langle Z^{(K)}_0,f\rangle +\int_0^t  \left(\sum_{\{u,v\}\in E}r^{uv}(Z^{(K)}_s)f(\{u,v\}) -d\,\langle Z^{(K)}_s,f\rangle\right) ds+M^{(K),f}_t\label{pbmK}
\end{align}where $(M^{(K),f}_t)_{t\in \R}$ is a square integrable martingale starting from 0 with quadratic variation:
\begin{align}
\langle M^{(K),f}\rangle_t= & \frac{1}{K}\int_0^t  \left(\sum_{\{u,v\}\in E}  r^{uv}(Z^{(K)}_s) f^2(\{u,v\})+ d \, \langle Z^{(K)}_s,f^2\rangle\right) ds.
\end{align}
Heuristically, as the quadratic variation of the martingale is of order $1/K$, the stochastic part of the process will disappear in the limit. Since the jumps of $Z^{(K)}$ are of order $1/K$, the limiting values of $(Z^{(K)})_{K\in \N^*}$ are necessarily continuous. Moreover, (\ref{momentordre2}) implies that every one-dimensional marginal has a finite mass. The limiting values hence belong to $\Co(\R_+,\mathcal{M}_F(E))$.
\par The proof of (i) separates classically in two steps. We establish tightness of the laws of $(Z^{(K)})_{K\in \N^*}$, which implies that this family of probability measures is relatively compact \citep[\eg][p.104]{ethierkurtz}. Then, we establish that every limiting value solves (\ref{limitedeterministe}) which has a unique solution.
\par For any $T>0$, the tightness on $\mathbb{D}([0,T],\mathcal{M}_F(E))$ is obtained by using a criterion due to \citet{roelly}, and since $E$ is finite, the problem amounts to prove the tightness on $\mathbb{D}([0,T],\R)$ of the sequence $(\langle Z^{(K)},f\rangle)_{K\in \N^*}$ for bounded test functions $f$. Given the local Lipschitz continuity of $r^{uv}(.)$, given that $|r^{uv}(Z)|\leq C\bar{r} N_t$ and given the moment estimates (\ref{momentordre2}) this is a classical computation which uses Aldous and Rebolledo criteria \citep[\eg][]{joffemetivier}. See \eg \citet[Proof of Th. 5.3]{fourniermeleard}.
\par Taking the limit in (\ref{pbmK}) thanks to (\ref{momentordre2}) again allows us to identify the adherence values of $(Z^{(K)})_{K\in \N^*}$ as the solution of (\ref{limitedeterministe}). This provides existence of a solution to (\ref{limitedeterministe}) and since we have established uniqueness, there is a unique limiting value to which the sequence converges.
\end{proof}

\noindent \textbf{Acknowledgements} This work received financial supports from the \textit{Chaire Mod\'{e}lisation Math\'{e}matique et Biodiversit\'{e}} of Ecole Polytechnique-Museum d'Histoire Naturelle and ANR MANEGE. The authors thank Kilian Raschel for suggesting the process $(\widehat{N}^{11}, \widehat{N}^{12})$.

{\footnotesize
\providecommand{\noopsort}[1]{}\providecommand{\noopsort}[1]{}\providecommand{\noopsort}[1]{}\providecommand{\noopsort}[1]{}

}
\end{document}